\documentclass[a4paper,12pt]{article}%
\usepackage{amssymb}
\usepackage{amsfonts}
\usepackage{amsmath}
\usepackage[nohead]{geometry}
\usepackage[singlespacing]{setspace}
\usepackage[bottom]{footmisc}
\usepackage{indentfirst}
\usepackage{endnotes}
\usepackage{graphicx}%
\usepackage{rotating}
\usepackage{amsthm}
\usepackage{booktabs}
\RequirePackage[OT1]{fontenc} \RequirePackage{amsthm,amsmath}
\RequirePackage[numbers]{natbib}
\RequirePackage[colorlinks,citecolor=blue,urlcolor=blue]{hyperref}
\RequirePackage{hypernat} \makeatletter
\newcommand{\field}[1]{\mathbb{#1}}

\newcommand{\E}{\field{E}}
\def\Zhang{\textcolor{red}}

\theoremstyle{example} \theoremstyle{remark} \theoremstyle{lemma}
\theoremstyle{definition} \theoremstyle{corol}
\theoremstyle{proposition} \theoremstyle{condition}
\theoremstyle{assumption}
\newtheorem{assumption}{\n{Assumption}}[section]
\newtheorem{theorem}{\n{Theorem}}[section]
\newtheorem{example}{\n{Example}}[section]
\newtheorem{remark}{\n{Remark}}[section]
\newtheorem{lemma}{\n{Lemma}}[section]

\newtheorem{proposition}{\n{Proposition}}[section]
\newcommand{\rmnum}[1]{\romannumeral #1}
\newcommand{\Rmnum}[1]{\expandafter\romannumeral #1}

\font\n=cmcsc12

\def\var{{\mbox{var}}}

  \textwidth = 422pt
\begin{document}
\title{\Large Testing High Dimensional Mean Under Sparsity}
\author{\large Xianyang Zhang\thanks{Department of Statistics, Texas A\&M University, Collage Station, TX 77843, USA. E-mail: zhangxiany@stat.tamu.edu. }
\medskip\\
}
\date{\normalsize This version: \today}
\maketitle
\sloppy%
\onehalfspacing \small \textbf{Abstract} Motivated by the likelihood
ratio test under the Gaussian assumption, we develop a maximum
sum-of-squares test for conducting hypothesis testing on high
dimensional mean vector. The proposed test which incorporates the
dependence among the variables is designed to ease the computational
burden and to maximize the asymptotic power in the likelihood ratio
test. A simulation-based approach is developed to approximate the
sampling distribution of the test statistic. The validity of the
testing procedure is justified under both the null and alternative
hypotheses. We further extend the main results to the two sample problem without the equal covariance
assumption. Numerical results suggest that the proposed test can be
more powerful than some existing alternatives.
\\
\strut \textbf{Keywords:} High dimensionality, Maximum-type-test,
Simulation-based approach, Sparsity, Sum-of-squares test
\strut \onehalfspacing 


\section{Introduction}
Due to technology advancement, modern statistical data analysis
often deals with high dimensional data arising from many areas such
biological studies. High dimensionality poses significant
challenge to hypothesis testing.
In this paper, we consider a canonical hypothesis testing problem in
multivariate analysis, namely inference on mean vector. Let
$\{X_i\}^{n}_{i=1}$ be a sequence of i.i.d $p$-dimensional random
vectors with $\E X_i=\theta.$ We are interested in testing
$$H_0: \theta=0_{p\times 1}~\text{versus}~~H_a:\theta\neq 0_{p\times 1}.$$
When $p\ll n$, the Hotelling's $T^2$ test has been shown to enjoy
some optimal properties for testing $H_0$ against $H_a$ [Anderson (2003)]. However, for large
$p$, the finite sample performance of the Hotelling's $T^2$ test is
often unsatisfactory. To cope with the high dimensionality, several
alternative approaches have been suggested, see e.g Bai and Saranadasa (1996), Srivastava
and Du (2008), Srivastava (2009), Chen and Qin (2010), Lopes et al. (2009), Cai et al. (2014), Gregory et al. (2015) and references therein. In general, these tests can be categorized into two types:
the sum-of-squares type test and the maximum type test. The
former is designed for testing dense but possibly weak signals,
i.e., $\theta$ contains a large number of small non-zero entries.
The latter is developed for testing sparse signals, i.e., $\theta$
has a small number of large coordinates. In this paper, our interest
concerns the case of sparse signals which can arise in many real
applications such as detecting disease outbreaks in early stage, anomaly detection in medical imaging [Zhang et al. (2000)] and ultrasonic flaw detection in highly scattering materials [James et al. (2001)].

Suitable transformation on the original data, which explores the
advantages of the dependence structure among the variables, can lead
to magnified signals and thus improves the power of the testing
procedure. This phenomenon has been observed in the literature [see
e.g. Hall and Jin (2010); Cai et al. (2014); Chen et al. (2014); Li
and Zhong (2015)]. To illustrate the point, we consider the signal
$\theta=(\theta_1,\dots,\theta_{200})'$, where $\theta$ contains 4
nonzero entries whose magnitudes are all equal to $\psi_j\sqrt{\log
(200)/100}\approx 0.23\psi_j$ with $\psi_j$ being independent and
identically distributed (i.i.d) random variables such that
$P(\psi_j=\pm 1)=1/2.$ Let $\Sigma=(\sigma_{i,j})^{p}_{i,j=1}$ with
$\sigma_{i,j}=0.6^{|i-j|}$, and $\Gamma=\Sigma^{-1}$. Figure
\ref{fig:sig} plots the original signal $\theta$ as well as the
signal after transformation and studentization
$\widetilde{\theta}=(\widetilde{\theta}_1,\dots,\widetilde{\theta}_p)'$
with $\widetilde{\theta}_j=(\Gamma\theta)_j/\sqrt{\gamma_{jj}}$. It
is clear that the linear transformation magnifies both the strength
and the number of signals (the number of nonzero entries increases
from 4 to 12 after the linear transformation). In the context of
signal recovery, the additional nonzero entries are treated as fake
signals and need to be excised, but they are potentially helpful in
simultaneous hypothesis testing as they carry certain information
about the presence of signal. In general, it appears to be sensible
to construct a test based on the transformed data, which targets not
only for the largest entry [see Cai et al. (2014)] but also other
leading components in $\widetilde{\theta}$. Intuitively such test
can be constructed based on $\bar{Z}=\Gamma\bar{X}$ with
$\bar{X}=\sum^{n}_{i=1}X_i/n$ being the sample mean, which serves as
a natural estimator for $\Gamma\theta.$ For known $\Gamma$, a test
statistic which examines the largest $k$ components of
$\Gamma\theta$ can be defined as,
$$T_n(k)=n\max_{1\leq j_1<j_2<\cdots<j_k\leq p}\sum^{k}_{s=1}\frac{\bar{z}_{j_s}^2}{\gamma_{j_s,j_s}}.$$
\begin{figure}[h]
\centering
\includegraphics[height=7cm,width=7cm]{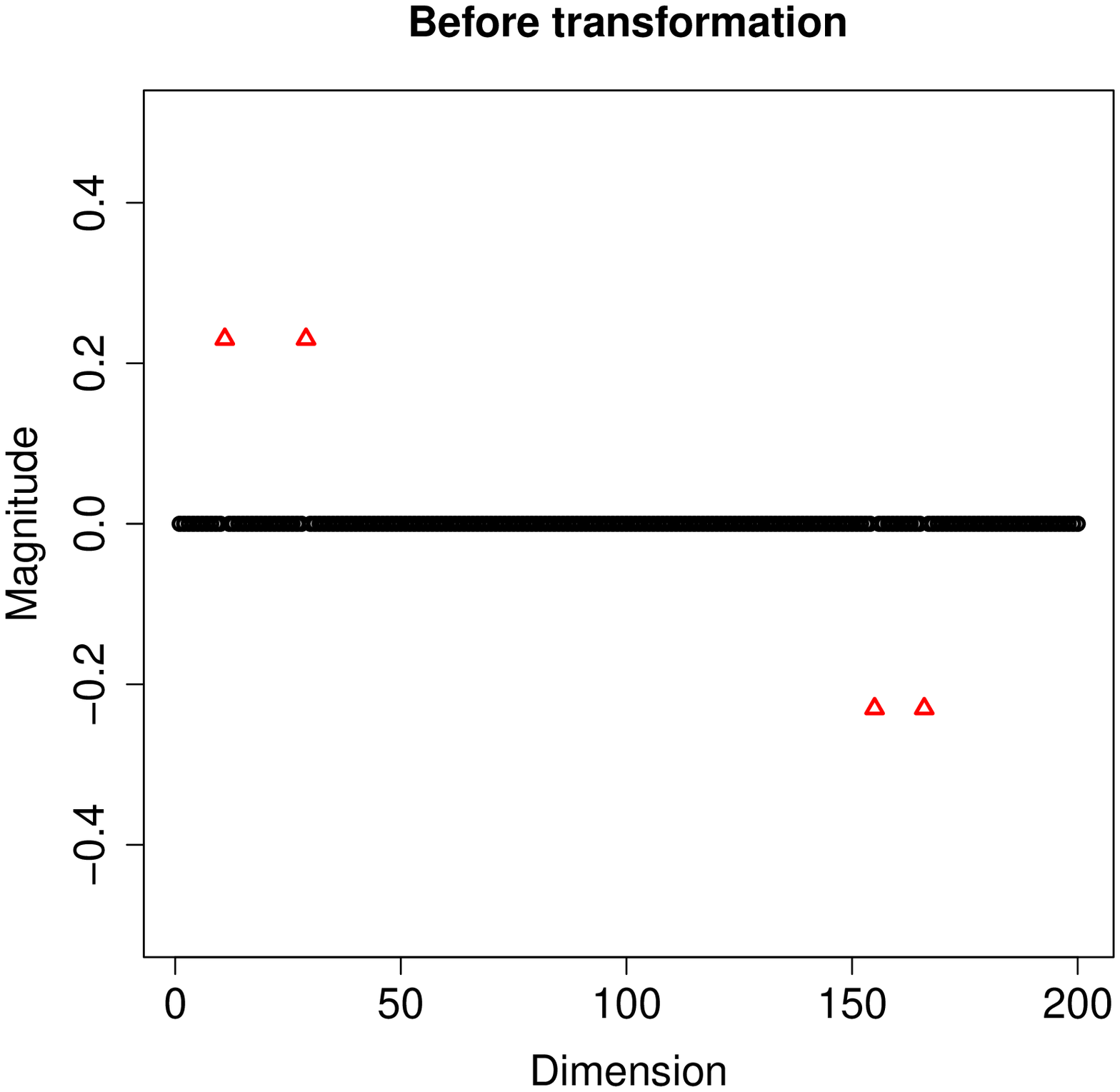}
\includegraphics[height=7cm,width=7cm]{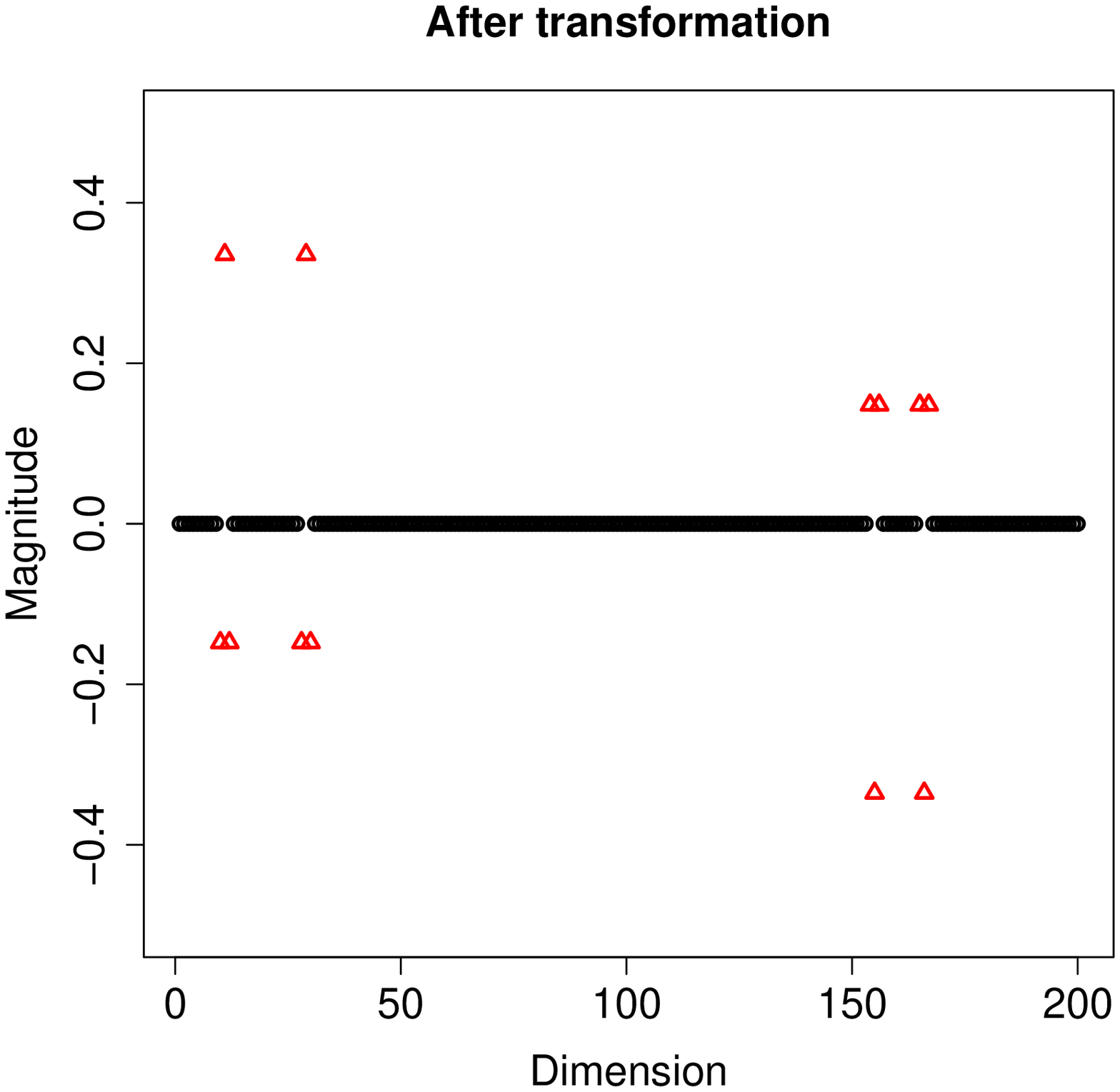}
\caption{Signals before (left panel) and after (right panel) the
linear transformation $\Gamma$. The non-zero entries are denoted by
\Zhang{$\bigtriangleup$}. }\label{fig:sig}
\end{figure}
If the \emph{original} mean $\theta$ contains exactly $k$ nonzero
components with relatively strong signals, it seems reasonable to
expect that $T_n(k)$ outperforms $T_n(1)$.
Interestingly, we shall show that the test statistic $T_n(k)$ is
closely related with the likelihood ratio (LR) test for testing
$H_0$ against a sparse alternative on $\theta$.

Our derivation also provides insight on some methods in the
literature. In particular, we show that the data transformation
based on the precision matrix proposed in Cai et al.(2014) can be
derived explicitly using the maximum likelihood principle when
$\Theta_a$ is the space of vectors with exactly one nonzero
component. We also reveal a connection between $T_n(k)$ and the thresholding test in Fan
(1996).

The rest of the paper is organized as follows. Adopting the maximum
likelihood viewpoint, we develop a new class of tests named maximum
sum-of-squares tests in Section \ref{sec:m}. Section \ref{sec:power}
presents a power analysis. In Section
\ref{sec:fes}, we introduce the feasible testing procedure by
replacing the precision matrix by its estimator. A simulation-based
approach is proposed to approximate the sampling distribution of the
test. We describe a modified testing procedure in Section
\ref{sec:k}. Section \ref{sec:theory} presents some theoretical
results based on the Gaussian approximation theory for high
dimensional vector. We extend our main results to the two sample
problem in Section \ref{sec:two}. Section \ref{sec:sim} reports some
numerical results. Section \ref{sec:con} concludes. The technical
details are deferred to the appendix.

\emph{Notation.} For a vector $a=(a_1,\dots,a_p)'$ and $q>0$, define
$|a|_q=(\sum^{p}_{i=1}|a_i|^q)^{1/q}$ and $|a|_{\infty}=\max_{1\leq
j\leq p}|a_j|$. Set $|\cdot|=|\cdot|_2$. Denote by $||\cdot||_0$ the
$l_0$ norm of a vector or the cardinality of a set. For
$C=(c_{ij})^{p}_{i,j=1}\in \mathbb{R}^{p\times p}$, define
$||C||_1=\max_{1\leq j\leq p}\sum^{p}_{i=1}|c_{ij}|$,
$||C||_2=\max_{|a|=1}|Ca|$ and $||C||_{\infty}=\max_{1\leq i,j\leq
p}|c_{ij}|$. Denote by $\text{diag}(C)$ the diagonal matrix
$\text{diag}(c_{11},c_{22},\dots,c_{pp})$. The notation
$N_p(\theta,\Sigma)$ is reserved for the $p$-variate multivariate
normal distribution with mean $\theta$ and covariance matrix
$\Sigma$.

\section{Main results}
\subsection{Likelihood ratio test}\label{sec:m}
Let $X_i=(x_{i1},\dots,x_{ip})'$ be a sequence of i.i.d $N_p(\theta,\Sigma)$ random
vectors with $\Sigma=(\sigma_{ij})^{p}_{i,j=1}$. We are interested in testing
\begin{equation}\label{prob}
H_0: \theta\in \Theta_0=\{0_{p\times 1}\} \quad \text{versus} \quad
H_{a}: \theta\in \Theta_a\subseteq \Theta_0^c.
\end{equation}
Let $\Theta_{a,k}=\{b\in\mathbb{R}^p:||b||_0=k\}$, where
$||\cdot||_0$ denotes the $l_0$ norm of a vector. Notice that
$\Theta_a\subseteq \Theta_0^c=\cup_{k=1}^p\Theta_{a,k}.$ A practical
challenge for conducting high dimensional testing is the
specification of the alternative space $\Theta_a,$ or in another
word, the direction of possible violation from the null hypothesis.

Hypothesis testing for high-dimensional mean has received
considerable attention in recent literature, see e.g. Srivastava and
Du (2008), Srivastava (2009), Chen and Qin (2010), Lopes et al.
(2009), Cai et al. (2014), Gregory et al. (2015) among others. Although existing testing
procedures are generally designed for a particular type of
alternatives, the alternative space is not often clearly specified.
In this paper, we shall study (\ref{prob}) with the alternative
space $\Theta_a$ stated in a more explicit way (to be more precise, it is stated in terms of the $l_0$ norm). By doing so, one can
derive the test which targets for a particular type of alternative.
This formulation also sheds some light on some existing tests.
To motivate the subsequent derivations, we consider the following
problem
\begin{equation}\label{prob2}
H_0:\theta\in\Theta_0\quad \text{versus}\quad
H_{a,k}:\theta\in\Theta_{a,k}.
\end{equation}
Given the covariance matrix $\Sigma$ or equivalently the precision
matrix $\Gamma=(\gamma_{ij})^{p}_{i,j=1}:=\Sigma^{-1}$, we shall
develop a testing procedure based on the maximum likelihood
principle. Under the Gaussian assumption, the negative
log-likelihood function (up to a constant) is given by
\begin{align*}
l_n(\theta)=\frac{1}{2}\sum^{n}_{i=1}(X_i-\theta)'\Gamma(X_i-\theta).
\end{align*}
The maximum likelihood estimator (MLE) under $H_{a,k}$ is defined as
\begin{equation}
\hat{\theta}=\arg\min_{\theta\in
\Theta_{a,k}}\sum^{n}_{i=1}(X_i-\theta)'\Gamma(X_i-\theta).
\end{equation}
To illustrate the idea, we first consider the case $\Gamma=I_p$,
i.e., the components of $X_i$ are i.i.d $N(0,1)$. It is
straightforward to see that the $k$ nonzero components of
$\hat{\theta}$ are equal to
$\bar{x}_{j_s^*}=\sum^{n}_{i=1}x_{ij_s^*}/n$, where
$$(j_1^*,j_2^*,\dots,j_k^*)=\underset{1\leq j_1<j_2<\cdots<j_k\leq p}{\arg\max}\sum^{k}_{s=1}\bar{x}_{j_s}^2.$$
Although the maximum is taken over $\binom{n}{k}$ possible sets, it
is easy to see that $j_1^*,\dots,j_k^*$ are just the indices
associated with the $k$ largest $|\bar{x}_j|$, i.e., we only need to
sort $|\bar{x}_j|$ and pick the indices associated with the $k$
largest values. In this case, the LR test (with known $\Gamma$) can
be written as maximum of sum-of-squares, i.e.,
\begin{align*}
LR_n(k)=n\max_{1\leq j_1<j_2<\cdots<j_k\leq
p}\sum^{k}_{s=1}\bar{x}_{j_s}^2.
\end{align*}
The LR test is seen as a combination of the maximum type test and
the sum-of-squares type test and it is designed to optimize the
power for testing $H_0$ against $H_{a,k}$ with $k\geq 1.$ The two
extreme cases are $k=1$ (the sparsest alternative) and $k=p$ (the
densest alternative). In the former case, we have
$LR_n(1)=n|\bar{X}|^2_{\infty}=n\max_{1\leq j\leq p}|\bar{x}_j|^2$,
while in the latter case,
$LR_n(p)=n|\bar{X}|^2=n\sum^{p}_{j=1}|\bar{x}_j|^2$, where
$\bar{X}=\sum^{n}_{i=1}X_i/n=(\bar{x}_{1},\dots,\bar{x}_p)'$.

We note that an alternative expression for the LR test is given by
\begin{align}\label{eq:sparse}
Thred_n(\delta)=n\sum^{p}_{j=1}\bar{x}_j^2\mathbf{1}\{|\bar{x}_j|>\delta\},
\end{align}
for some $\delta>0,$ where $\mathbf{1}\{\cdot\}$ denotes the
indicator function. Thus $LR_n(k)$ can also be viewed as a
thresholding test [see Donoho and Johnstone (1994); Fan (1996)]. In
this paper, we focus on the regime of very sparse (e.g. the number
of nonzero entries grows slowly with $n$) while strong signals (e.g.
the cumulative effect of the nonzero entries of $\theta$ is greater
than $\sqrt{2k\log (p)/n}$), and choose
$\delta=|\bar{x}_{j_{k+1}^*}|$ with $|x_{j_1^*}|\geq |x_{j_2^*}|\geq
\cdots|x_{j_p^*}|$ in (\ref{eq:sparse}) (assuming that
$|x_{j_k^*}|>|x_{j_{k+1}^*}|$). For weaker but denser signals, Fan
(1996) suggested the use of $\delta=\sqrt{2\log(pa_p)/n}$ for
$a_p=c_1(\log p)^{-c_2}$ with $c_1,c_2>0.$ A more delicate regime is
where the signals are weak so that they cannot have a visible effect
on the upper extremes, e.g., the strength of signals is
$\sqrt{2r\log(p)/n}$ for $r\in (0,1)$. In this case, the signals and
noise may be almost indistinguishable. To tackle this challenging
problem, the thresholding test with $\delta=\sqrt{2s\log (p)/n}$ for
$s\in (0,1)$ was recently considered in Zhong et al. (2013). And a
second level significance test by taking maximum over a range of
significance levels (the so-called Higher Criticism test) was used
to test the existence of any signals [Donoho and Jin (2004)].

It has been shown in the literature that incorporating the
componentwise dependence helps to boost the power of the testing
procedure [Hall and Jin (2010); Cai et al. (2014); Chen et al.
(2014)]. Below we develop a general test which takes the advantage
of the correlation structure contained in $\Gamma.$ We first
introduce some notation. For a set $S\subseteq \{1,2,\dots,p\}$, let
$\Gamma_{S,S}$ be the submatrix of $\Gamma$ that contains the rows
and columns in $S$. Similarly we can define $\Gamma_{S,-S}$ with the
rows in $S$ and the columns in $\{1,2,\dots,p\}\setminus S$. Further
let $\theta_S=(\theta_j)_{j\in S}$, $\bar{X}_S=(\bar{x}_j)_{j\in S}$
and $\bar{X}_{-S}=(\bar{x}_j)_{j\in \{1,2,\dots,p\}\setminus S}$.
For a set $S$ with $||S||_{0}=k$, we consider the following
optimization problem,
\begin{align*}
&\max_{\theta\in\mathbb{R}^p:\theta_j=0,j\notin
S}\left(\theta'\Gamma\bar{X}-\frac{1}{2}\theta'\Gamma\theta\right)
=\max_{\theta_S\in\mathbb{R}^{k}}\left\{
\theta_S'(\Gamma_{S,S}\bar{X}_S+\Gamma_{S,-S}\bar{X}_{-S})-\frac{1}{2}\theta'_S\Gamma_{S,S}\theta_S\right\}.
\end{align*}
The solution to the above problem is
$\theta_S=\Gamma_{S,S}^{-1}(\Gamma_{S,S}\bar{X}_S+\Gamma_{S,-S}\bar{X}_{-S})$
with the corresponding maximized value equal to
$(\Gamma_{S,S}\bar{X}_S+\Gamma_{S,-S}\bar{X}_{-S})'\Gamma_{S,S}^{-1}(\Gamma_{S,S}\bar{X}_S+\Gamma_{S,-S}\bar{X}_{-S})/2$.
Based on the above derivation, the LR test for testing $H_0$ against
$H_{a,k}$ is given by
\begin{align*}
LR_n(k)=&\max_{\theta\in
\Theta_{a,k}}\sum^{n}_{i=1}\left\{X_i'\Gamma
X_i-(X_i-\theta)'\Gamma(X_i-\theta)\right\}=2n\max_{\theta\in
\Theta_{a,k}}\left(\theta'\Gamma\bar{X}-\frac{1}{2}\theta'\Gamma\theta\right)
\\=&2n\max_{S:||S||_{0}=k}\max_{\theta_S\in\mathbb{R}^{k}}\left\{
\theta_S'(\Gamma_{S,S}\bar{X}_S+\Gamma_{S,-S}\bar{X}_{-S})-\frac{1}{2}\theta'_S\Gamma_{S,S}\theta_S\right\}
\\=&n\max_{S:||S||_{0}=k}(\Gamma_{S,S}\bar{X}_S+\Gamma_{S,-S}\bar{X}_{-S})'\Gamma_{S,S}^{-1}(\Gamma_{S,S}\bar{X}_S+\Gamma_{S,-S}\bar{X}_{-S}).
\end{align*}
Letting $Z=(z_1,\dots,z_p)'=\Gamma \bar{X}$, a simplified
expression is then given by
$$LR_n(k)=n\max_{S:||S||_{0}=k}Z_S'\Gamma_{S,S}^{-1}Z_S.$$
It is worth pointing out that $LR_n(k)$ is indeed the LR test for testing
$$H_0:\theta\in\Theta_0 \quad \text{against}\quad  H_{a,1:k}: \theta\in
\cup_{j=1}^k\Theta_{a,j},$$ because
$\hat{\theta}=\text{argmin}_{\theta\in
\cup^{k}_{j=1}\Theta_{a,j}}\sum^{n}_{i=1}(X_i-\theta)'\Gamma(X_i-\theta)=\text{argmin}_{\theta\in
\Theta_{a,k}}\sum^{n}_{i=1}(X_i-\theta)'\Gamma(X_i-\theta).$ As an
illustration, we consider the following two examples.

\begin{example}[Sparsest case]
{\rm When $k=1$, we have
\begin{align*}
LR_n(1)=n\max_{1\leq j\leq p}\frac{|z_j|^2}{\gamma_{jj}},
\end{align*}
which has been recently considered in Cai et al. (2014) in the two
sample problem. Cai et al. (2014) pointed out that ``the linear
transformation $\Gamma X_i$ magnifies the signals and the number of
the signals owing to the dependence in the data''. Although a rigorous theoretical justification was provided in Cai et al. (2014), the linear transformation
based on $\Gamma$ still seems somewhat mysterious. Here we ``rediscover'' the test from a different
perspective.
}
\end{example}

\begin{example}[Densest case]
{\rm To test against the dense alternative $H_{a,p},$ one may
consider
\begin{align*}
LR_n(p)=n\bar{X}'\Gamma\bar{X}=n\sum_{i,j=1}^p
\bar{x}_i\bar{x}_j\gamma_{ij},
\end{align*}
or its $U$-statistic version,
\begin{align*}
LR_{n,U}(p)=\frac{1}{n-1}\sum_{i,j=1}^p \gamma_{ij}\sum_{k\neq
l}x_{ki}x_{lj}.
\end{align*}
In view of the results in Chen and Qin (2010), the asymptotic behavior of such test is expected to be very
different from $LR_n(k)$ with relatively small $k$. A serious investigation
for this test is beyond the scope of the current paper. }
\end{example}

We note that the test statistic $LR_n(k)$ involves taking
maximization over $\binom{p}{k}$ tuples $(j_1,\dots,j_k)$ with
$1\leq j_1<\cdots<j_k\leq p$, which can be computationally very
intensive if $p$ is large. To reduce the computational burden, we
consider the following modified test by replacing $\Gamma_{S,S}$
with $\text{diag}(\Gamma_{S,S})$ which contains only the diagonal
elements of $\Gamma_{S,S}$. With this substitution, we have
\begin{align*}
T_n(k)=&n\max_{S:||S||_{0}=k}Z_S'\text{diag}^{-1}(\Gamma_{S,S})Z_S
=n\max_{1\leq j_1<j_2\cdots<j_k\leq
p}\sum^{k}_{l=1}\frac{|z_{j_l}|^2}{\gamma_{j_l,j_l}}.
\end{align*}
To compute the modified statistic, one only needs to sort the values
$|z_{j_l}|^2/\gamma_{j_l,j_l}$ and find the indices corresponding to
the $k$ largest ones, say $j_1^*,j_2^*,\dots,j_k^*$. Then the test
statistic can be computed as
$$T_n(k)=n\sum^{k}_{l=1}\frac{|z_{j_l^*}|^2}{\gamma_{j_l^*,j_l^*}}.$$
Therefore, the computation cost for $T_n(k)$ with $k>1$ is
essentially the same as $LR_n(1)$. By the matrix inversion formula
$\Gamma_{-j,j}=-\Sigma_{-j,-j}^{-1}\Sigma_{-j,j}\gamma_{jj}$,
$z_j/\gamma_{jj}=\bar{x}_j-\Sigma_{j,-j}\Sigma_{-j,-j}^{-1}\bar{X}_{-j}$.
From the above derivation, we note that $n|z_{j}|^2/\gamma_{jj}$ can
be interpreted as the likelihood ratio test for testing $\theta_j=0$
given that $\theta_{k}=0$ for $k\neq j$. This strategy is
conceptually simple and can be conveniently implemented in practice.
Also it can be generalized to other parametric models.

\begin{remark}
{\rm
One may employ the so-called graph-assisted procedure [see e.g. Jin et al. (2014); Ke et al. (2014)] to circumvent the NP hard problem in the definition of $LR_n(k)$. Under the Gaussian assumption, $\Gamma$ defines a graph $(V,E)$ in terms of conditional (in)dependence, that is the nodes $i$
and $j$ are connected if and only if $\gamma_{ij}\neq 0.$ Let $J(1),\dots,J(q_0)$ be all the connected components of $(V,E)$ with size less or equal to $k$.
Then an alternative test statistic can be defined as,
\begin{align}\label{graph}
LR_{graph,n}(k)=n\max\sum^{k_0}_{i=1}Z_{J(j_i)}'\Gamma_{J(j_i),J(j_i)}^{-1}Z_{J(j_i)},
\end{align}
where the maximization is over all $\{j_1,\dots,j_{k_0}\}\subseteq \{1,2,\dots,q_0\}$ such that $\sum^{k_0}_{i=1}||J(j_i)||_0\leq k.$
Under suitable assumptions on $\Gamma$, it was shown in Jin et al. (2014) that the number of all connected components with size less or equal to $k$ is of the order $O(p)$ (up to a
multi-$\log(p)$ term). Greedy algorithm can be used to list all the sub-graphs. Note that $T_n(k)$ corresponds to the case where $J(j)=\{j\}$ for $1\leq j\leq p.$
Thus $LR_{graph,n}(k)$ could be viewed as a generalized version of $T_n(k)$ with the ability to explore the dependence in $Z$ via the connected
components of $(V,E)$.
}
\end{remark}

\begin{remark}
{\rm
Another strategy is to use marginal thresholding to screen out a large number of irrelevant variables [see e.g. Fan et al. (2015)]. Specifically, define $\mathcal{M}=\{1\leq j\leq p:|z_{j}|/\sqrt{\gamma_{j,j}}>\delta\}$, where $\delta$ is a proper threshold.
A test statistic can be defined as
\begin{align*}
LR_{threshold,n}(k)=&n\max_{S\subseteq \mathcal{M}:||S||_{0}=k}Z_S'\Gamma_{S,S}^{-1}Z_S,
\end{align*}
where the optimization can be solved for $\mathcal{M}$ with relatively small size using greedy algorithm.
}
\end{remark}


\subsection{Power analysis}\label{sec:power}
To better understand the power performance of $T_n(k)$, we present
below a small numerical study. Let $W=(w_1,w_2,\dots,w_p)\sim
N_p(0,\text{diag}^{-1/2}(\Gamma)\Gamma\text{diag}^{-1/2}(\Gamma))$
and recall that
$\widetilde{\theta}=(\widetilde{\theta}_1,\dots,\widetilde{\theta}_p)'$
with $\widetilde{\theta}_j=(\Gamma\theta)_j/\sqrt{\gamma_{jj}}$ from
the introduction. Define $T^W(k;\widetilde{\theta})=\max_{1\leq
j_1<j_2<\cdots j_k\leq
p}\sum^{k}_{l=1}(w_{j_l}+\sqrt{n}\widetilde{\theta}_{j_l})^2.$ It is obvious
that $T_n(k)=^d T^W(k;\widetilde{\theta})$, where ``$=^d$'' means
equal in distribution. Denote by $C_k(\alpha)$ the $100(1-\alpha)$th
quantile of $T^W(k;0)$, which can be obtained via simulation (in our study, $C_k(\alpha)$ is estimated via $100000$ simulation runs). Define
the power function
$$\mathcal{P}(k,\alpha,\widetilde{\theta},\Gamma)=P(T^W(k;\widetilde{\theta})>C_k(\alpha)).$$
We focus on the AR(1) covariance structure
$\Sigma=(\sigma_{i,j})^{p}_{i,j=1}$ with $\sigma_{i,j}=0.6^{|i-j|}$
and $\Gamma=\Sigma^{-1}$. The mean vector $\theta$ is assumed to
contain $k_0$ nonzero components with the same magnitude
$\sqrt{2r\log(p)/n}$, where the locations of the nonzero components
are drawn without replacement from $\{1,2,\dots,p\}$. Figure
\ref{fig:power} presents the power function
$\mathcal{P}(k,0.05,\widetilde{\theta},\Gamma)$ as a curve of $r$
which determines the signal strength, where $k_0=1,5,10,20$, and
$p=200,1000.$ These results are consistent with our statistical
intuition. In particular, for $k_0>1$, the power of $T_n(k)$ with
$k>1$ dominates the power of $T_n(1)$. Therefore from the power
consideration, when $\theta$ contains more than one nonzero
component, it seems beneficial to use $T_n(k)$ with $k>1$. This is
further confirmed in our simulation studies, see Section
\ref{sec:sim}.

\subsection{Feasible test}\label{sec:fes}
We have so far focused on the oracle case in which the precision
matrix is known. However, in most applications $\Gamma$ is unknown
and thus needs to be estimated. Estimating the precision matrix has
been extensively studied in the literature in recent years [see e.g.
Meinshausen and B\"{u}hlmann (2006); Bickel and Levina (2008a;
2008b); Friedman et al. (2008); Yuan (2010); Cai and Liu (2011); Cai
et al. (2011); Liu and Wang (2012); Sun and Zhang (2013)].

When $\Gamma$ is known to be banded or bandable, one can employ the
banding method based on the cholesky decomposition [Bickel and
Levina (2008a)] to estimate $\Gamma$. For sparse precision matrix
without knowing the banding structure, the nodewise Lasso and its
variants [Meinshausen and B\"{u}hlmann (2006); Liu and Wang (2012);
Sun and Zhang (2013)] or the constrained $l_1$-minimization for
inverse matrix estimation (CLIME) [Cai et al. (2011)] can be used to
estimate $\Gamma$.

In this paper, we use the nodewise Lasso regression to estimate the
precision matrix $\Gamma$ [Meinshausen and B\"{u}hlmann (2006)], but
other estimation approaches can also be used as long as the
resulting estimator satisfies some desired properties [see
(\ref{e1})-(\ref{e3})]. Let
$\widetilde{\mathbf{X}}:=(\widetilde{X}_1,\widetilde{X}_2,\dots,\widetilde{X}_n)'\in\mathbb{R}^{n\times
p}$ with $\widetilde{X}_i=X_i-\bar{X}$. Let
$\widetilde{\mathbf{X}}_{-j}$ be the $n\times (p-1)$ matrix without
the $j$th column. For $j=1,2,\dots,p$, consider
\begin{equation}
\widehat{\gamma}_j=\arg\min_{\gamma\in\mathbb{R}^{p-1}}(|\widetilde{X}_j-\widetilde{\mathbf{X}}_{-j}\gamma|^2/n+2\lambda_j|\gamma|_1),
\end{equation}
with $\lambda_j>0$, where $\widehat{\gamma}_j=\{\widehat{\gamma}_{jk}:~1\leq k\leq
p,~k\neq j\}$. Let
$\widehat{C}=(\widehat{c}_{ij})_{i,j=1}^{p}$ be a $p\times p$
matrix with $\widehat{c}_{ii}=1$ and
$\widehat{c}_{ij}=-\widehat{\gamma}_{ij}$ for $i\neq j$. Let
$\widehat{\tau}_j^2=|\widetilde{X}_j-\widetilde{\mathbf{X}}_{-j}\widehat{\gamma}_j|^2/n+\lambda_j|\widehat{\gamma}_j|_1$
and write
$\widehat{T}^2=\text{diag}(\widehat{\tau}^2_1,\dots,\widehat{\tau}^2_p)$
as a diagonal matrix. The nodewise Lasso estimator for
$\Gamma$ is constructed as
$$\widehat{\Gamma}=\widehat{T}^{-2}\widehat{C}.$$

\begin{remark}
{\rm We note that $\widehat{\Gamma}$ does not have to be symmetric.
As in Yuan (2010), we can improve them by using a symmetrization
step
$$\widetilde{\Gamma}=\underset{A:A=A'}{\arg\min}||A-\widehat{\Gamma}||_1,$$
which can be solved by linear programming. It is obvious that $\widetilde{\Gamma}$ is symmetric, but
not guaranteed to be positive-definite. Alternatively, semi-definite programming, which is somewhat more expensive
computationally, can be used to produce a nonnegative-definite $\widetilde{\Gamma}$.
}
\end{remark}

Given a suitable precision matrix estimator
$\widehat{\Gamma}=(\widehat{\gamma}_{ij})^{p}_{j=1}$ (e.g. obtained via nodewise Lasso), our feasible
test can be defined by replacing $\Gamma$ with its estimator, i.e.,
\begin{align*}
T_{fe,n}(k)=&n\max_{1\leq j_1<j_2\cdots<j_k\leq p}\sum^{k}_{l=1}\frac{|\widehat{z}_{j_l}|^2}{\widehat{\gamma}_{j_l,j_l}},
\end{align*}
where
$\widehat{Z}=\widehat{\Gamma}\bar{X}=(\widehat{z}_{1},\dots,\widehat{z}_p)'.$
Under suitable assumptions, it has been shown in Cai et al. (2014)
that $T_{fe,n}(1)$ converges to an extreme distribution of Type I.
To mimic the sampling distribution of $T_{fe,n}(k)$ for $k\geq 1$
under sparsity assumption, we propose a simulation-based approach which is related with the multiplier bootstrap approach in Chernozhukov et al. (2015).
The procedure can be described as follows:
\begin{enumerate}
\item Estimate $\Gamma$ by $\widehat{\Gamma}$ using a suitable regularization method.
\item Generate $\widehat{Z}^*=(\widehat{z}_{1}^*,\dots,\widehat{z}_p^*)'=\widehat{\Gamma}\sum^{n}_{i=1}(X_i-\bar{X})e_i/n$, where $e_i\sim^{i.i.d} N(0,1)$ are independent of the sample.
\item Compute the simulation-based statistic as
\begin{align*}
T_{fe,n}^*(k)=n\max_{1\leq j_1<j_2\cdots<j_k\leq
p}\sum^{k}_{l=1}\frac{|\widehat{z}_{j_l}^*|^2}{\widehat{\gamma}_{j_l,j_l}}.
\end{align*}
\item Repeat Steps 2-3 several times to get the $1-\alpha$ quantile of
$T_{fe,n}^*(k),$ which serves as the simulation-based critical
value.
\end{enumerate}

\subsection{Choice of $k$ and a modified test}\label{sec:k}
In this subsection, we propose a data dependent method for choosing
$k$ which is motivated from the power consideration. Consider the
Hotelling's $T^2$ test $T^2_n:=n\bar{X}'\widehat{S}^{-1}\bar{X}$
with $\widehat{S}$ being the sample covariance matrix. Bai and
Saranadasa (1996) showed that the asymptotic power function for the
Hotelling's $T^2$ test under $p/n\rightarrow b\in(0,1)$ has the form
\begin{align}\label{power}
\Phi\left(-z_{1-\alpha}+\sqrt{\frac{n(n-p)}{2p}}\theta'\Gamma\theta\right),
\end{align}
where $\Phi$ is the distribution function of $N(0,1)$ and
$z_{1-\alpha}$ is the $1-\alpha$ quantile of $N(0,1)$. Intuitively,
for a set $S$ with $||S||_0=k$ and $k<n$, one may expect that the
asymptotic power function of $nZ_S'\Gamma_{S,S}^{-1}Z_S$
is determined by the term,
$$\sqrt{\frac{n-k}{2k}}(\Gamma\theta)'_S\Gamma^{-1}_{S,S}(\Gamma\theta)_S.$$
For known $\Gamma$, we note that
$Z_S'\Gamma^{-1}_{S,S}Z_S-\frac{k}{n}$ is an unbiased estimator for
$(\Gamma\theta)'_S\Gamma^{-1}_{S,S}(\Gamma\theta)_S$. From the power
consideration, a natural test statistic can be defined as
\begin{align*}
\max_{1\leq k\leq M}\max_{||S||_0=k}\sqrt{\frac{n-k}{2k}}\left(Z_S'\Gamma^{-1}_{S,S}Z_S-\frac{k}{n}\right),
\end{align*}
where $M$ is an upper bound. By replacing $\Gamma^{-1}_{S,S}$ with
$\text{diag}^{-1}(\Gamma_{S,S})$, a computational feasible testing
procedure is then given by
$$\widetilde{T}_{n}(M)=\max_{1\leq k\leq M}\max_{||S||_0=k}\sqrt{\frac{1-k/n}{2k}}\left(nZ_S'\text{diag}^{-1}(\Gamma_{S,S})Z_S-k\right)
=\max_{1\leq k\leq M}\sqrt{\frac{1-k/n}{2k}}\left(T_n(k)-k\right).$$
Substituting $\Gamma$ with $\widehat{\Gamma}$, we therefore propose the
following test
$$\widetilde{T}_{fe,n}(M)=\max_{1\leq k\leq M}\sqrt{\frac{1-k/n}{2k}}\left(T_{fe,n}(k)-k\right).$$
To approximate its sampling distribution, we suggest the following modified
simulation-based statistic in Step 3 above,
\begin{align*}
\widetilde{T}_{fe,n}^*(M)=\max_{1\leq k\leq
M}\sqrt{\frac{1-k/n}{2k}}\left(n\max_{1\leq j_1<j_2\cdots<j_k\leq
p}\sum^{k}_{l=1}\frac{|\widehat{z}_{j_l}^*|^2}{\widehat{\gamma}_{j_l,j_l}}-k\right),
\end{align*}
where
$\widehat{Z}^*=(\widehat{z}_{1}^*,\dots,\widehat{z}_p^*)'=\widehat{\Gamma}\sum^{n}_{i=1}(X_i-\bar{X})e_i/n$
with $e_i\sim^{i.i.d} N(0,1)$ that are independent of the sample.


\subsection{Theoretical results}\label{sec:theory}
In this subsection, we study the theoretical properties of the
proposed test and justify the validity of the simulation-based approach. To facilitate the derivations, we make
the following assumptions. Denote by $\lambda_{\min}(\Sigma)$ and
$\lambda_{\max}(\Sigma)$ the smallest and the largest eigenvalues of
$\Sigma$ respectively. Let $d=\max_{1\leq j\leq p}||\{\gamma_{jk}:
k\neq j, 1\leq k\leq p\}||_0$.

\begin{assumption}\label{ass1}
Suppose $\max_{1\leq j\leq p}\sigma_{j,j}<c_1$ and $c_2<\lambda_{\min}(\Sigma)$ for some $c_1,c_2>0$.
\end{assumption}

\begin{assumption}\label{ass2}
Suppose $d^2\log(p)/n=o(1)$.
\end{assumption}

Let $\widehat{\Sigma}=\sum^{n}_{i=1}(X_i-\bar{X})(X_i-\bar{X})'/n$.
Denote by $\widehat{\Gamma}_j$ and $\Gamma_j$ the $j$th rows of
$\widehat{\Gamma}$ and $\Gamma$ respectively.

\begin{proposition}\label{prop0}
Under Assumptions \ref{ass1}-\ref{ass2}, we have
\begin{align}
&\max_{1\leq j\leq p}|\widehat{\gamma}_{jj}-\gamma_{jj}|=O_p\left(\sqrt{\frac{d\log(p)}{n}}\right), \label{e1}
\\&\max_{1\leq j\leq p}|\widehat{\Gamma}_j-\Gamma_j|_1=O_p\left(d\sqrt{\frac{\log(p)}{n}}\right), \label{e2}
\\&||\widehat{\Gamma}\widehat{\Sigma}\widehat{\Gamma}'-\widehat{\Gamma}'||_{\infty}=O_p\left(\sqrt{\frac{d\log(p)}{n}}\right). \label{e3}
\end{align}
\end{proposition}

By the arguments in the proofs of Lemma 5.3 and Lemma 5.4 in van de Geer et al. (2014), we have (\ref{e1}), (\ref{e2}) and (\ref{e3}) hold
if $\widetilde{X}_i=X_i-\bar{X}$ is replaced by $X_i-\theta$ in the nodewise Lasso regression and $\widehat{\Sigma}$ is replaced by $\sum^{n}_{i=1}(X_i-\theta)(X_i-\theta)'/n$.
A careful inspection of their proofs shows that the conclusion remains valid when $\theta$ is replaced with $\bar{X}$. We omit the technical details here to conserve space.
We are now in position to present the main results in this
section. Define the quantity
$\phi(\Gamma;k)=\min_{|v|=1,||v||_0\leq k}v'\Gamma v$. Let $X^n_1=\{X_1,\dots,X_n\}$.

\begin{theorem}\label{thm}
Assume that $k^2d(\log(np))^{5/2}/\sqrt{n}=o(1)$ and
$\phi(\Gamma;k)>c$ for some positive constant $c$. Under
Assumptions \ref{ass1}-\ref{ass2} and $H_0$, we have
$$\sup_{t\geq 0}\left|P\left(T_{fe,n}^*(k)\leq t\bigg|X_1^n\right)-P\left(T_{fe,n}(k)\leq t\right)\right|=o_p(1).$$
\end{theorem}

\begin{theorem}\label{thm2}
Assume that $M^4d(\log(np))^{5/2}/\sqrt{n}=o(1)$ and
$\phi(\Gamma;M)>c$ for some positive constant $c$. Under
Assumptions \ref{ass1}-\ref{ass2} with $k$ replaced by $M$ and
$H_0$, we have
$$\sup_{t_M\geq t_{M-1}\geq \cdots \geq t_1\geq 0}\left|P\left(\bigcap^{M}_{j=1}\left\{T_{fe,n}^*(j)\leq t_j\right\}\bigg|X_1^n\right)-P\left(\bigcap^M_{j=1}\left\{T_{fe,n}(j)\leq t_j\right\}\right)\right|=o_p(1).$$
As a consequence, we have
$$\sup_{t\geq 0}\left|P\left(\widetilde{T}_{fe,n}^*(M)\leq t\bigg|X_1^n\right)-P\left(\widetilde{T}_{fe,n}(M)\leq t\right)\right|=o_p(1).$$
\end{theorem}

Next we study the power property of the proposed testing procedure. To proceed, we impose the following conditions.

\begin{assumption}\label{ass3}
Assume that $\max_{1\leq j_1<j_2<\cdots<j_k\leq
p}\sum^{k}_{l=1}\gamma_{j_l,j_l}\theta_{j_l}^2\geq
(2k+\epsilon)\log(p)/n$ for some $\epsilon>0$.
\end{assumption}

\begin{assumption}\label{ass4}
Suppose $\sum^{p}_{j=1}I\{\theta_j\neq 0\}=p^r$ for some $0\leq
r<1/4$, and the non-zero locations are randomly uniformly drawn from
$\{1,2,\dots,p\}$. And the scheme is independent of $\{X_i-\theta\}^{n}_{i=1}$.
\end{assumption}

\begin{assumption}\label{ass4a}
Let $\text{diag}^{-1/2}(\Gamma)\Gamma\text{diag}^{-1/2}(\Gamma)=(\nu_{ij})_{i,j=1}^p$. Assume that $\max_{1\leq i<j\leq p}|\nu_{ij}|\leq c_0<1$ for some constant $0<c_0<1$. Further assume that $\lambda_{\max}(\Sigma)\leq C_0$ for some constant $C_0>0.$
\end{assumption}

Define $c^*_{\alpha}(k)=\inf\{t>0:P(T_{fe,n}^*(k)\leq t|X^n_1)\geq
1-\alpha\}$ and
$\widetilde{c}^*_{\alpha}(M)=\inf\{t>0:P(\widetilde{T}_{fe,n}^*(M)\leq
t|X^n_1)\geq 1-\alpha\}$ the simulation-based critical values. The
consistency of the testing procedure is established in the following
theorem.

\begin{theorem}\label{thm:consis}
Suppose $k^2d(\log(np))^{5/2}/\sqrt{n}=o(1)$. Under
Assumptions \ref{ass1}-\ref{ass4a}, we have
\begin{equation}
P(T_{fe,n}(k)>c_{\alpha}^*(k))\rightarrow 1.
\end{equation}
Moreover, suppose Assumption \ref{ass3} holds with $k=M$. Then for $M$ such that
$M^4d(\log(np))^{5/2}/\sqrt{n}=o(1)$,
\begin{align*}
P(\widetilde{T}_{fe,n}(M)>\widetilde{c}_{\alpha}^*(M))\rightarrow 1.
\end{align*}
\end{theorem}
When $k=1$, Assumption \ref{ass3} reduces to
$\max_{1\leq j\leq p}|\theta_j|/\sqrt{\sigma_{jj}}\geq \sqrt{2\{1/(\sigma_{jj}\gamma_{jj})+\epsilon_0\}\log(p)/n},$
for some $\epsilon_0>0$. According to Theorem 3 of Cai et al. (2014), the separation rate $\sqrt{\log (p)/n}$ is
minimax optimal.

Finally, we point out that the Gaussian assumption can be relaxed by
employing the recently developed Central Limit Theorem in high
dimension [Chernozhukov et al. (2015)]. For a random variable $X$,
we define the sub-Gaussian norm [see Definition 5.7 of Vershynin
(2012)] as
$$||X||_{\psi}=\sup_{q\geq 1}q^{-1/2}(\E|X|^q)^{1/q}.$$
\begin{assumption}\label{ass5}
Assume that $\sup_{v\in\mathbb{S}^{p-1}}||v'X_i||_{\psi}<c_3$ and $\sup_{v\in\mathbb{S}^{p-1}}||v'\Gamma X_i||_{\psi}< c_4$ for some constants $c_3,c_4>0$.
\end{assumption}

Let $W=(w_1,w_2,\dots,w_p)\sim
N_p(0,\text{diag}^{-1/2}(\Gamma)\Gamma\text{diag}^{-1/2}(\Gamma))$
and define $T^W(k)=\max_{1\leq j_1<j_2<\cdots j_k\leq
p}\sum^{k}_{l=1}w_{j_l}^2.$ Specifically we have the following
result, which indicates that under the sub-Gaussian assumption, the
distribution of $T_{fe,n}(k)$ can be approximated by its Gaussian
counterpart $T^W(k)$.
\begin{proposition}\label{prop00}
Assume that $d(k\log(np))^{7/2}/\sqrt{n}=o(1)$ and
$\phi(\Gamma;k)>c$ for some positive constant $c>0$. Under
Assumptions \ref{ass1}, \ref{ass2}, \ref{ass5} and $H_0$, we have
$$\sup_{t\geq 0}\left|P\left(T_{fe,n}(k)\leq t\right)-P\left(T^W(k)\leq t\right)\right|=o(1).$$
\end{proposition}

\section{Extension to the two sample problem}\label{sec:two}
\subsection{Likelihood ratio test}
The maximum likelihood viewpoint allows a direct extension of the above procedure to the two sample problem.
Consider two samples $\{X_i\}^{n_1}_{i=1}\sim^{i.i.d} N_p(\theta_1,\Sigma_1)$
and $\{Y_i\}^{n_2}_{i=1}\sim^{i.i.d} N_p(\theta_2,\Sigma_2)$, where the two samples are independent of each other. A canonical problem in multivariate analysis is the hypothesis testing of
$$H_0': \theta_1-\theta_2\in \Theta_0\quad \text{versus}\quad H_a': \theta_1-\theta_2\in \Theta_a\subseteq \Theta_0^c.$$
Given the priori $\theta_1-\theta_2 \in \Theta_{a,k}$, we consider
$$H_0': \Delta\in \Theta_0\quad \text{versus}\quad H_{a,k}': \Delta\in\Theta_{a,k},$$
where $\Delta=\theta_1-\theta_2$.

Notation-wise, let $\Gamma_j=\Sigma_j^{-1}$ for $j=1,2.$ Define
$C_1=(n_1\Gamma_1+n_2\Gamma_2)^{-1}n_1\Gamma_1$ and
$C_2=(n_1\Gamma_1+n_2\Gamma_2)^{-1}n_2\Gamma_2.$ Further let
$\Omega^{21}=C_2'\Gamma_1C_2$, $\Omega^{12}=C_1'\Gamma_2C_1$,
$\widetilde{X}=C_2'\Gamma_1(\bar{X}-\widetilde{\theta})$ and
$\widetilde{Y}=C_1'\Gamma_2(\bar{Y}-\widetilde{\theta})$, where
$\bar{X}=\sum^{n_1}_{i=1}X_i/n_1$, $\bar{Y}=\sum^{n_2}_{i=1}Y_i/n_2$
and
$$\widetilde{\theta}=(n_1\Gamma_1+n_2\Gamma_2)^{-1}\left(\Gamma_1\sum^{n_1}_{i=1}X_i+\Gamma_2\sum^{n_2}_{i=1}Y_i\right)$$
which is the MLE for $\theta:=\theta_1=\theta_2$ under the null. The
following proposition naturally extends the result in Section
\ref{sec:m} to the two sample case.
\begin{proposition}\label{prop1}
The LR test for testing $H_0'$ against $H_{a,k}'$ is given by
$$LR_n(k)=\max_{S:||S||_0=k}(n_1\widetilde{X}_S-n_2\widetilde{Y}_S)'(n_1\Omega_{S,S}^{21}+n_2\Omega_{S,S}^{12})^{-1}(n_1\widetilde{X}_S-n_2\widetilde{Y}_S).$$
\end{proposition}

\subsection{Equal covariance structure}\label{se:equal}
We first consider the case of equal covariance, i.e.,
$\Gamma:=\Gamma_1=\Gamma_2$. Simple calculation yields that
$C_1=n_1(n_1+n_2)^{-1}$, $C_2=n_2(n_1+n_2)^{-1}$,
$\widetilde{\theta}=(n_1\bar{X}+n_2\bar{Y})/(n_1+n_2),
\widetilde{X}=n_2^2\Gamma(\bar{X}-\bar{Y})/(n_1+n_2)^2,
\widetilde{Y}=n_1^2\Gamma(\bar{Y}-\bar{X})/(n_1+n_2)^2$ and
$n_1\Omega_{S,S}^{21}+n_2\Omega_{S,S}^{12}=n_1n_2\Gamma/(n_1+n_2).$
Thus the LR test can be simplified as,
\begin{equation}\label{eq:inf}
LR_n(k)=\max_{S:||S||_0=k}\frac{n_1n_2}{n_1+n_2}(\Gamma\bar{X}-\Gamma\bar{Y})'_S(\Gamma_{S,S})^{-1}(\Gamma\bar{X}-\Gamma\bar{Y})_S.
\end{equation}
We note that $LR_n(1)$ reduces to the two sample test proposed in
Cai et al. (2014). By replacing $\Gamma_{S,S}$ with
$\text{diag}(\Gamma_{S,S})$ in (\ref{eq:inf}), we obtain the
(infeasible) statistic
\begin{align*}
T_{n}(k)=\max_{1\leq j_1< j_2<\cdots<j_k\leq
p}\frac{n_1n_2}{n_1+n_2}\frac{(\Gamma\bar{X}-\Gamma\bar{Y})^2_{j_l}}{\gamma_{j_l,j_l}},
\end{align*}
which is computationally efficient.

Let $\widehat{\Gamma}=(\widehat{\gamma}_{i,j})^{p}_{i,j=1}$ be a
suitable estimator for $\Gamma$ based on the pooled sample. The
feasible test is given by
$$T_{fe,n}(k)=\max_{1\leq j_1< j_2<\cdots<j_k\leq p}\frac{n_1n_2}{n_1+n_2}\frac{(\widehat{\Gamma}\bar{X}-\widehat{\Gamma}\bar{Y})^2_{j_l}}{\widehat{\gamma}_{j_l,j_l}}.$$
To approximate the sampling distribution of the above test, one can
employ the simulation-based approach described below:
\begin{enumerate}
\item Estimate $\widehat{\Gamma}$ using suitable regularization method based on the pooled sample.
\item Let $X^*=\sum^{n_1}_{i=1}(X_i-\bar{X})e_{i}/n_1$ and $Y^*=\sum^{n_2}_{i=1}(Y_i-\bar{Y})\tilde{e}_{i}/n_1$, where
$\{e_i\}$ and $\{\tilde{e}_i\}$ are two independent sequences of i.i.d $N(0,1)$ random variables that are independent of the sample.
\item Compute the simulation-based statistic $T_{fe,n}^*(k)$ by replacing $\bar{X}$ and $\bar{Y}$ with $X^*$ and $Y^*$.
\item Repeat Steps 2-3 several times to get the $1-\alpha$ quantile of
$T_{fe,n}^*(k)$, which serves as the simulation-based critical
value.
\end{enumerate}

Next, we briefly discuss the choice of $k$. By Theorem
2.1 in Bai and Saranadasa (1996), we know that the asymptotic power
function for the two sample Hotelling's $T^2$ test is given by
\begin{align*}
\Phi\left(-z_{1-\alpha}+\sqrt{\frac{N(N-p)}{2p}}\frac{n_1n_2}{N^2}(\theta_1-\theta_2)'\Gamma(\theta_1-\theta_2)\right),
\end{align*}
under $p/N\rightarrow b\in (0,1)$, where $N=n_1+n_2-2$. Thus for
$k<N$, the asymptotic power of $T_{fe,n}(k)$ is related to
$$\sqrt{\frac{N-k}{2k}}\max_{||S||_0=k}\{\widehat{\Gamma}(\theta_1-\theta_2)\}'_S\text{diag}^{-1}(\widehat{\Gamma}_{S,S})\{\widehat{\Gamma}(\theta_1-\theta_2)\}_S.$$
Notice that
$$(\Gamma\bar{X}-\Gamma\bar{Y})'_S\text{diag}^{-1}(\Gamma_{S,S})(\Gamma\bar{X}-\Gamma\bar{Y})_S-\frac{k(n_1+n_2)}{n_1n_2}$$ is an
unbiased estimator for
$\{\Gamma(\theta_1-\theta_2)\}'_S\text{diag}^{-1}(\Gamma_{S,S})\{\Gamma(\theta_1-\theta_2)\}_S$.
Thus we propose to choose $k$ by
\begin{align*}
\widehat{k}=&\arg\max_{1\leq k\leq
M'}\sqrt{\frac{N-k}{2k}}\left(T_{fe,n}(k)-k\right),
\end{align*}
where $M'$ is a pre-specified upper bound for $k$. Following the
same spirit in Section \ref{sec:k}, a modified test statistic is
given by
$$\widetilde{T}_{fe,n}(k)=
\max_{1\leq k\leq M'}\sqrt{\frac{1-k/N}{2k}}\left(\max_{1\leq
j_1<j_2<\cdots<j_k\leq
p}\frac{n_1n_2}{n_1+n_2}\frac{(\widehat{\Gamma}\bar{X}-\widehat{\Gamma}\bar{Y})^2_{j_l}}{\widehat{\gamma}_{j_l,j_l}}-k\right),$$
and the simulation-based procedure can be used to approximate its
sampling distribution.

We can justify the validity of the testing procedure under both the
null and alternative hypotheses. The arguments are similar to those in the one sample case, see Sections \ref{sec:fes} and
\ref{appendix}.


\subsection{Unequal covariance structures}
In the case of unequal covariance structures i.e., $\Gamma_1\neq
\Gamma_2$, we cannot use the pooled sample to estimate the
covariance structures. Let $\widehat{\Gamma}_i$ with $i=1,2$ be
suitable precision matrix estimators based on each sample
separately. Denote by $\widehat{C}_i$ the estimator for $C_i$ with
$i=1,2$. A particular choice here is given by
$$\widehat{C}_1=(n_1\widehat{\Gamma}_1+n_2\widehat{\Gamma}_2)^{-1}n_1\widehat{\Gamma}_1,\quad \widehat{C}_2=(n_1\widehat{\Gamma}_1+n_2\widehat{\Gamma}_2)^{-1}n_2\widehat{\Gamma}_2.$$
Further define
$\widehat{\Omega}^{21}=\widehat{C}_2'\widehat{\Gamma}_1\widehat{C}_2$,
$\widehat{\Omega}^{12}=\widehat{C}_1'\widehat{\Gamma}_2\widehat{C}_1$,
$\widehat{X}=\widehat{C}_2'\widehat{\Gamma}_1(\bar{X}-\widehat{\theta})$
and
$\widehat{Y}=\widehat{C}_1'\widehat{\Gamma}_2(\bar{Y}-\widehat{\theta})$,
where
$$\widehat{\theta}=\widehat{C}_1\bar{X}+\widehat{C}_2\bar{Y}.$$
Let
$\widehat{\Psi}=(\widehat{\psi}_{ij})_{i,j=1}^{p}=n_1\widehat{\Omega}^{21}+n_2\widehat{\Omega}^{12}$,
and
$\widehat{G}=(\widehat{g}_1,\dots,\widehat{g}_p)'=n_1\widehat{X}-n_2\widehat{Y}$.
By replacing $\widehat{\Psi}$ with $\text{diag}(\widehat{\Psi})$, we
suggest the following computational feasible test,
\begin{equation}
T_{fe,n}(k)=\max_{1\leq j_1<j_2<\cdots<j_k\leq p}\sum^{k}_{l=1}\frac{|\widehat{g}_{j_l}|^2}{\widehat{\psi}_{j_l,j_l}}.
\end{equation}
When $k=1$, we have
\begin{equation}
T_{fe,n}(1)=\max_{1\leq j\leq p}\frac{|\widehat{g}_{j}|^2}{\widehat{\psi}_{jj}},
\end{equation}
which can be viewed as an extension of Cal et al. (2014)'s test statistic to
the case of unequal covariances. Again one can employ the
simulation-based approach to obtain the critical values for
$T_{fe,n}(k)$. In this case, a modified test can be defined in a
similar manner as
$$\widetilde{T}_{fe,n}(k)=
\max_{1\leq k\leq M''}\sqrt{\frac{1-k/N}{2k}}\left(\max_{1\leq
j_1<j_2<\cdots<j_k\leq
p}\sum^{k}_{l=1}\frac{|\widehat{g}_{j_l}|^2}{\widehat{\psi}_{j_l,j_l}}-k\right)$$
for some upper bound $M''$.

\section{Simulation studies}\label{sec:sim}
\subsection{Empirical size and power}
In this section, we report the numerical results for comparing the
proposed testing procedure with some existing alternatives.
Specially we focus on the two sample problem for testing
$H_0':\Delta\in \Theta_0$ against the alternatives
$H_{a,k}':\Delta\in\Theta_{a,k}$. Without loss of generality, we set
${\theta}_2={0}$. Note that under $H_{a,k}'$, ${\theta}_1$ has $k$
non-zero elements. Denote by $\lfloor x \rfloor$ the largest integer
not greater than $x$. We consider the settings below.
\begin{itemize}
\item[(1)] Case 1: $k=\lfloor 0.05 p\rfloor$ and
the non-zero entries are equal to $\varphi_j\sqrt{\log(p)/n}$, where
$\varphi_j$ are i.i.d random variables with $P(\varphi_j=\pm
1)=1/2$.
\item[(2)] Case 2: $k=\lfloor \sqrt{p}\rfloor$ and the strength of the signals is the same as (1).
\item[(3)] Case 3: $k=\lfloor p^{0.3}\rfloor$ and the nonzero entries are all equal to $\sqrt{4r\log p/n}$ with $r=0.1,0.2,0.3,0.4,$ and $0.5.$
\end{itemize}
Here the locations of the nonzero entries are drawn without
replacement from $\{1,2,\dots,p\}$. Following Cai et al. (2014), the
following four covariance structures are considered.
\begin{itemize}
  \item[(a)] (block diagonal ${\Sigma}$): ${\Sigma} = \left( \sigma_{j,k} \right)$\, where $\sigma_{j,j} = 1$\, and $\sigma_{j,k} = 0.8$\, for $2(r-1)+1 \le j \ne k \le 2r$,\, where $r = 1,\ldots,\lfloor p/2\rfloor$ and $\sigma_{j,k}=0$\, otherwise.

  \item[(b)] (`bandable' ${\Sigma}$): ${\Sigma} = \left( \sigma_{j,k} \right)$\, where $\sigma_{j,k} = 0.6^{|j-k|}$\, for $1 \le j, k \le p$.

  \item[(c)] (banded ${\Gamma}$): ${\Gamma} = \left( \gamma_{j,k} \right)$\, where $\gamma_{j,j}=2$\, for\, $j=1,\ldots,p$,\, $\gamma_{j,(j+1)}=0.8$\, for \,$j=1,\ldots,p-1$,\, $\gamma_{j,(j+2)}=0.4$\, for \,$j=1,\ldots,p-2$,\, $\gamma_{j,(j+3)}=0.4$\, for \,$j=1,\ldots,p-3$,\, $\gamma_{j,(j+4)}=0.2$\, for \,$j=1,\ldots,p-4$,\, $\gamma_{j,k}=\gamma_{k,j}$\, for \,$j,k=1,\ldots,p$,\, and \,$\gamma_{j,k}=0$\, otherwise.

  \item[(d)] (block diagonal $\Gamma$): Denote by $D$ a
diagonal matrix with diagonal elements generated independently
from the uniform distribution on $(1,3)$. Let $\Sigma_0$ be
generated according to (a). Define $\Gamma=D^{1/2}\Sigma^2_0D^{1/2}$ and
$\Sigma=\Gamma^{-1}$.
\end{itemize}
For each covariance structure, two independent random samples are
generated with the same sample size $n_1=n_2=80$ from the following
multivariate models,
\begin{equation}\label{eq:sim}
X=\theta_1+\Sigma^{1/2}U_1,\quad Y=\theta_2+\Sigma^{1/2}U_2,
\end{equation}
where $U_1$ and $U_2$ are two independent $p$-dimensional random
vectors with independent components such that $\E(U_j)=0$ and
$\var(U_j)=I_p$ for $j=1,2$. We consider two cases: $U_j\sim N(0,
I_p)$, and the component of $U_j$ is standardized Gamma(4,1) random
variable such that it has zero mean and unit variance. The
dimension $p$ is equal to $50,100$ or $200$. Throughout the
simulations, the empirical sizes and powers are calculated based on
1000 Monte Carlo replications.

To estimate the precision matrix, we use the nodewise
square root Lasso [Belloni et al. (2012)] proposed in Liu and Wang
(2012), which is essentially equivalent to the scaled-Lasso from Sun and Zhang (2013).
To select the tuning parameter $\lambda$ in the nodewise
square root Lasso, we consider the following criteria,
\begin{align*}
\lambda^*=\text{argmin}_{\lambda\in \Lambda_n}||\widehat{\Gamma}(\lambda)\widehat{\Sigma}\widehat{\Gamma}(\lambda)'-\widehat{\Gamma}(\lambda)||_{\infty}
\end{align*}
where $\widehat{\Sigma}$ is the pooled sample covariance matrix and
the minimization is taken over a prespecified finite set
$\Lambda_n$. Moreover, we employ the data dependent method in
Section \ref{se:equal} to select $k$ with the upper bound $M'=40$
(we also tried $M'=20,80$ and found that the results are basically
the same as those with $M'=40$). For the purpose of comparison, we
also implemented the Hotelling's $T^2$ test and the two sample tests
proposed in Bai and Saranadasa (1996), Chen and Qin (2010), and Cai
et al. (2014). As the results under the Gamma model are
qualitatively similar to those under the Gaussian model, we only
present the results from the Gaussian model. Table \ref{tb1}
summarizes the sizes and powers in cases 1 and 2. The empirical
powers in case 3 with $r$ ranging from 0.1 to 0.5 are presented in
Figure \ref{fig:p1}. Some remarks are in order regarding the
simulation results: (\rmnum{1}) the empirical sizes are reasonably
close to the nominal level 5\% for all the tests; (\rmnum{2}) the
proposed tests and the maximum type test in Cai et al. (2014)
significantly outperform the sum-of-squares type testing procedures
in terms of power under Models (a), (b) and (d); Under Model (c),
the proposed method is quite competitive to Chen and Qin (2010)'s
test which delivers more power than Cai et al. (2014)'s test in some
cases; (\rmnum{3}) $T_{fe,n}(k)$ is consistently more powerful than
Cai et al. (2014)'s test in almost all the cases; (\rmnum{4}) the
modified test $\widetilde{T}_{fe,n}(M')$ is insensitive to the upper
bound $M'$ (as shown in our unreported results). And its power is
very competitive to $T_{fe,n}(k)$ with a suitably chosen $k$.


\subsection{Power comparison under different signal allocations}
We conduct additional simulations to compare the power of the proposed method with alternative approaches under different signal allocations.
The data are generated from (\ref{eq:sim}) with Gaussian distribution and bandable covariance structure (b). Let $k=\lfloor 0.1p\rfloor$
and consider the following four patterns of allocation, where the locations of the nonzero entries are drawn without replacement from $\{1,2,\dots,p\}$.
\begin{itemize}
\item[\rmnum{1}] (Square root): the nonzero entries are equal $\sqrt{4r\log(p)/n}\sqrt{j/k}$ for $1\leq j\leq k$.
\item[\rmnum{2}] (Linear): the nonzero entries are equal $\sqrt{4r\log(p)/n}(j/k)$ for $1\leq j\leq k$.
\item[\rmnum{3}] (Rational): the nonzero entries are equal $\sqrt{4r\log(p)/n}(1/j)$ for $1\leq j\leq k$.
\item[\rmnum{4}] (Random): the nonzero entries are drawn uniformly from $(-\sqrt{4r\log(p)/n},\sqrt{4r\log(p)/n})$.
\end{itemize}
Figure \ref{fig:p2} reports the empirical rejection probabilities
for $p=100,200$, and $r$ ranging from 0.1 to 0.5. We observe that the slower the strength of the signals decays, the higher power the tests can generate.
The proposed method generally outperforms the two sample tests in Chen and Qin (2010) and Cai et al. (2014) especially when the magnitudes of signals decay slowly.
This result makes intuitive sense as when the magnitudes of signals are close, the top few signals together provide a stronger indication for the violation from the null as compared to the indication using only the largest signal.
To sum up, the numerical results demonstrate the advantages of the proposed
method over some competitors in the literature.

\section{Concluding remark}\label{sec:con}
In this paper, we developed a new class of tests named maximum
sum-of-squares tests for conducting inference on high dimensional
mean under sparsity assumption. The maximum type test has been shown to be optimal under
very strong sparsity [Arias-Castro et al. (2011)]. Our result
suggests that even for very sparse signal (e.g. $k$ grows slowly
with $n$), the maximum type test may be improved. It is worth
mentioning that our method can be extended to more general settings.
For example, consider a parametric model with the negative
log-likelihood (or more generally loss function)
$\mathcal{L}(Y,X'\beta)$, where $\beta\in \mathbb{R}^p$ is the
parameter of interest, $X$ is the $p$-dimensional covariate and $Y$
is the response variable. We are interested in testing
$H_0:\beta=0_{p\times 1}$ versus $H_{a,k}:\beta\in \Theta_{a,k}.$
Given $n$ observations $\{Y_i,X_i\}^{n}_{i=1}$, the LR test for
testing $H_0$ against $H_{a,k}$ is then defined as
$LR_n(\beta)=2\sum^{n}_{i=1}\mathcal{L}(Y_i,0)-2\min_{\beta\in
\Theta_{a,k}}\sum^{n}_{i=1}\mathcal{L}(Y_i,X_i'\beta).$ In the case
of linear model, it is related with the maximum spurious correlations
recently considered in Fan et al. (2015) under the null. It is of interest to study
the asymptotic properties of $LR_n(\beta)$ and investigate the Wilks
phenomenon in this more general context.



\section{Technical appendix}
\subsection{Preliminaries}

We provide proofs of the main results in the paper. Throughout the
appendix, let $C$ be a generic constant which is different from line
to line. Define the unit sphere $\mathbb{S}^{p-1}=\{b\in\mathbb{R}^{p}:|b|=1\}$.

For any $1\leq k\leq p$, define
\begin{align*}
\mathcal{A}(t;k)=\underset{j=1}{\overset{\binom{p}{k}}{\bigcap}}\mathcal{A}_{j}(t),
\quad \mathcal{A}_j(t)=\{w\in\mathbb{R}^p: w_{S_j}'w_{S_j}\leq t\}.
\end{align*}
Here $S_j$ is the $j$th subset of $[p]:=\{1,2,\dots,p\}$ with
cardinality $k$ for $1\leq j\leq \binom{p}{k}$. It is
straightforward to verify that $ \mathcal{A}_{j}(t)$ is convex and
it only depends on $w_{S_j}$, i.e. the components in $S_j$. The dual
representation [see Rockafellar (1970)] for the convex set
$\mathcal{A}_j(t)$ with $1\leq j\leq \binom{p}{k}$ is given by
\begin{align*}
\mathcal{A}_j(t)=\bigcap_{v\in \mathbb{S}^{p-1},
v_{S_j}\in \mathbb{S}^{k-1}}\{w\in\mathbb{R}^p:w'v\leq \sqrt{t}\},
\end{align*}
where we have used the fact that $\sup_{v\in \mathbb{S}^{p-1},
v_{S_j}\in \mathbb{S}^{k-1}}w'v=|w_{S_j}|$ by the Cauchy-Schwartz
inequality. Define $\mathcal{F}=\{v\in \mathbb{S}^{p-1}, ||v||_0\leq
k\}$. It is not hard to see that
$$\mathcal{A}(t;k)=\bigcap_{v\in \mathcal{F}}\{w\in\mathbb{R}^p:w'v\leq \sqrt{t}\}.$$

Let $\mathcal{X}$ be a subset of a Euclidean space and let
$\epsilon>0$. A subset $N_{\epsilon}$ of $\mathcal{X}$ is called an
$\epsilon$-net of $\mathcal{X}$ if every point $x\in\mathcal{X}$ can
be approximated to within $\epsilon$ by some point $y\in
N_{\epsilon}$, i.e. $|x-y|\leq \epsilon$. The minimal cardinality of
an $\epsilon$-net of $\mathcal{X}$, if finite, is denoted by
$N(\mathcal{X},\epsilon)$ and is called the covering number of
$\mathcal{X}$.

\begin{lemma}\label{lemma1}
For $\epsilon>0$, there exists an $\epsilon$-net of $\mathcal{F}$,
denoted by $\mathcal{F}_{\epsilon}$, such that
$||\mathcal{F}_{\epsilon}||_0\leq
\left\{\frac{(2+\epsilon)ep}{\epsilon k}\right\}^k$ and
\begin{align}\label{eq:lemma}
\bigcap_{v\in \mathcal{F}_{\epsilon}}\{w\in\mathbb{R}^p:w'v\leq
(1-\epsilon)\sqrt{t}\}\subseteq  \mathcal{A}(t;k)\subseteq
\bigcap_{v\in \mathcal{F}_{\epsilon}}\{w\in\mathbb{R}^p:w'v\leq
\sqrt{t}\}. \end{align}
\end{lemma}

\begin{proof}[Proof of Lemma \ref{lemma1}]
For the unit sphere $\mathbb{S}^{k-1}$ equipped with the Euclidean
metric, it is well-known that the $\epsilon$-covering number
$N(\mathbb{S}^{k-1},\epsilon)\leq (1+2/\epsilon)^k$, see e.g. Lemma
5.2 of Vershynin (2012). Notice that
$$\mathcal{F}=\{v\in \mathbb{S}^{p-1}, ||v||_0\leq k\}=\bigcup_{S\subseteq [p]:||S||_0=k}\{v\in \mathbb{S}^{p-1}: v_{S}\in \mathbb{S}^{k-1}\},$$
where $[p]=\{1,2,\dots,p\}$. Because $\binom{p}{k}\leq (ep/k)^k$, we
have
$$N(\mathcal{F},\epsilon)\leq \binom{p}{k}\left(1+\frac{2}{\epsilon}\right)^k\leq \left\{\frac{(2+\epsilon)ep}{\epsilon k}\right\}^k.$$
Recall that
$$\mathcal{A}(t;k)=\bigcap_{v\in \mathcal{F}}\{w\in\mathbb{R}^p:w'v\leq \sqrt{t}\}.$$
Let $\mathcal{F}_{\epsilon}$ be an $\epsilon$-net of $\mathcal{F}$
with cardinality $N(\mathcal{F},\epsilon)$, and
$A_1(t):=A_1(t;\epsilon)=\bigcap_{v\in
\mathcal{F}_{\epsilon}}\{w\in\mathbb{R}^p:w'v\leq
(1-\epsilon)\sqrt{t}\}$. It is easy to see that
$$\mathcal{A}(t;k)\subseteq
\bigcap_{v\in \mathcal{F}_{\epsilon}}\{w\in\mathbb{R}^p:w'v\leq
\sqrt{t}\}.$$
For any $v\in\mathcal{F}$, we can find
$v_0\in\mathcal{F}_{\epsilon}$ such that $|v-v_0|\leq \epsilon.$
Thus for $w\in A_1(t)$, we have
\begin{align*}
w'v=w'v_0+|v-v_0|\frac{w'(v-v_0)}{|v-v_0|}\leq
(1-\epsilon)\sqrt{t}+\epsilon \max_{v_1\in\mathcal{F}}w'v_1.
\end{align*}
Taking maximum over $v\in\mathcal{F}$, we obtain
$\max_{v\in\mathcal{F}}w'v\leq (1-\epsilon)\sqrt{t}+\epsilon
\max_{v_1\in\mathcal{F}}w'v_1$, which implies that
$\max_{v\in\mathcal{F}}w'v\leq \sqrt{t}$ and thus $w\in
\mathcal{A}(t;k).$
\end{proof}

\subsection{Proofs of the main results}\label{appendix}

\begin{proof}[Proof of Theorem \ref{thm}]
The triangle inequality yields that
\begin{align*}
&\sup_{t\geq 0}\left|P\left(T_{fe,n}^*(k)\leq
t\bigg|X_1^n\right)-P\left(T_{fe,n}(k)\leq t\right)\right|
\\
\leq &\sup_{t\geq 0}\left|P\left(T_{n}(k)\leq
t\right)-P\left(T_{fe,n}(k)\leq t\right)\right|
\\&+\sup_{t\geq 0}\left|P\left(T_{n}(k)\leq t\right)-P\left(T_{fe,n}^*(k)\leq
t\bigg|X_1^n\right)\right|:=\rho_{1,n}+\rho_{2,n}.
\end{align*}
We bound $\rho_{1,n}$ and $\rho_{2,n}$ in Step 1 and Step 2 respectively.

\textbf{Step 1 (bounding $\rho_{1,n}$)}: Let
$\widehat{\xi}_j=\frac{|\widehat{z}_{j}|}{\sqrt{\widehat{\gamma}_{jj}}}$
and $\widehat{\xi}_{(j)}$ be the order statistic such that
$$\widehat{\xi}_{(1)}\geq \widehat{\xi}_{(2)}\geq \cdots\geq \widehat{\xi}_{(p)}.$$
Similarly we can define $\xi_j$ and $\xi_{(j)}$ in the same way as $\widehat{\xi}_j$ and $\widehat{\xi}_{(j)}$ by replacing $\widehat{\Gamma}$ with the precision matrix $\Gamma$. We have
\begin{align*}
|\widehat{\xi}_j-\xi_j|\leq \left| \frac{|\widehat{z}_{j}|}{\sqrt{\widehat{\gamma}_{jj}}}-\frac{|z_{j}|}{\sqrt{\widehat{\gamma}_{jj}}}\right|
+\left| \frac{|z_{j}|}{\sqrt{\widehat{\gamma}_{jj}}}-\frac{|z_{j}|}{\sqrt{\gamma_{jj}}}\right|
\leq \frac{|z_j-\widehat{z}_j|}{\sqrt{\widehat{\gamma}_{jj}}}+\left|\frac{\sqrt{\gamma_{jj}}-\sqrt{\widehat{\gamma}_{jj}}}{\sqrt{\gamma_{jj}}\sqrt{\widehat{\gamma}_{jj}}}\right||z_j|:=I_1+I_2.
\end{align*}
By Proposition \ref{prop0}, we have $\max_{1\leq j\leq p}|\widehat{\Gamma}_j-\Gamma_j|_{1}=O_p(d\sqrt{\log(p)/n})$ and
$\sup_{1\leq j\leq p}|\gamma_{jj}-\widehat{\gamma}_{jj}|=O_p(\sqrt{d\log(p)/n})$. Also note that $c_1<\min_{1\leq j\leq p}\gamma_{jj}
\leq \max_{1\leq j\leq p}\gamma_{jj}<c_2$ for some constants $0<c_1\leq c_2<\infty.$ Together with the fact that $|\bar{X}|_{\infty}=O_p(\sqrt{\log(p)/n})$, we deduce
$$\sup_{1\leq j\leq p}|z_j-\widehat{z}_j|\leq \max_{1\leq j\leq p}|\widehat{\Gamma}_j-\Gamma_j|_{1}|\bar{X}|_{\infty}=O_p\left(d\log(p)/n\right),$$
and
\begin{align}\label{eq:sqrt}
\sup_{1\leq j\leq
p}|\sqrt{\gamma_{jj}}-\sqrt{\widehat{\gamma}_{jj}}|=\sup_{1\leq
j\leq
p}\left|\frac{\gamma_{jj}-\widehat{\gamma}_{jj}}{\sqrt{\gamma_{jj}}+\sqrt{\widehat{\gamma}_{jj}}}\right|
=O_p\left(\sqrt{d\log(p)/n}\right).
\end{align}
As $\sup_{1\leq j\leq p}|z_j|=O_p(\sqrt{\log(p)/n})$, we obtain
\begin{align*}
\sup_{1\leq j\leq p}|\widehat{\xi}_j-\xi_j|=O_p\left(d\log(p)/n\right),
\end{align*}
and
$$\sup_{1\leq j\leq p}|\widehat{\xi}_j^2-\xi_j^2|=O_p\left(d\log(p)/n\right)\sup_{1\leq j\leq p}|\widehat{\xi}_j+\xi_j|=O_p\left(d(\log(p)/n)^{3/2}\right).$$
Thus we deduce that
\begin{align*}
\left|T_n(k)-T_{fe,n}(k)\right|\leq &\max_{1\leq
j_1<j_2<\cdots<j_k\leq
p}n\sum^{k}_{l=1}\left|\xi_{j_l}^2-\widehat{\xi}_{j_l}^2\right|
\\ \leq&
nk\max_{1\leq j\leq p}|\widehat{\xi}_j^2-\xi_j^2|=O_p\left(kd(\log(p))^{3/2}/\sqrt{n}\right).
\end{align*}
By the assumption $k^2d(\log(np))^{5/2}/\sqrt{n}=o(1)$, we can pick
$\zeta_1$ and $\zeta_2$ such that
$$P(\left|T_n(k)-T_{fe,n}(k)\right|\geq \zeta_1)\leq \zeta_2,$$
where $\zeta_1k\log(np)=o(1)$ and $\zeta_2=o(1)$. Define the event
$\mathcal{B}=\{\left|T_n(k)-T_{fe,n}(k)\right|<\zeta_1\}$. Then we
have
\begin{align*}
&|P(T_n(k)\leq t)-P(T_{fe,n}(k)\leq t)|
\\ \leq & P(T_n(k)\leq t,T_{fe,n}(k)>t)+P(T_n(k)>t,T_{fe,n}(k)\leq t)
\\ \leq & P(T_n(k)\leq t,T_{fe,n}(k)>t,\mathcal{B})+P(T_n(k)>t,T_{fe,n}(k)\leq t,\mathcal{B})+2\zeta_2
\\ \leq & P(t-\zeta_1<T_n(k)\leq t)+P(t+\zeta_1\geq T_n(k)>t)+2\zeta_2.
\end{align*}
Let $V_i=\text{diag}^{-1/2}(\Gamma)\Gamma X_i$ and
$V=\sum^{n}_{i=1}V_i/\sqrt{n}$. Notice that $\left\{T_n(k)\leq
t\right\}=\left\{V\in \mathcal{A}(t;k)\right\}=\{\max_{v\in\mathcal{F}}v'V\leq \sqrt{t}\}.$ By Lemma \ref{lemma1}, we can find an $\epsilon$-net $\mathcal{F}_{\epsilon}$ of
$\mathcal{F}$ such that $||\mathcal{F}_{\epsilon}||_0\leq \{(2+\epsilon)ep/(\epsilon k)\}^k$ and
$$A_1(t):=\bigcap_{v\in\mathcal{F}_{\epsilon}}\{v'V\leq (1-\epsilon)\sqrt{t}\}\subseteq \mathcal{A}(t;k)\subseteq A_2(t):=\bigcap_{v\in\mathcal{F}_{\epsilon}}\{v'V\leq \sqrt{t}\}.$$
We set $\epsilon=1/n$ throughout the following arguments. Notice that
\begin{align*}
&P(t-\zeta_1<T_n(k)\leq t)
\\=&P(T_n(k)\leq t)-P(T_n(k)\leq t-\zeta_1)
\\ \leq &P(\max_{v\in\mathcal{F}_{\epsilon}}v'V\leq \sqrt{t})-P(\max_{v\in\mathcal{F}_{\epsilon}}v'V\leq (1-\epsilon)\sqrt{t-\zeta_1})
\\ \leq &P((1-\epsilon)(\sqrt{t}-\sqrt{\zeta_1})\leq \max_{v\in\mathcal{F}_{\epsilon}}v'V\leq \sqrt{t})
\\ \leq &P((1-\epsilon)(\sqrt{t}-\sqrt{\zeta_1})\leq \max_{v\in\mathcal{F}_{\epsilon}}v'V\leq (1-\epsilon)\sqrt{t})+P((1-\epsilon)\sqrt{t}<\max_{v\in\mathcal{F}_{\epsilon}}v'V\leq \sqrt{t})
\\ :=&I_{1}+I_2.
\end{align*}
Because $\phi(\Gamma;k)>c>0$, we have
$\var(\sum^{n}_{i=1}v'V_i/\sqrt{n})>c'$ for all
$v\in \mathcal{F}$ and some constant $c'>0.$ By the Nazarov
inequality [see Lemma A.1 in Chernozhukov et al. (2015) and Nazarov
(2003)], we have
\begin{align}\label{nazarov}
I_1\leq & C\sqrt{\zeta_1k\log(np/k)}=o(1).
\end{align}
To deal with $I_2$, we note when $t\leq k^3\{\log(np/k)\}^2$, $\epsilon\sqrt{t}\leq k^{3/2}\log(np/k)/n$. Again by the Nazarov's inequality, we have
$$I_2\leq P(\sqrt{t}-k^{3/2}\log(np/k)/n<\max_{v\in\mathcal{F}_{\epsilon}}v'V\leq \sqrt{t})\leq k^2\log(np/k)\sqrt{\log(np/k)}/n=o(1).$$
When $t>k^3\{\log(np/k)\}^2$, we have
\begin{align*}
I_2\leq P((1-\epsilon)\sqrt{t}\leq \max_{v\in\mathcal{F}_{\epsilon}}v'V)\leq \frac{\E\max_{v\in\mathcal{F}_{\epsilon}}v'V}{(1-\epsilon)k^{3/2}\log(np/k)}.
\end{align*}
By Lemma 7.4 in Fan et al. (2015), we have $\E\max_{v\in\mathcal{F}_{\epsilon}}v'V\leq C\sqrt{k\log(np/k)}.$
It thus implies that
$$I_2 \leq \frac{C\sqrt{k\log(np/k)}}{(1-\epsilon)k^{3/2}\log(np/k)}=o(1).$$
Summarizing the above derivations, we have $\rho_{1,n}=o(1)$.

\textbf{Step 2 (bounding $\rho_{2,n}$)}: Define
$\widehat{V}^*=\text{diag}^{-1/2}(\widehat{\Gamma})\widehat{\Gamma}\sum^{n}_{i=1}(X_i-\bar{X})e_i/\sqrt{n}$
with $e_i\sim^{i.i.d} N(0,1)$, where $e_i$'s are independent of
$X^n_1.$ Further define
\begin{align*}
\bar{\rho}=\max\{|P(V\in A_1(t))-P(\widehat{V}^*\in
A_{1}(t)|X^n_1)|,|P(V\in A_2(t))-P(\widehat{V}^*\in
A_2(t)|X^n_1)|\}.
\end{align*}
Using similar arguments in Step 1, we have
\begin{align*}
P(\widehat{V}^*\in\mathcal{A}(t;k) |X^n_1)\leq & P(\widehat{V}^*\in A_2(t) |X^n_1)\leq P(V\in A_2(t))+\bar{\rho}
\\ \leq & P(V\in A_1(t))+\bar{\rho}+o(1)
\\ \leq & P(V\in \mathcal{A}(t;k))+\bar{\rho}+o(1).
\end{align*}
Similarly we have $P(\widehat{V}^*\in\mathcal{A}(t;k)|X^n_1)\geq P(V\in \mathcal{A}(t;k))-\bar{\rho}-o(1)$. Together, we obtain
$$|P(\widehat{V}^*\in \mathcal{A}(t;k)|X^n_1)-P(V\in \mathcal{A}(t;k))|\leq \bar{\rho}+o(1).$$

Let $D=\text{diag}^{-1/2}(\Gamma)\Gamma\text{diag}^{-1/2}(\Gamma)$
and
$\widehat{D}=\text{diag}^{-1/2}(\widehat{\Gamma})\widehat{\Gamma}\widehat{\Sigma}\widehat{\Gamma}'\text{diag}^{-1/2}(\widehat{\Gamma}),$
where
$\widehat{\Sigma}=\sum^{n}_{i=1}(X_i-\bar{X})(X_i-\bar{X})'/n$.
Define $\Delta_n=\max_{u,v\in \mathcal{F}}|u(\widehat{D}-D)v|.$
Notice that $V\sim N(0,D)$ and $\widehat{V}^*|X^n_1 \sim
N(0,\widehat{D})$. To bound $\bar{\rho}$, we note that by equation
(49) in Chernozhukov et al. (2015),
\begin{align*}
&\sup_{t\geq 0}|P(V\in A_1(t))-P(\widehat{V}^*\in A_1(t)|X^n_1)|
\\=&\sup_{t\geq 0}|P(\max_{v\in\mathcal{F}_{\epsilon}}u'V \leq \sqrt{t})-P(\max_{u\in\mathcal{F}_{\epsilon}}u'\widehat{V}^*\leq \sqrt{t}|X^n_1)|
\\ \leq& C\Delta_n^{1/3}(k\log(np/k))^{2/3}.
\end{align*}
and similarly $|P(V\in A_2(t))-P(\widehat{V}^*\in
A_2(t)|X^n_1)|\leq C\Delta_n^{1/3}(k\log(np/k))^{2/3}.$
Therefore we get
$$\rho_{2,n}\leq C\Delta_n^{1/3}(\log(p))^{2/3}+o(1).$$

\textbf{Step 3}: Finally we bound $\Delta_n$. Note that for
any $u,v\in \mathcal{F}$,
\begin{align*}
&|u'(\widehat{D}-D)v|
\\ \leq& |u'(\widehat{D}-\text{diag}^{-1/2}(\widehat{\Gamma})\widehat{\Gamma}'\text{diag}^{-1/2}(\widehat{\Gamma}))v|+
|u'(\text{diag}^{-1/2}(\widehat{\Gamma})\widehat{\Gamma}'\text{diag}^{-1/2}(\widehat{\Gamma})-
\text{diag}^{-1/2}(\widehat{\Gamma})\Gamma\text{diag}^{-1/2}(\widehat{\Gamma}))v|
\\&+|u'(\text{diag}^{-1/2}(\widehat{\Gamma})\Gamma\text{diag}^{-1/2}(\widehat{\Gamma})-D)v|:=J_1+J_2+J_3.
\end{align*}
For the first term, we have
\begin{align*}
J_1=& |u'\text{diag}^{-1/2}(\widehat{\Gamma})(\widehat{\Gamma}\widehat{\Sigma}'\widehat{\Gamma}-\widehat{\Gamma}')\text{diag}^{-1/2}(\widehat{\Gamma})v|
\\ \leq &
|\text{diag}^{-1/2}(\widehat{\Gamma})u|_1|(\widehat{\Gamma}\widehat{\Sigma}'\widehat{\Gamma}-\widehat{\Gamma}')\text{diag}^{-1/2}(\widehat{\Gamma})v|_{\infty}
\\ \leq & |\text{diag}^{-1/2}(\widehat{\Gamma})u|_1|\text{diag}^{-1/2}(\widehat{\Gamma})v|_1||\widehat{\Gamma}\widehat{\Sigma}\widehat{\Gamma}'-\widehat{\Gamma}'||_{\infty}
\\ \leq & k|\text{diag}^{-1/2}(\widehat{\Gamma})u|_2|\text{diag}^{-1/2}(\widehat{\Gamma})v|_2||\widehat{\Gamma}\widehat{\Sigma}\widehat{\Gamma}'-\widehat{\Gamma}'||_{\infty}
=O_p(k\sqrt{d\log(p)/n}),
\end{align*}
where we have used Proposition \ref{prop0}. To handle the second term, note that
\begin{align*}
J_2=& |v'\text{diag}^{-1/2}(\widehat{\Gamma})(\widehat{\Gamma}-\Gamma)\text{diag}^{-1/2}(\widehat{\Gamma})u|
\\ \leq &  |\text{diag}^{-1/2}(\widehat{\Gamma})v|_1|(\widehat{\Gamma}-\Gamma)\text{diag}^{-1/2}(\widehat{\Gamma})u|_{\infty}
\\ \leq &  \sqrt{k}|\text{diag}^{-1/2}(\widehat{\Gamma})v|_2|\text{diag}^{-1/2}(\widehat{\Gamma})u|_{\infty}\max_{1\leq j\leq p}|\widehat{\Gamma}_j-\Gamma_j|_1
=O_p(d\sqrt{k\log(p)/n}).
\end{align*}
Finally, we have
\begin{align*}
J_3\leq& |u'(\text{diag}^{-1/2}(\widehat{\Gamma})\Gamma\text{diag}^{-1/2}(\widehat{\Gamma})-\text{diag}^{-1/2}(\Gamma)\Gamma\text{diag}^{-1/2}(\widehat{\Gamma}))v|+
|u'(\text{diag}^{-1/2}(\Gamma)\Gamma\text{diag}^{-1/2}(\widehat{\Gamma})-D)v|
\\ \leq & ||\Gamma||_2||\text{diag}^{-1/2}(\widehat{\Gamma})||_2\max_{1\leq j\leq p}|1/\sqrt{\gamma_{jj}}-1/\sqrt{\widehat{\gamma}_{jj}}|
+||\Gamma||_2||\text{diag}^{-1/2}(\Gamma)||_2\max_{1\leq j\leq p}|1/\sqrt{\gamma_{jj}}-1/\sqrt{\widehat{\gamma}_{jj}}|
\\=&\sqrt{d\log(p)/n}.
\end{align*}
Under the assumption that $k^2d(\log(np))^{5/2}/\sqrt{n}=o(1)$, we
have $(\log(np))^2J_i=o_p(1)$ for $1\leq i\leq 3$. Therefore we get
$(\log(np))^2\Delta_n=o_p(1)$, which implies that
$\rho_{2,n}=o_p(1)$. The proof is thus completed by combining Steps
1-3.
\end{proof}

\begin{proof}[Proof of Theorem \ref{thm2}]
We first note that by the triangle inequality,
\begin{align*}
&\sup_{t_M\geq t_{M-1}\geq \cdots \geq t_1\geq 0}\left|P\left(\bigcap^{M}_{j=1}\{T_{fe,n}^*(j)\leq
t_j\}\bigg|X_1^n\right)-P\left(\bigcap^M_{j=1}\{T_{fe,n}(j)\leq
t_j\}\right)\right|
\\ \leq &\sup_{t_M\geq t_{M-1}\geq \cdots \geq t_1\geq 0}\left|P\left(\bigcap^{M}_{j=1}\{T_n(j)\leq
t_j\}\right)-P\left(\bigcap^M_{j=1}\{T_{fe,n}(j)\leq
t_j\}\right)\right|
\\&+\sup_{t_M\geq t_{M-1}\geq \cdots \geq t_1\geq 0}\left|P\left(\bigcap^{M}_{j=1}\{T_n(j)\leq
t_j\}\right)-P\left(\bigcap^M_{j=1}\{T_{fe,n}^*(j)\leq
t_j\}\bigg|X_1^n\right)\right|:=\varrho_{1,n}+\varrho_{2,n}.
\end{align*}

\textbf{Step 1 (bounding $\varrho_{1,n}$)}: Following the proof of
Theorem \ref{thm}, we have for any $1\leq j\leq M$,
\begin{align*}
\max_{1\leq j\leq M}\left|T_n(j)-T_{fe,n}(j)\right|\leq nM\max_{1\leq j\leq p}|\widehat{\xi}_j^2-\xi_j^2|
=O_p\left(Md(\log(p))^{3/2}/\sqrt{n}\right).
\end{align*}
Under the assumption that $M^4d(\log(np))^{5/2}/\sqrt{n}=o(1)$, one
can pick $\zeta$ such that $P(\max_{1\leq j\leq
M}\left|T_n(j)-T_{fe,n}(j)\right|>\zeta)=o(1)$ and $\zeta M^3\log
(np)=o(1)$. Define $\mathcal{B}=\{\max_{1\leq j\leq
M}\left|T_n(j)-T_{fe,n}(j)\right|\leq \zeta\}$. We note that
\begin{align*}
&\left|P\left(\bigcap^{M}_{j=1}\{T_n(j)\leq
t_j\}\right)-P\left(\bigcap^{M}_{j=1}\{T_{fe,n}(j)\leq
t_j\}\right)\right|
\\ \leq & P\left(\bigcap^{M}_{j=1}\{T_n(j)\leq
t_j\},\bigcup^{M}_{j=1}\{T_{fe,n}(j)>t_j\},\mathcal{B}\right)
\\&+P\left(\bigcup^{M}_{j=1}\{T_n(j)>
t_j\},\bigcap^{M}_{j=1}\{T_{fe,n}(j)\leq
t_j\},\mathcal{B}\right)+o(1)
\\ \leq & P\left(\bigcup^{M}_{j=1}\{t_j-\zeta<T_n(j)\leq
t_j\}\right)+P\left(\bigcup^{M}_{j=1}\{t_j< T_n(j)\leq
t_j+\zeta\}\right)+o(1).
\end{align*}
By the arguments in the proof of Theorem \ref{thm}, we have
$P\left(t_j-\zeta<T_n(j)\leq t_j\right)=o(1)$. A careful inspection
of the proof shows that
\begin{align*}
&\max_{1\leq j\leq M}\max\{P\left(t_j-\zeta<T_n(j)\leq
t_j\right),P\left(t_j<T_n(j)\leq t_j+\zeta\right)\}
\\ \leq&
C\left\{\sqrt{\zeta
M\log(np)}+M^2(\log(np))^{3/2}/n+1/(M\sqrt{\log(np/M)})\right\},
\end{align*}
where the uniformity over $1\leq j\leq M$ is due to the fact that the constant $C$ in (\ref{nazarov}) is independent of $t$. By the union bound, we deduce that
\begin{align*}
&\left|P\left(\bigcap^{M}_{j=1}\{T_n(j)\leq
t_j\}\right)-P\left(\bigcap^{M}_{j=1}\{T_{fe,n}(j)\leq
t_j\}\right)\right|
\\ \leq &\sum^{M}_{j=1}\left(P\left(t_j-\zeta<T_n(j)\leq
t_j\right)+P\left(t_j< T_n(j)\leq t_j+\zeta\right)\right)+o(1)
\\ \leq &M\max_{1\leq j\leq M}\left(P\left(t_j-\zeta<T_n(j)\leq
t_j\right)+P\left(t_j< T_n(j)\leq t_j+\zeta\right)\right)+o(1)
\\ \leq & C\left(\sqrt{\zeta
M^3\log(np)}+M^3(\log(np))^{3/2}/n+1/\sqrt{\log(np/M)}\right)+o(1)=o(1).
\end{align*}

\textbf{Step 2 (bounding $\varrho_{2,n}$)}: For
$\mathbf{t}=(t_1,\dots,t_M)$, define
\begin{align*}
\mathcal{A}(\mathbf{t})=\bigcap^{M}_{j=1}\mathcal{A}(t_j;j)=\bigcap^M_{j=1}\bigcap_{S\subseteq
[p],||S||_0=j}\{w\in\mathbb{R}^p: w_{S}'w_{S}\leq t_j\}.
\end{align*}
It is easy to see that
\begin{align*}
\bigcap^{M}_{j=1}\{T_n(j)\leq t_j\}=\{V\in
\mathcal{A}(\mathbf{t})\},\quad \bigcap^{M}_{j=1}\{T_{fe,n}^*(j)\leq
t_j\}=\{\widehat{V}^*\in \mathcal{A}(\mathbf{t})\}.
\end{align*}
By Lemma \ref{lemma1}, we know for any fixed
$\mathbf{t}$,
$$\mathbf{A}_1(\mathbf{t}):=\bigcap_{j=1}^MA_1(t_j)\subseteq \mathcal{A}(\mathbf{t})\subseteq \mathbf{A}_2(\mathbf{t}):=\bigcap_{j=1}^MA_2(t_j),$$
where $A_1(t_j)=\bigcap_{v\in
\mathcal{F}_{\epsilon}(j)}\{w\in\mathbb{R}^p:w'v\leq
(1-\epsilon)\sqrt{t_j}\}$ and $A_2(t_j)=\bigcap_{v\in
\mathcal{F}_{\epsilon}(j)}\{w\in\mathbb{R}^p:w'v\leq \sqrt{t_j}\}$
with $\epsilon=1/n$ and $\mathcal{F}_{\epsilon}(j)$ being an
$\epsilon$-net for
$\mathcal{F}(j):=\{v\in\mathbb{S}^{p-1}:||v||_0\leq j\}$. Note that
$\mathbf{A}_1(\mathbf{t})$ and $\mathbf{A}_2(\mathbf{t})$ are both
intersections of no more than $M\{(2+1/n)epn\}^M$ half spaces. Thus
following the arguments in Steps 2 and 3 of the proof of Theorem
\ref{thm}, we can show that $\varrho_{2,n}=o_p(1)$, which completes
our proof.
\end{proof}

\begin{proof}[Proof of Theorem \ref{thm:consis}]
The proof contains two steps. In the first step, we establish the
consistency of the infeasible test $T_n(k)$, while in the second
step we further show that the estimation effect caused by replacing
$\Gamma$ with $\widehat{\Gamma}$ is asymptotically negligible.
\\
\textbf{Step 1:} Consider the infeasible test $T_n(k)=\max_{1\leq
j_1<j_2<\cdots<j_k\leq
n}\sum^{k}_{l=1}\frac{nz_{j_l}^2}{\gamma_{j_l,j_l}}.$ Define
$\widetilde{\theta}=(\widetilde{\theta}_1,\dots,\widetilde{\theta}_p)'$
with $\widetilde{\theta}_j=(\Gamma\theta)_j/\sqrt{\gamma_{jj}}$, and
$q_j=\frac{z_j-(\Gamma\theta)_j}{\sqrt{\gamma_{jj}}}.$ Note that
$\frac{z_{j}^2}{\gamma_{jj}}=\left(q_j+\widetilde{\theta}_j\right)^2.$
Also by Lemma 3 of Cai et al. (2014), we have for any $2r<a<1-2r$,
\begin{equation}\label{eq:sig}
P\left(\max_{j\in
H}|\widetilde{\theta}_j-\sqrt{\gamma_{jj}}\theta_{j}|=O(p^{r-a/2})\max_{j\in
H}|\theta_j|\right)\rightarrow 1,
\end{equation}
where $H$ denotes the support of $\theta$. Suppose $\sqrt{\gamma_{j_1^*,j_1^*}}\theta_{j_1^*}\geq
\sqrt{\gamma_{j_2^*,j_2^*}}\theta_{j_2^*}\geq \cdots\geq
\sqrt{\gamma_{j_p^*,j_p^*}}\theta_{j_p^*}.$ We deduce that with
probability tending to one,
\begin{align*}
T_n(k)=&\max_{1\leq j_1<j_2<\cdots<j_k\leq
n}\sum^{k}_{l=1}n\left(q_{j_l}^2+\widetilde{\theta}_{j_l}^2+2q_{j_l}\widetilde{\theta}_{j_l}\right)
\\ \geq & \left(n\sum^{k}_{l=1}\gamma_{j_l^*,j_l^*}\theta_{j_l^*}^2+n\sum^{k}_{l=1}q_{j_l^*}^2+2n\sum^{k}_{l=1}q_{j_l^*}
\sqrt{\gamma_{j_l^*j_l^*}}\theta_{j_l^*}\right)(1+o(1))
\\ \geq & \left\{n\sum^{k}_{l=1}\gamma_{j_l^*,j_l^*}\theta_{j_l^*}^2+n\sum^{k}_{l=1}q_{j_l^*}^2-2n\left(\sum^{k}_{l=1}q_{j_l^*}^2\right)^{1/2}
\left(\sum^{k}_{l=1}\gamma_{j_l^*j_l^*}\theta_{j_l^*}^2\right)^{1/2}\right\}(1+o(1)).
\end{align*}
We claim that $n\sum^{k}_{l=1}q_{j_l}^2=O_p(k)$ for any $\widetilde{S}:=\{j_1,\dots,j_k\}\subseteq [p].$ Note that
\begin{align*}
\var\left(\sum^{k}_{l=1}(nq_{j_l}^2-1)/\sqrt{k}\right)=&\frac{2}{k}\sum^{k}_{i,l=1}\gamma_{j_i,j_l}^2/(\gamma_{j_i,j_i}\gamma_{j_l,j_l})
\leq \frac{C}{k}\sum^{k}_{i,l=1}\gamma_{j_i,j_l}^2=\frac{C}{k}\text{Tr}(\Gamma_{\widetilde{S},\widetilde{S}}^2)
\leq C/\lambda_{\min}^2\leq C,
\end{align*}
where $\text{Tr}(\cdot)$ denotes the trace of a matrix.
Therefore conditional on $\widetilde{S}$, $n\sum^{k}_{l=1}q_{j_l}^2=O_p(k)$. As $q_j$ is independent of the non-zero locations of $\theta$, $n\sum^{k}_{l=1}q_{j_l^*}^2=O_p(k)$.

By the assumption that
$\sum^{k}_{l=1}\gamma_{j_l^*,j_l^*}\theta_{j_l^*}^2\geq
(2k+\epsilon)\log(p)/n$ and
$n\sum^{k}_{l=1}q_{j_l^*}^2=O_p(k)$, we obtain
\begin{align}\label{power-eq1}
T_n(k)\geq
\left\{(2k+\epsilon)\log(p)+O_p(k)-O_p(\sqrt{(2k+\epsilon)k\log(p)})\right\}(1+o(1)).
\end{align}
Under Assumption \ref{ass4a}, we have by Lemma 6 of Cai et al. (2014),
$$\max_{1\leq j_1<j_2<\cdots<j_k\leq p}n\sum^{k}_{l=1}q_{j_l}^2\leq kn\max_{1\leq j\leq p}q_{j}^2=\{2k\log(p)-k\log\log(p)\}+O_p(1).$$
As shown in Theorem \ref{thm}, the bootstrap statistic
$T_{fe,n}^*(k)$ imitates the sampling distribution of $\max_{1\leq
j_1<j_2\cdots<j_k\leq n}\sum^{k}_{l=1}nq_{j_l}^2$. By Step 2 in the
proof of Theorem \ref{thm}, we have
\begin{align}\label{power-eq2}
c_{\alpha}^*(k)\leq \{2k\log p-k\log\log(p)\}+O_p(1).
\end{align}
Combining (\ref{power-eq1}) and (\ref{power-eq2}), we get
$$P(T_n(k)>c_{\alpha}^*(k))\rightarrow 1.$$
\\
\textbf{Step 2:} Next we quantify the difference between $T_n(k)$
and $T_{fe,n}(k)$. Note that
\begin{align*}
\left|T_n(k)-T_{fe,n}(k)\right|\leq& \max_{1\leq
j_1<j_2<\cdots<j_k\leq
p}n\sum^{k}_{l=1}\left|\xi_{j_l}^2-\widehat{\xi}_{j_l}^2\right| \leq
nk\max_{1\leq j\leq p}|\widehat{\xi}_j^2-\xi_j^2|.
\end{align*}
We define $\widehat{\theta}_j$ and $\widehat{q}_j$ by replacing
$\Gamma$ with $\widehat{\Gamma}$ in $\widetilde{\theta}_j$ and
$q_j$. Simple algebra yields that
\begin{align}\label{eq:xi}
\widehat{\xi}_j^2-\xi_j^2=&(\widehat{q}_j+\widehat{\theta}_j)^2-(q_j+\widetilde{\theta}_j)^2
=(\widehat{q}_j-q_j+\widehat{\theta}_j-\widetilde{\theta}_j)(q_j+\widetilde{\theta}_j+\widehat{q}_j+\widehat{\theta}_j).
\end{align}
Using similar argument in Step 1 of the proof of Theorem \ref{thm},
we obtain
\begin{align*}
\max_{1\leq j\leq p}|q_j-\widehat{q}_j|=O_p\left(d\log(p)/n\right),
\end{align*}
and
\begin{align*}
\max_{1\leq j\leq
p}|\widetilde{\theta}_j-\widehat{\theta}_j|=O_p(d\sqrt{\log(p)/n})\max_{1\leq
j\leq p}|\theta_j| +O_p(\sqrt{d\log(p)/n})\max_{1\leq j\leq
p}|\widetilde{\theta}_j|,
\end{align*}
where we have used (\ref{e2}), (\ref{eq:sqrt}), the triangle inequality and the fact that $\max_{1\leq j\leq
p}|(\Gamma\theta)_j|\leq C\max_{1\leq j\leq
p}|\widetilde{\theta}_j|$ for some $C>0.$ By (\ref{eq:sig}), we have
with probability tending to one, $\max_{j\in
H}|\widetilde{\theta}_j|\leq (C'+o(1))\max_{j\in H}|\theta_j|$ for
$C'>0.$ Define the event
$$A=\left\{\max_{1\leq j\leq p}|\theta_j|<C_0\sqrt{\log(p)/n}\right\},$$
for some large enough constant $C_0>0.$ On $A$, we have $\max_{1\leq
j\leq p}|\widetilde{\theta}_j-\widehat{\theta}_j|=O_p(d\log(p)/n).$
In view of (\ref{eq:xi}), we have on the event $A$,
$$|\widehat{\xi}_j^2-\xi_j^2|=O_p(d(\log(p)/n)^{3/2})$$
which implies that
$$\left|T_n(k)-T_{fe,n}(k)\right|\leq O_p(kd(\log(p))^{3/2}/\sqrt{n}).$$
For any $\epsilon>0$, pick $C''$ such that
$$P(\left|T_n(k)-T_{fe,n}(k)\right|\leq C^{''}kd(\log(p))^{3/2}/\sqrt{n}|A)\geq 1-\epsilon.$$
Thus we have
\begin{equation}
\begin{split}\label{eq:f1}
P(T_{fe,n}(k)>c_{\alpha}^*(k)|A)\geq &
P(T_{n}(k)>c_{\alpha}^*(k)+|T_n(k)-T_{fe,n}(k)||A)
\\\geq &P(T_{n}(k)>c_{\alpha}^*(k)+C^{''}kd(\log(p))^{3/2}/\sqrt{n}|A)-\epsilon.
\end{split}
\end{equation}
Recall for $\zeta>0$ with
$\zeta k\log(np)=o(1)$, we have
$$P(t\leq T_n(k)\leq t+\zeta)=o(1).$$
See Step 1 in the proof of Theorem \ref{thm}. Under the assumption
that $k^2d(\log(np))^{5/2}/\sqrt{n}=o(1)$, we have
\begin{align*}
&P(T_{n}(k)>c_{\alpha}^*(k)|A)-P(T_{n}(k)>c_{\alpha}^*(k)+C^{''}kd(\log(p))^{3/2}/\sqrt{n}|A)
\\=&P(c_{\alpha}^*(k)<T_{n}(k)\leq c_{\alpha}^*(k)+C^{''}kd(\log(p))^{3/2}/\sqrt{n}|A)=o(1).
\end{align*}
Together with (\ref{eq:f1}) and the result in Step 1, we obtain
\begin{align*}
P(T_{fe,n}(k)>c_{\alpha}^*(k)|A)\geq
P(T_{n}(k)>c_{\alpha}^*(k)|A)-o(1)-\epsilon\rightarrow 1-\epsilon.
\end{align*}

Suppose $\max_{1\leq j\leq p}|\theta_j|=|\theta_{k_1^*}|$. On $A^c$,
we have for large enough $C_0$,
\begin{align*}
T_n(k)\geq
&n\left(q_{k_1^*}^2+\widetilde{\theta}_{k_1^*}^2+2q_{k_1^*}\widetilde{\theta}_{k_1^*}\right)
\geq C_1\log(p)>2k\log(p)-k\log\log(p),
\end{align*}
which holds with probability tending to one. It implies that
$P(T_n(k)>2k\log(p)-k\log\log(p)|A^c)\rightarrow 1.$ Similar
argument indicates that
\begin{align}\label{eq:f2}
P(T_{fe,n}(k)>c_{\alpha}^*(k)|A^c)\rightarrow 1.
\end{align}
By (\ref{eq:f1}) and (\ref{eq:f2}), we deduce that
$$P(T_{fe,n}(k)>c_{\alpha}^*(k))=P(A)P(T_{fe,n}(k)>c_{\alpha}^*|A)+P(A^c)P(T_{fe,n}(k)>c_{\alpha}^*|A^c)\rightarrow 1-\epsilon P(A).$$
The conclusion follows as $\epsilon$ is arbitrary.

Finally, we show the consistency of $\widetilde{T}_{fe,n}(M)$. As
$\widetilde{T}_{fe,n}^*(M)$ imitates the sampling distribution of
$\widetilde{T}_{fe,n}(M)$ under the null, we know
$$\widetilde{c}_{\alpha}^*(M)=\sqrt{2M}\log(p)(1+o_p(1)).$$
Therefore we have
\begin{align*}
P(\widetilde{T}_{fe,n}(M)>\widetilde{c}_{\alpha}^*(M))\geq
P\left(\sqrt{\frac{1-M/n}{2M}}\left(T_{fe,n}(M)-M\right)
>\widetilde{c}_{\alpha}^*(M)\right)\rightarrow 1.
\end{align*}
\end{proof}

\begin{proof}[Proof of Proposition \ref{prop00}]
Under Assumption \ref{ass5}, we have
$|\bar{X}|_{\infty}=O_p(\sqrt{\log(p)/n})$ and
$|Z|_{\infty}=O_p(\sqrt{\log(p)/n})$. By the arguments in van de
Geer et al. (2014), (\ref{e1}), (\ref{e2}) and (\ref{e3}) still hold
under the sub-gaussian assumption. Therefore using the same
arguments in Step 1 of the proof of Theorem \ref{thm}, we can pick
$\zeta_1$ and $\zeta_2$ such that
$$P(\left|T_n(k)-T_{fe,n}(k)\right|>\zeta_1)\leq \zeta_2,$$
where $\zeta_1k\log(np)=o(1)$ and $\zeta_2=o(1)$. Thus we have
\begin{align*}
&|P(T_n(k)\leq t)-P(T_{fe,n}(k)\leq t)| \leq P(t-\zeta_1<T_n(k)\leq t)+P(t+\zeta_1\geq T_n(k)>t)+2\zeta_2.
\end{align*}
By Lemma \ref{lemma1}, we have
\begin{align*}
&P(t-\zeta_1<T_n(k)\leq t)
\\ \leq &P((1-\epsilon)(\sqrt{t}-\sqrt{\zeta_1})\leq \max_{v\in\mathcal{F}_{\epsilon}}v'V\leq (1-\epsilon)\sqrt{t})+P((1-\epsilon)\sqrt{t}\leq \max_{v\in\mathcal{F}_{\epsilon}}v'V\leq \sqrt{t}).
\end{align*}
Corollary 2.1 in Chernozukov et al. (2015) yields that
\begin{align*}
&P(t-\zeta_1<T_n(k)\leq t)
\\ \leq &P((1-\epsilon)(\sqrt{t}-\sqrt{\zeta_1})\leq \max_{v\in\mathcal{F}_{\epsilon}}v'W\leq (1-\epsilon)\sqrt{t})+P((1-\epsilon)\sqrt{t}\leq \max_{v\in\mathcal{F}_{\epsilon}}v'W\leq \sqrt{t})+c_{n,p,k},
\end{align*}
where $\epsilon=1/n$ and $c_{n,p,k}=C\{k\log(pn/k)\}^{7/6}/n^{1/6}=o(1)$ under the assumption in Proposition \ref{prop00}. Thus following the arguments in the proof of Theorem \ref{thm}, we can show that
\begin{align}\label{p1}
\sup_{t\geq 0}\left|P\left(T_{fe,n}(k)\leq
t\right)-P\left(T_{n}(k)\leq t\right)\right|=o_p(1),
\end{align}
as we only need to deal with the Gaussian vector $W$ and the
arguments are analogous as above. On the other hand, by Lemma
\ref{lemma1} and Corollary 2.1 in Chernozukov et al. (2015), we have
\begin{align}
P\left(T_{n}(k)\leq t\right)-P\left(T^W(k)\leq t\right)\leq& P(\max_{v\in\mathcal{F}_{\epsilon}}v'V\leq \sqrt{t})-P(\max_{v\in\mathcal{F}_{\epsilon}}v'W\leq (1-\epsilon)\sqrt{t}) \nonumber
\\ \leq& P((1-\epsilon)\sqrt{t}\leq \max_{v\in\mathcal{F}_{\epsilon}}v'W\leq \sqrt{t})+c_{n,p,k}. \label{e4}
\end{align}
Similarly $P\left(T^W(k)\leq t\right)-P\left(T_{n}(k)\leq t\right)$ can be bounded above by the same quantity on the RHS of (\ref{e4}). Thus we have
\begin{align}\label{p2}
\sup_{t\geq 0}|P\left(T_{n}(k)\leq t\right)-P\left(T^W(k)\leq t\right)|\leq& \sup_{t\geq 0}|P((1-\epsilon)\sqrt{t}\leq \max_{v\in\mathcal{F}_{\epsilon}}v'W\leq \sqrt{t})|+c_{n,p,k}=o(1),
\end{align}
where we have used the Nazarov inequality and Lemma 7.4 in Fan et al. (2015) to control the term $\sup_{t\geq 0}|P((1-\epsilon)\sqrt{t}\leq \max_{v\in\mathcal{F}_{\epsilon}}v'W\leq \sqrt{t})|$.
The conclusion thus follows from (\ref{p1}) and (\ref{p2}).
\end{proof}

\begin{proof}[Proof of Proposition \ref{prop1}]
The negative log-likelihood (up to a constant) is given by
\begin{align*}
l_n(\theta_1,\theta_2)=\frac{1}{2}\sum^{n_1}_{i=1}(X_i-\theta_1)'\Gamma_1(X_i-\theta_1)+\frac{1}{2}\sum^{n_2}_{i=1}(Y_i-\theta_2)'\Gamma_2(Y_i-\theta_2).
\end{align*}
Under the null, we have $\theta:=\theta_1=\theta_2$. The MLE for $\theta$ is given by
$$\widetilde{\theta}=(n_1\Gamma_1+n_2\Gamma_2)^{-1}\left(\Gamma_1\sum^{n_1}_{i=1}X_i+\Gamma_2\sum^{n_2}_{i=1}Y_i\right).$$
Define
$$(\widetilde{\Delta},\widetilde{\theta}_2)=\underset{\theta_2\in\mathbb{R}^p,\Delta\in \Theta_{a,k}}{\arg\min}l_n(\theta_2+\Delta,\theta_2)$$
and $\widetilde{\theta}_1=\widetilde{\theta}_2+\widetilde{\Delta}.$
Taking the derivative of $l_n(\theta_2+\Delta,\theta_2)$ with
respect to $\theta_2$ and setting it to be zero, we obtain
$$\widetilde{\theta}_2=\widetilde{\theta}-C_1\widetilde{\Delta}.$$
Thus by direct calculation, we have
\begin{align*}
\min_{\theta_1-\theta_2\in
\Theta_{a,k}}l_n(\theta_1,\theta_2)=&\min_{\Delta\in\Theta_{a,k}}\bigg[\frac{n_1}{2}\left\{\Delta'C_2'\Gamma_1C_2\Delta-2\Delta'C_2'\Gamma_1(\bar{X}-\widetilde{\theta})\right\}
\\&+\frac{n_2}{2}\left\{\Delta'C_1'\Gamma_2C_1\Delta+2\Delta'C_1'\Gamma_2(\bar{Y}-\widetilde{\theta})\right\}\bigg]+l_n(\widetilde{\theta},\widetilde{\theta}).
\end{align*}
The log-likelihood ratio test for
testing $H_0'$ against $H_{a,k}'$ is given by
\begin{align*}
LR_n(k)=&2l_n(\widetilde{\theta},\widetilde{\theta})
-2\min_{\theta_1-\theta_2\in
\Theta_{a,k}}l_n(\theta_1,\theta_2)
\\=&\max_{\Delta\in\Theta_{a,k}}\bigg[n_1\left\{2\Delta'C_2'\Gamma_1(\bar{X}-\widetilde{\theta})-\Delta'C_2'\Gamma_1C_2\Delta\right\}
\\&-n_2\left\{\Delta'C_1'\Gamma_2C_1\Delta+2\Delta'C_1'\Gamma_2(\bar{Y}-\widetilde{\theta})\right\}\bigg].
\end{align*}
Recall that $\Omega^{21}=C_2'\Gamma_1C_2$, $\Omega^{12}=C_1'\Gamma_2C_1$,
$\widetilde{X}=C_2'\Gamma_1(\bar{X}-\widetilde{\theta})$ and
$\widetilde{Y}=C_1'\Gamma_2(\bar{Y}-\widetilde{\theta})$. Therefore we have
\begin{align}
LR_n(k)=&\max_{S:||S||_0=k}\max_{\Delta_S\in\mathbb{R}^k}\bigg[n_1\left\{2\Delta'_S\widetilde{X}_S-\Delta'_S\Omega^{21}_{S,S}\Delta_S\right\}-n_2\left\{\Delta'_S\Omega^{12}_{S,S}\Delta_S+2\Delta'_S\widetilde{Y}_S\right\}\bigg] \label{eq:m}
\\=&\max_{S:||S||_0=k}(n_1\widetilde{X}_S-n_2\widetilde{Y}_S)'(n_1\Omega_{S,S}^{21}+n_2\Omega_{S,S}^{12})^{-1}(n_1\widetilde{X}_S-n_2\widetilde{Y}_S), \nonumber
\end{align}
where the maximizer in (\ref{eq:m}) is equal to $\hat{\Delta}_S=(n_1\Omega_{S,S}^{21}+n_2\Omega_{S,S}^{12})^{-1}(n_1\widetilde{X}_S-n_2\widetilde{Y}_S).$
\end{proof}

\begin{figure}[h]
\centering
\includegraphics[height=5.3cm,width=5.8cm]{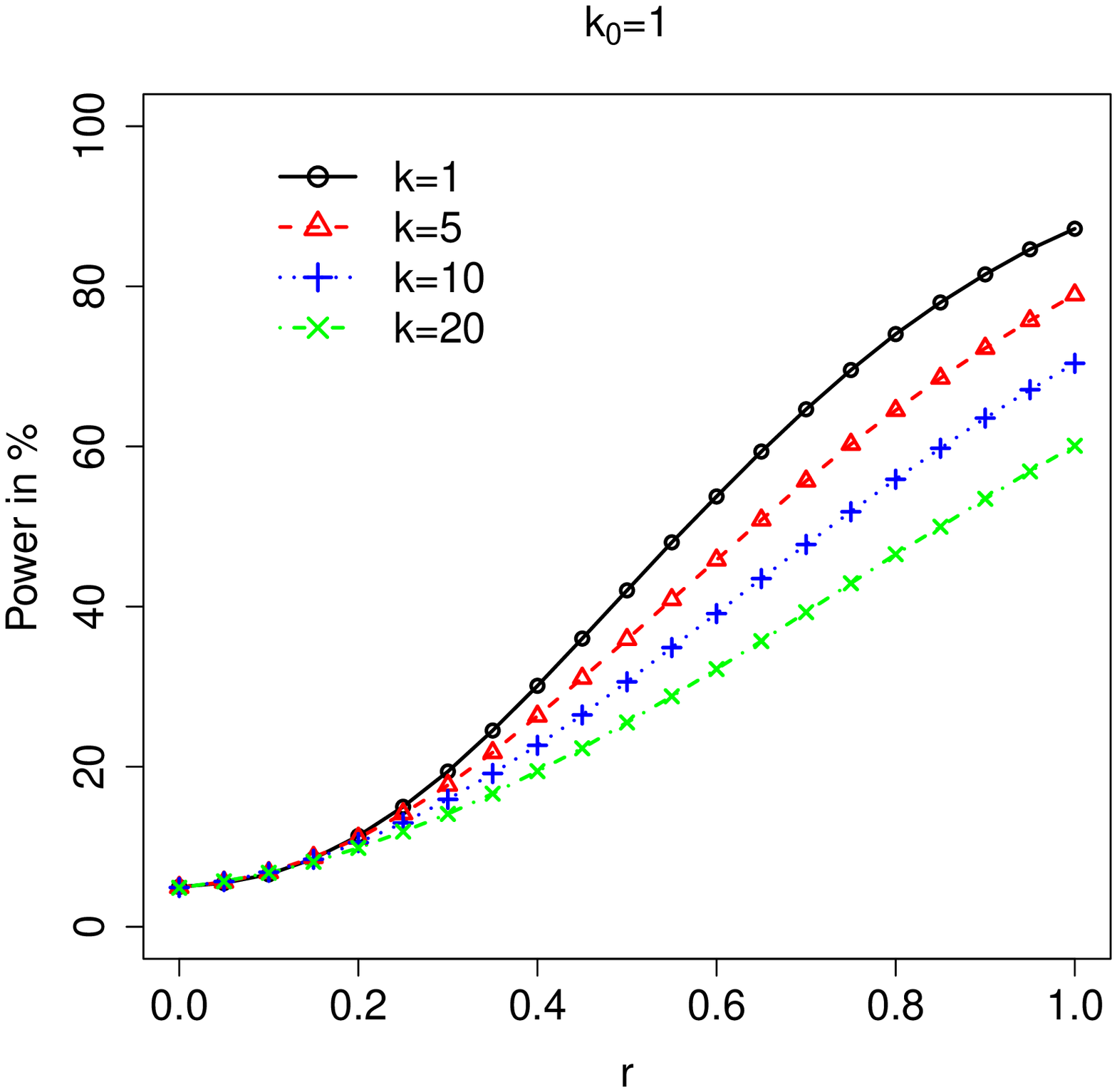}
\includegraphics[height=5.3cm,width=5.8cm]{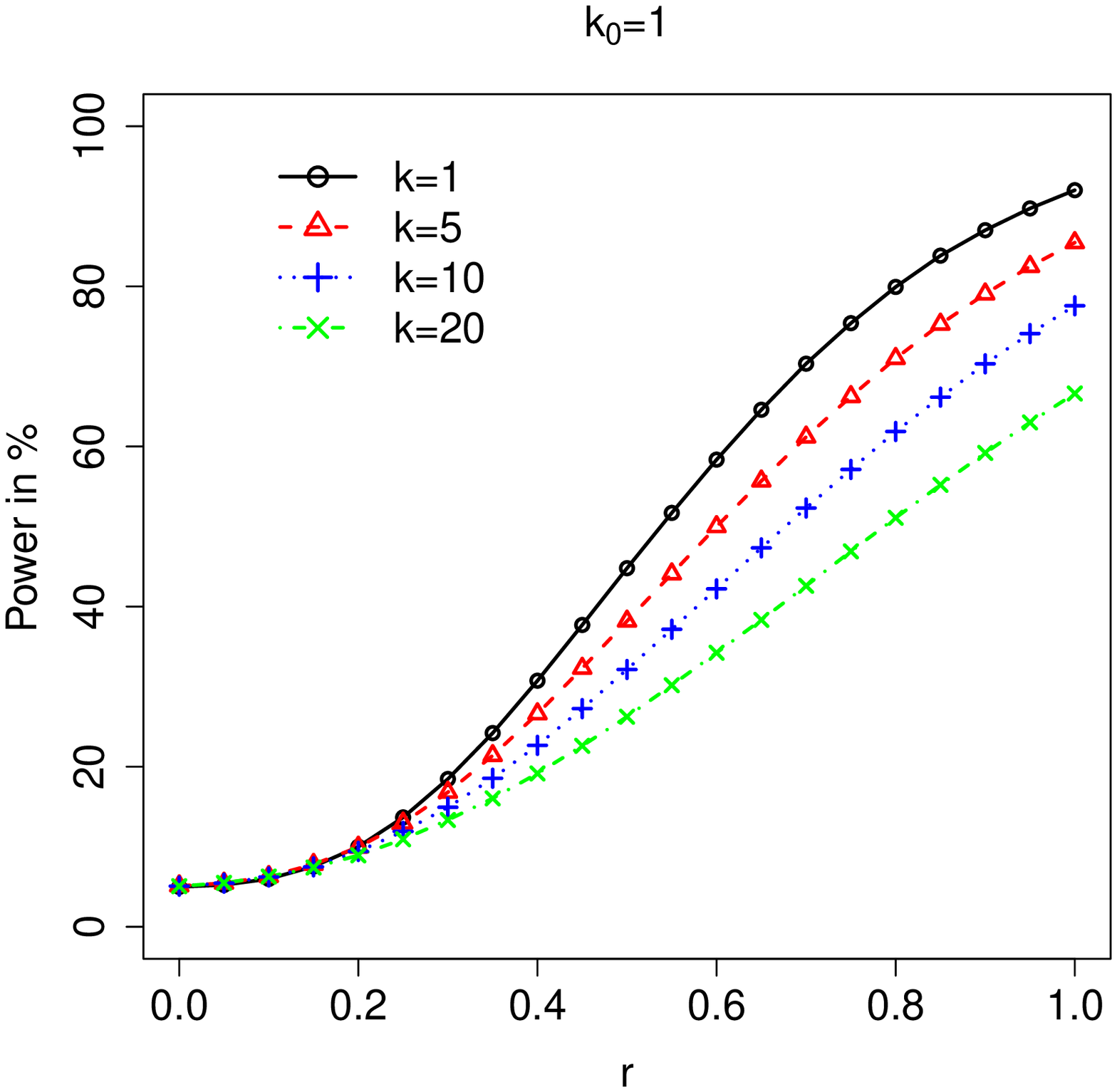}
\includegraphics[height=5.3cm,width=5.8cm]{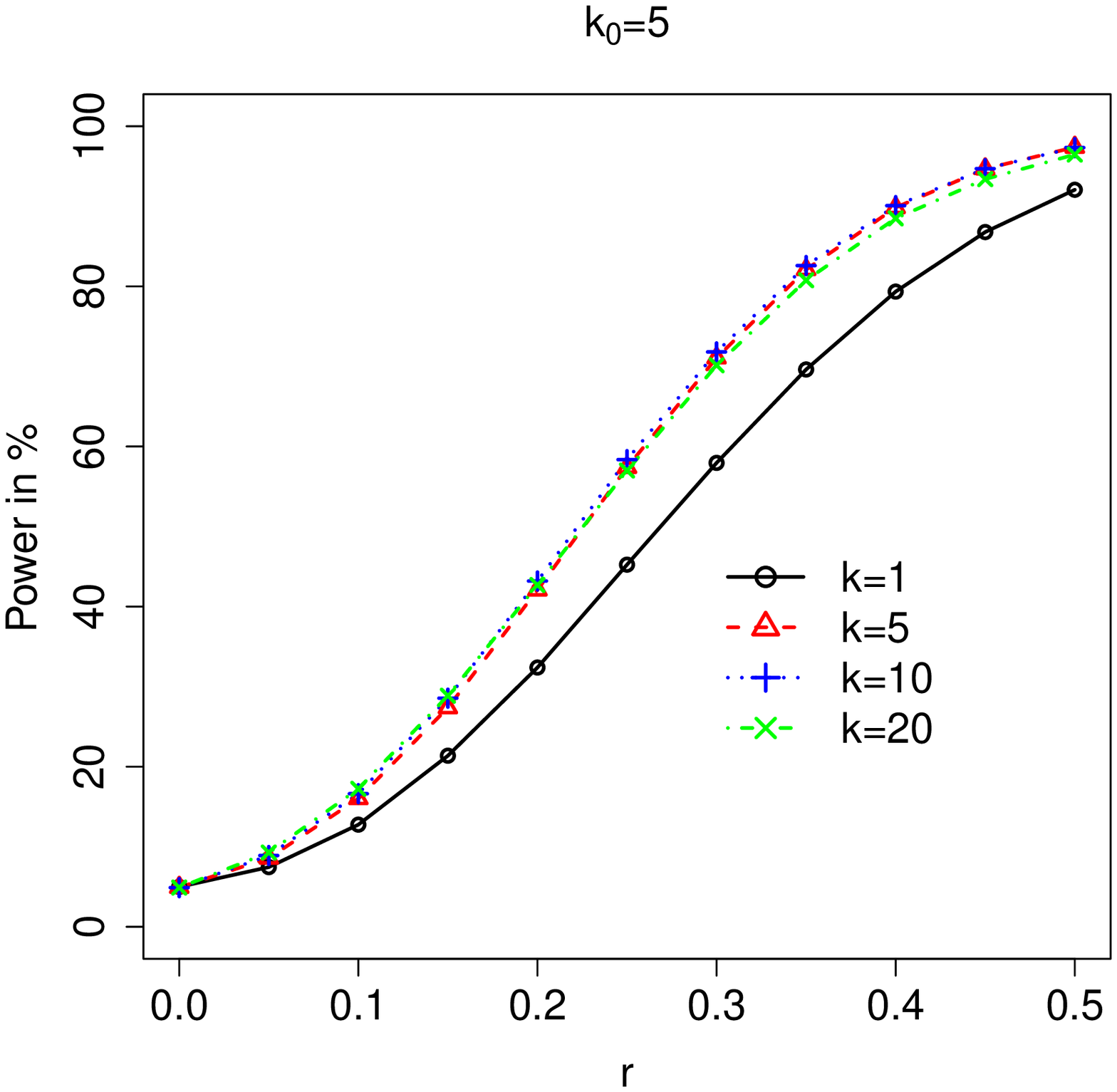}
\includegraphics[height=5.3cm,width=5.8cm]{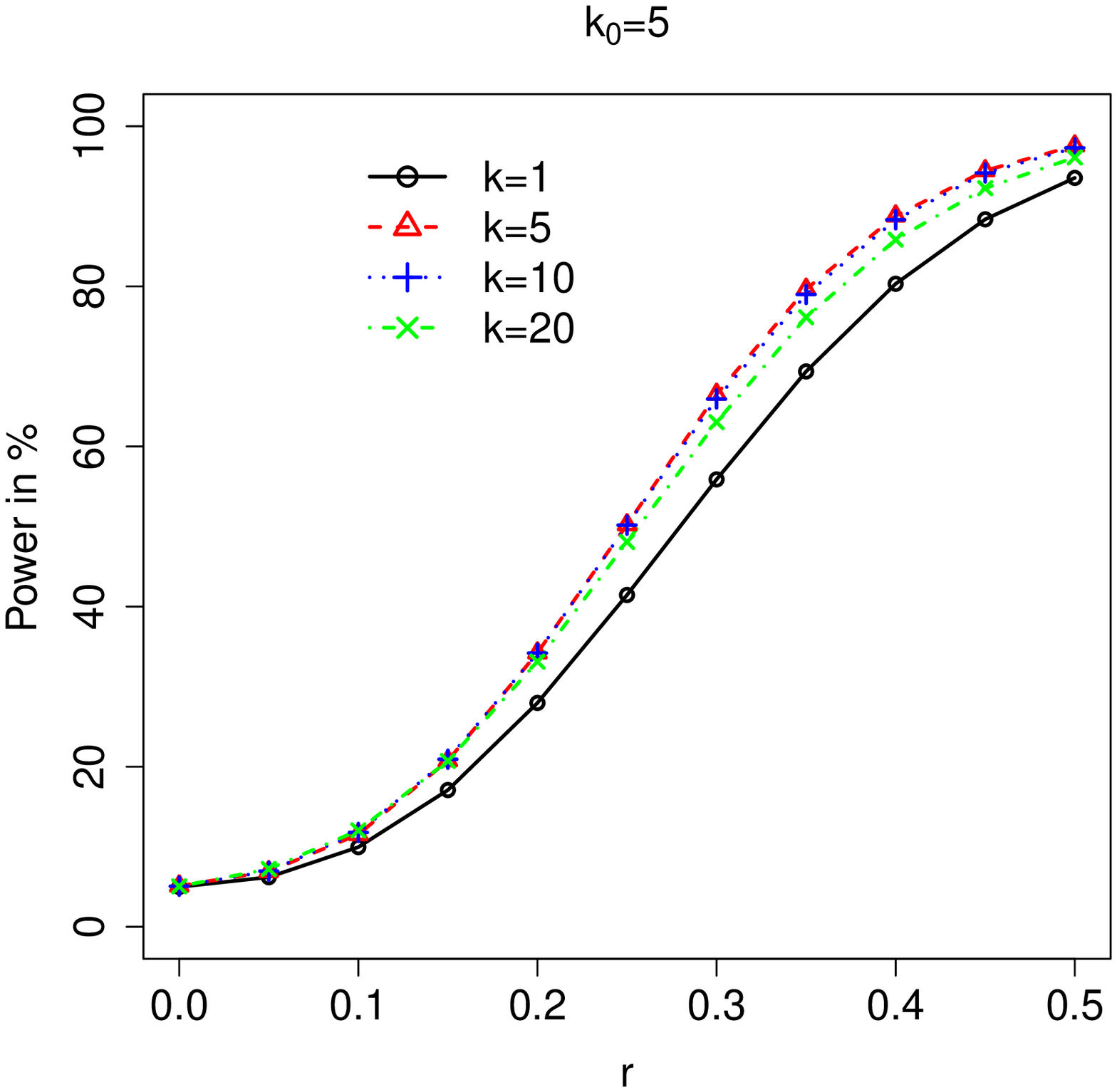}
\includegraphics[height=5.3cm,width=5.8cm]{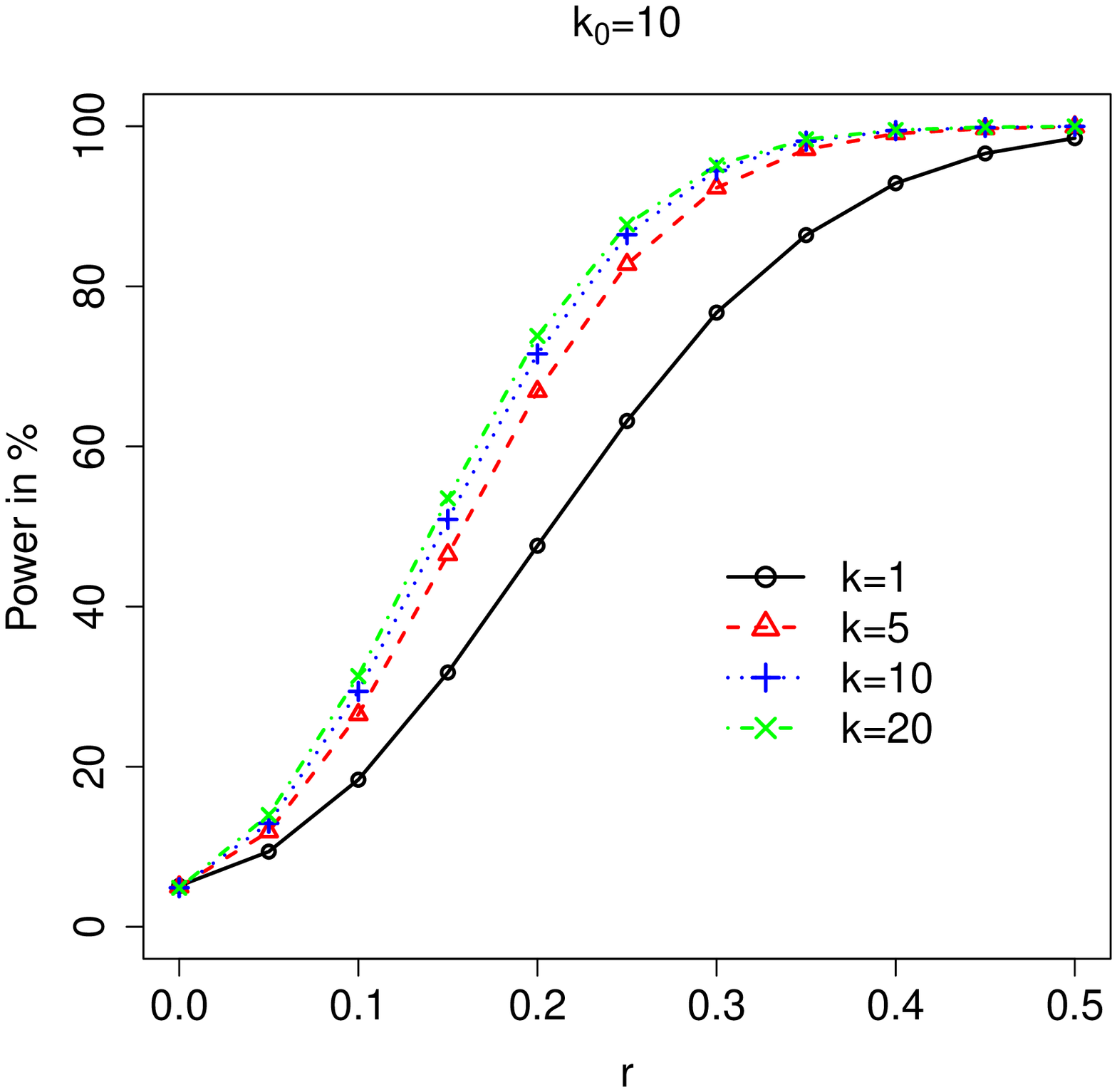}
\includegraphics[height=5.3cm,width=5.8cm]{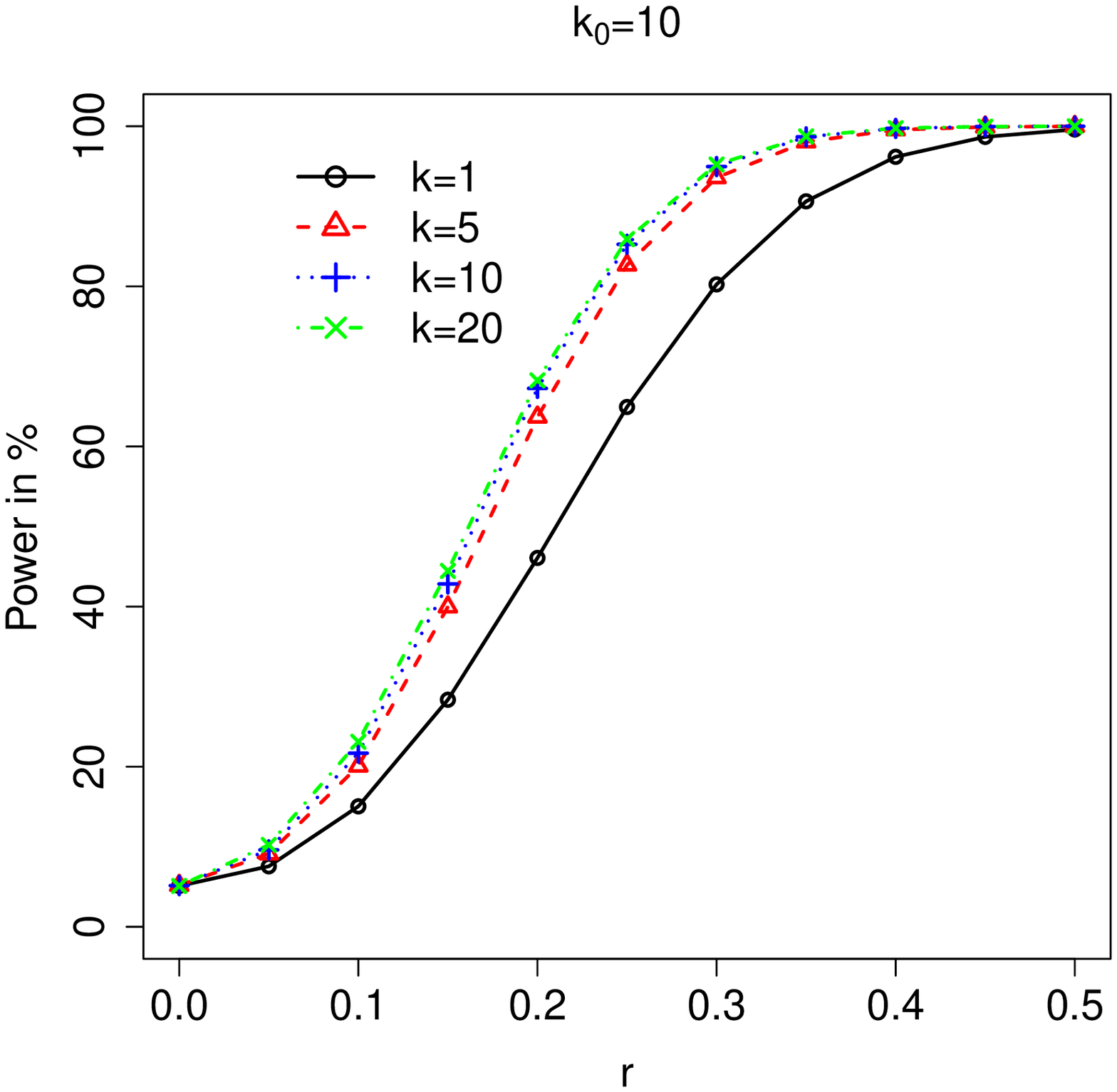}
\includegraphics[height=5.3cm,width=5.8cm]{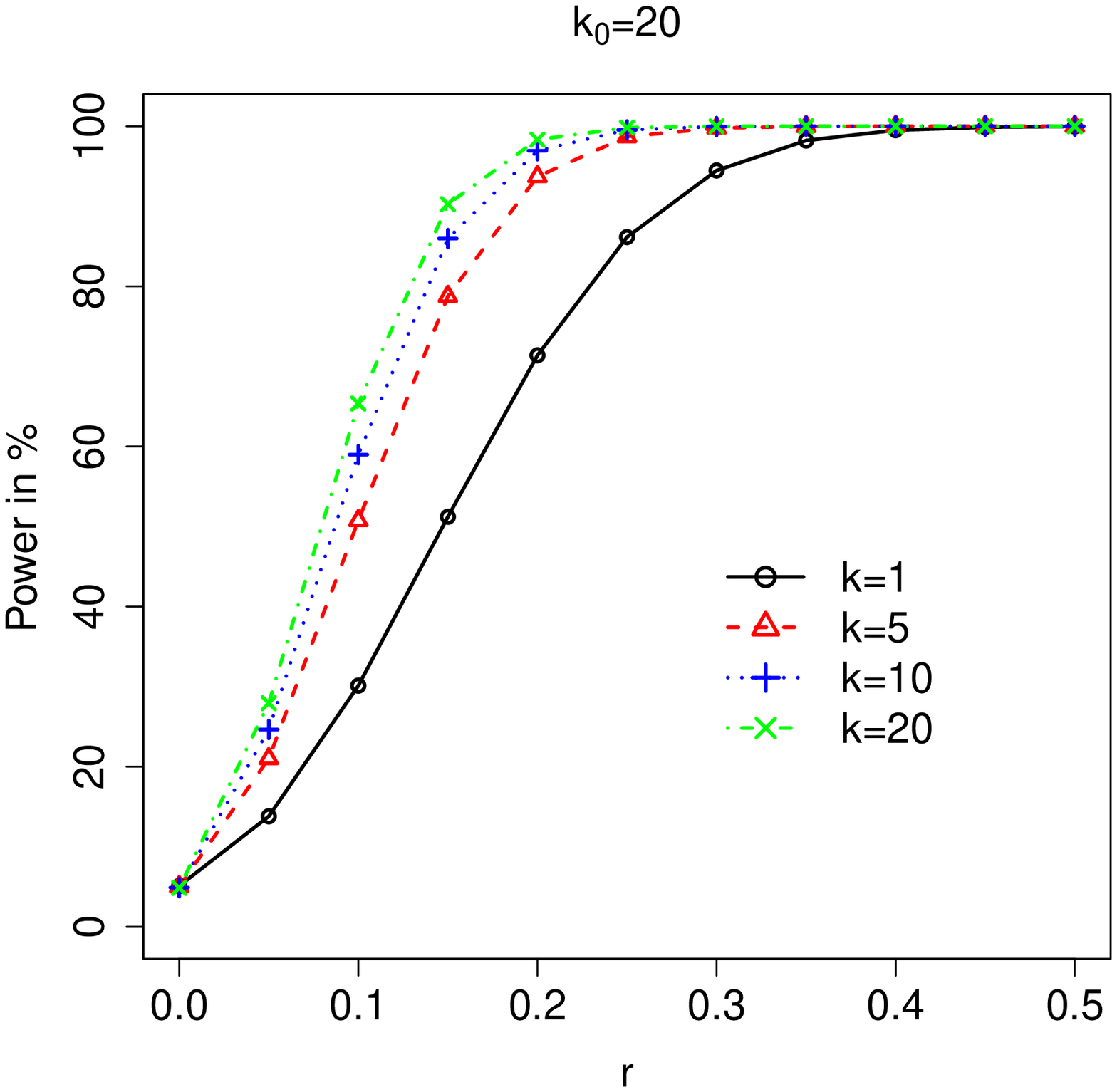}
\includegraphics[height=5.3cm,width=5.8cm]{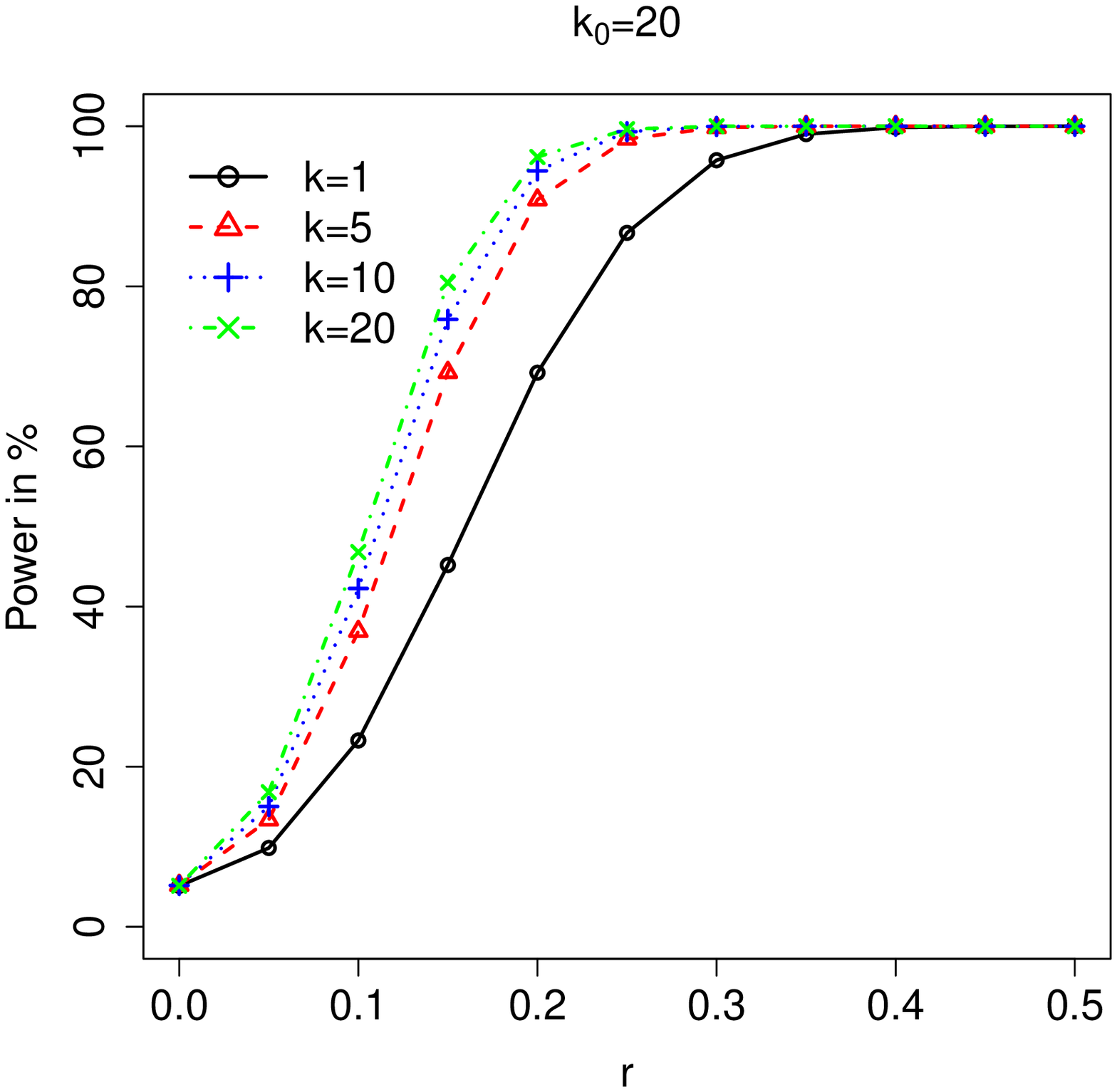}
\caption{Power curves for $T_n(k)$, where $k=1,5,10,20$, and
$\Sigma=(\sigma_{i,j})^{p}_{i,j=1}$ with $\sigma_{i,j}=0.6^{|i-j|}$.
Here $p=200$ for the left panels and $p=1000$ for the right panels.
The number of Monte Carlo replications is 100000.}\label{fig:power}
\end{figure}

\begin{table}[ht]\footnotesize
\caption{Rejection probabilities in \% for Models (a), (b), (c), and (d), where
$p=50,100,200$, and $n_1=n_2=80$. The results are obtained based on
1000 Monte Carlo replications.}\label{tb1}
\begin{center}
\begin{tabular}{clccccccccccc}
\toprule
Model & & $p$  & $T^2$ & BS & CQ & CLX & $T_{fe,n}(4)$ & $T_{fe,n}(8)$ & $T_{fe,n}(12)$ & $T_{fe,n}(24)$ & $\widetilde{T}_{fe,n}(40)$ \\
\hline
(a) & $H_0$& $50$ & 5.5 & 6.8 & 6.8 & 4.1 & 5.5 & 5.5 & 5.6 & 5.4 & 6.0 \\
&& $100$  & 6.6 & 6.3 & 6.3 & 6.2 & 6.9 & 6.1 & 5.8 & 4.8 & 6.7 \\
&& $200$  & NA & 5.1 & 5.1 & 5.7 & 6.8 & 5.3 & 5.5 & 4.3 & 5.1\\
&Case 1 & $50$  & 22.9 & 10.5 & 10.5 & 32.4 & 40.7 & 40.2 & 38.9 & 35.5 & 42.7  \\
&& $100$  & 34.8 & 19.6 & 19.6 & 56.8 & 80.4 & 81.7 & 81.6 & 79.4 & 82.1 \\
&& $200$  & NA & 28.9 & 28.9 & 81.7 & 96.5 & 97.7  & 98.4 & 98.6 & 98.5\\
&Case 2 & $50$  & 88.3 & 32.0 & 32.0 & 70.3 & 91.1 & 94.3 & 96.1 & 96.1 & 92.6 \\
&& $100$  & 55.8 & 37.6 & 37.6 & 77.4 & 93.5 & 96.1 & 96.3 & 97.0 & 95.9 \\
&& $200$  & NA & 42.0 & 42.0 & 97.9 & 99.9 & 100.0  & 100.0 & 100.0 & 100.0\\

\hline (b) & $H_0$ & $50$ & 5.5 & 6.6 & 6.6 & 5.2 & 5.9 & 6.7 & 5.9 & 6.2 & 6.4 \\
&& $100$  & 6.6 & 8.2 & 8.2 & 5.8 & 8.7 & 6.6 & 5.9 & 5.6 & 6.9\\
&& $200$  & NA & 5.8 & 5.8 & 6.5 & 8.1 & 6.9 & 6.0 & 4.9 & 6.0\\
&Case 1 & $50$  & 16.5 & 9.2 & 9.2 & 23.0 & 28.2 & 28.7 & 27.4 & 24.6 & 29.1  \\
&& $100$  & 30.9 & 16.8 & 16.8 & 35.3 & 53.0 & 53.2 & 52.2 & 49.5 &53.4\\
&& $200$  & NA & 22.8 & 22.8 & 57.7 & 80.3 & 84.0 & 84.9 & 83.3 & 84.1\\
&Case 2 & $50$ & 71.5 & 24.6 & 24.5 & 62.6 & 83.2 & 87.9 & 88.4 & 86.7 & 83.6\\
&& $100$  & 47.0 & 31.3 & 31.3 & 50.1  & 74.8 & 77.9 & 78.7 & 79.7 & 77.7\\
&& $200$  & NA & 33.8 & 33.8 & 78.5 & 93.3 & 95.3 & 95.6 & 95.9 & 95.9\\

\hline (c)&$H_0$& $50$ & 5.5 & 6.6 & 6.6 & 4.4 & 7.3 & 7.2 & 6.8 & 6.2 & 7.1  \\
&& $100$ & 6.6 & 7.0 & 7.0 & 6.9 & 8.1 & 7.4 & 7.2 & 6.4 & 7.1\\
&& $200$  & NA & 5.9 & 5.9 & 7.5 & 8.2 & 7.4 & 6.3 & 5.8 & 6.5\\
&Case 1 & $50$ & 15.6 & 14.4 & 14.4 & 16.5 & 22.0 & 21.2 & 19.5 & 17.6 & 21.8 \\
&& $100$  & 19.7 & 28.0 & 28.0 & 30.1 & 43.0 & 42.5 & 40.6 & 37.2 & 43.7\\
&& $200$  & NA & 47.9 & 47.9 & 45.9 & 69.4 & 71.9 & 71.4 & 68.4 & 71.1\\
&Case 2 & $50$ & 47.1 & 52.3 & 52.3 & 37.1 & 56.6 & 60.3 & 59.1 & 57.8 & 56.8 \\
&& $100$  & 52.4 & 63.2 & 63.3 & 53.5 &  78.3 & 80.4 & 81.4 & 81.6 & 80.1\\
&& $200$  &  NA & 70.0 & 70.0 & 60.0 & 82.9 & 86.6 & 87.5 & 87.9 & 87.7\\

\hline (d) &$H_0$& $50$ & 5.5 & 6.6 & 6.5 & 3.5 & 5.6 & 5.9 & 5.9 & 5.7 & 5.6  \\
&& $100$ & 6.6 & 6.3 & 6.3 &  5.1 & 7.0 & 6.2 & 5.7 & 5.5 & 6.7\\
&& $200$  & NA & 5.6 & 5.6 &  5.3 & 6.5 & 6.1 & 5.2 & 4.6 & 5.5\\
&Case 1 & $50$ & 32.3 & 7.7 & 7.7 & 26.2 & 35.7 & 35.2 & 33.2 & 30.2 & 35.5 \\
&& $100$  & 36.8 & 7.8 & 7.8 & 77.5 & 91.9 & 93.8 & 93.8 & 93.1 & 93.6\\
&& $200$  & NA & 7.4 & 7.4 & 96.1 & 99.6 & 100.0 & 100.0 & 100.0 & 100.0\\
&Case 2 & $50$ & 79.9 & 8.2 & 8.2 & 82.7 & 95.9 & 98.0 & 98.4 & 99.0 & 97.5 \\
&& $100$  & 86.3 & 7.7 & 7.7 & 99.9 & 100.0 & 100.0 & 100.0 & 100.0 & 100.0\\
&& $200$  & NA & 8.4 & 8.4 & 97.5 & 100.0 & 100.0 & 100.0 & 100.0 & 100.0\\
\hline
\end{tabular}
\\Note: $T^2$, BS, CQ and CLX denote the Hotelling's $T^2$ test and the two sample tests in Bai and Saranadasa (1996), Chen and Qin (2010), and Cai
et al. (2014) respectively.
\end{center}
\end{table}

\begin{figure}[h]
\centering
\includegraphics[height=5.6cm,width=6cm]{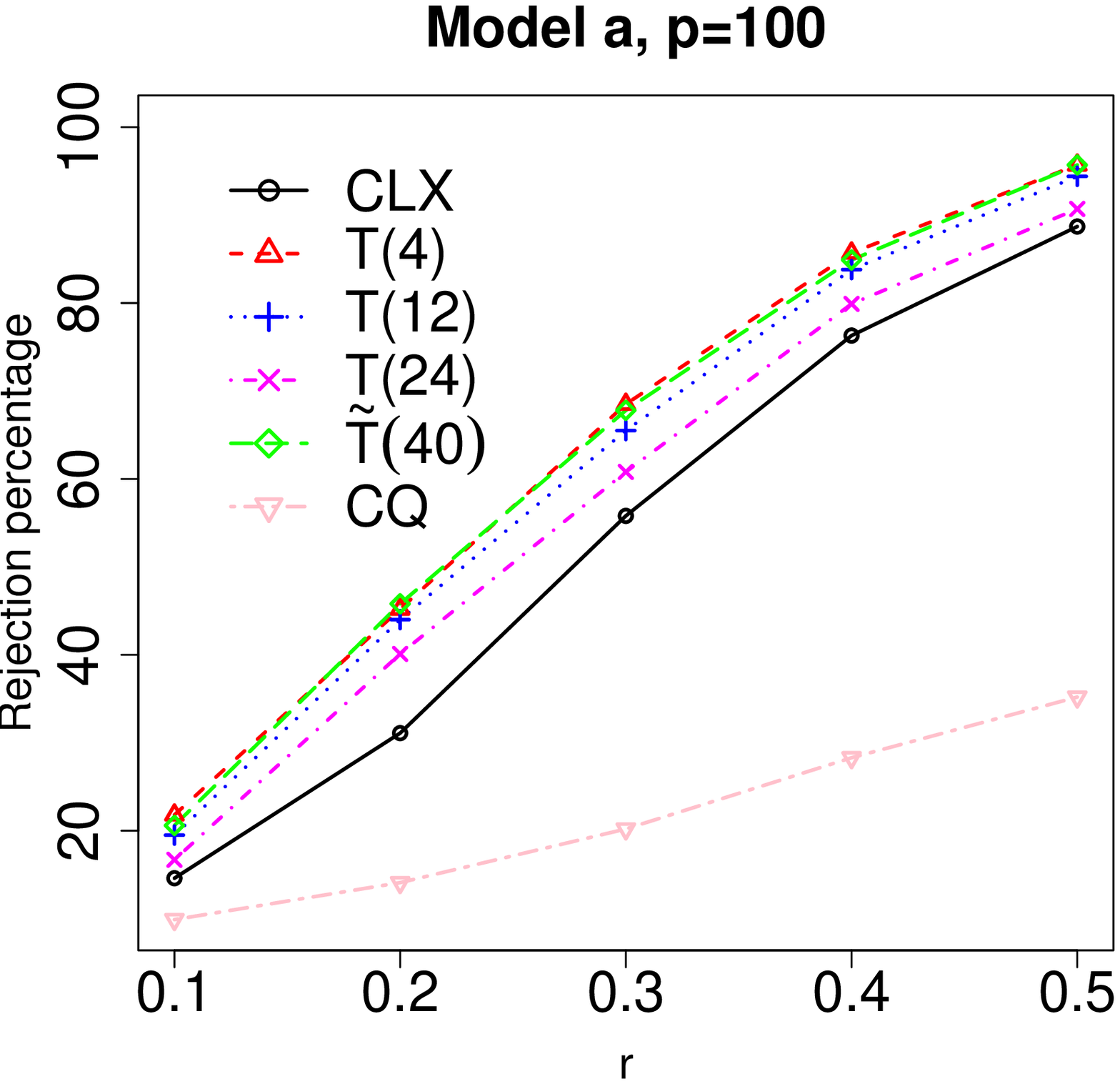}
\includegraphics[height=5.6cm,width=6cm]{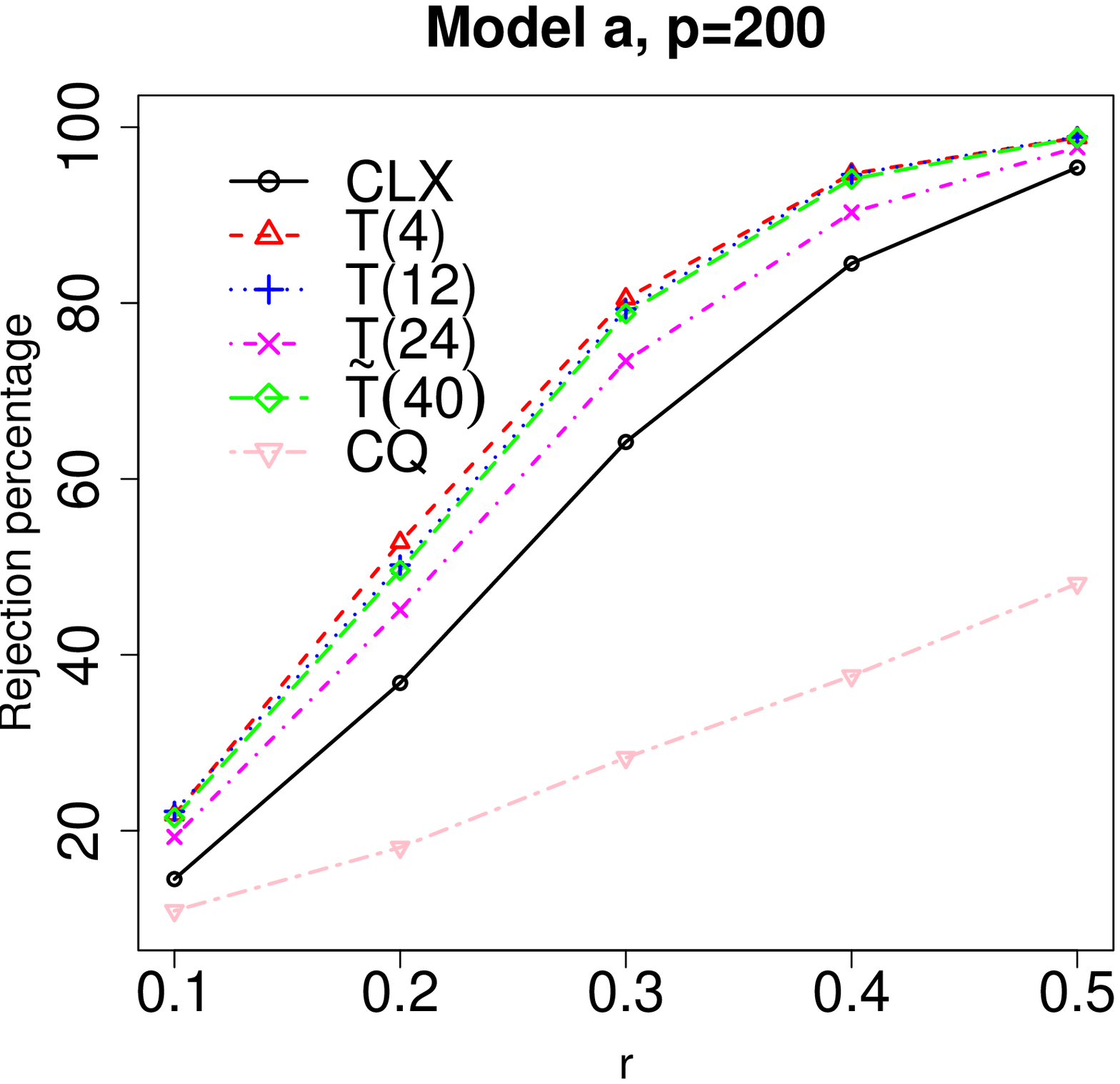}
\includegraphics[height=5.6cm,width=6cm]{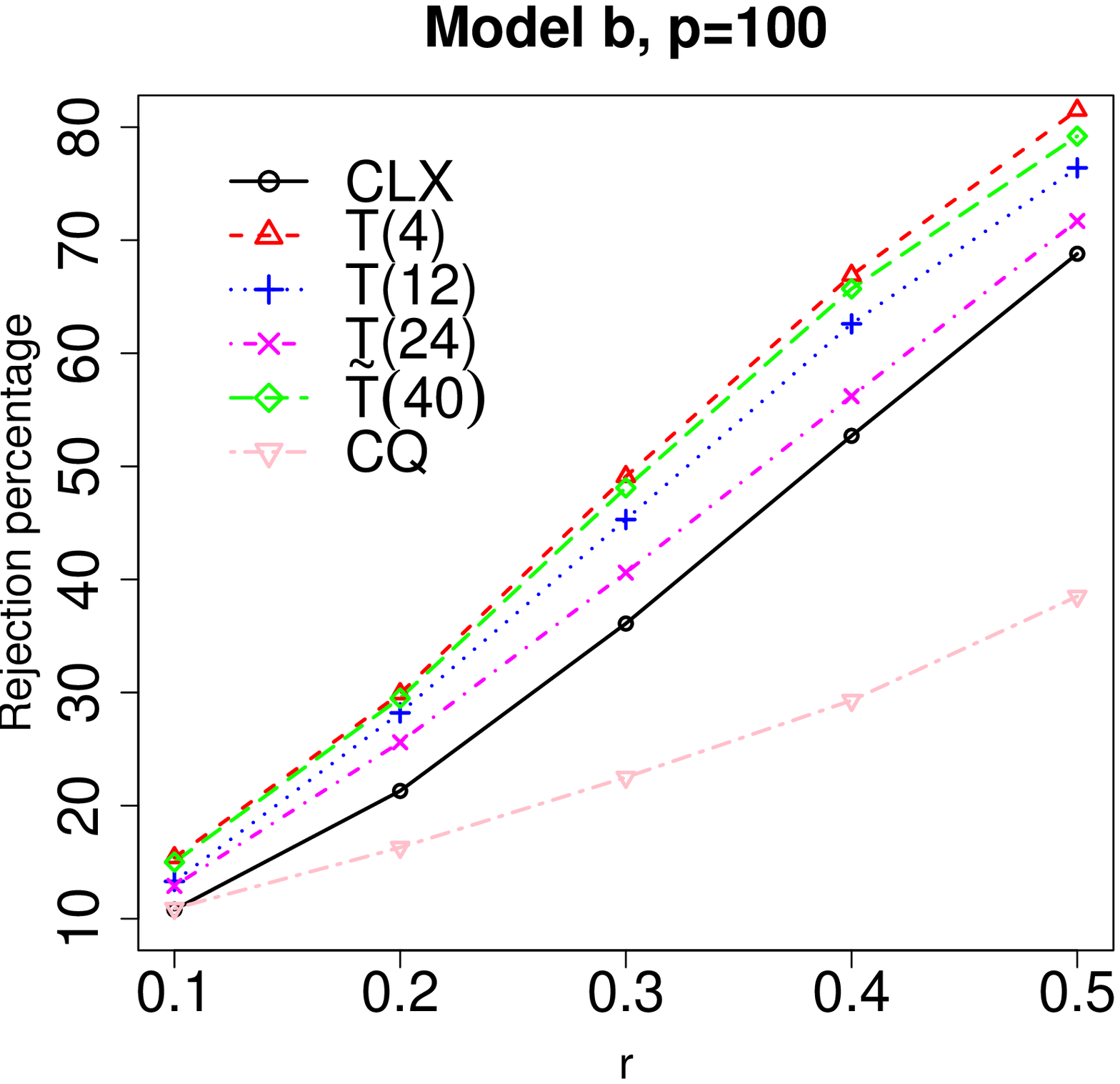}
\includegraphics[height=5.6cm,width=6cm]{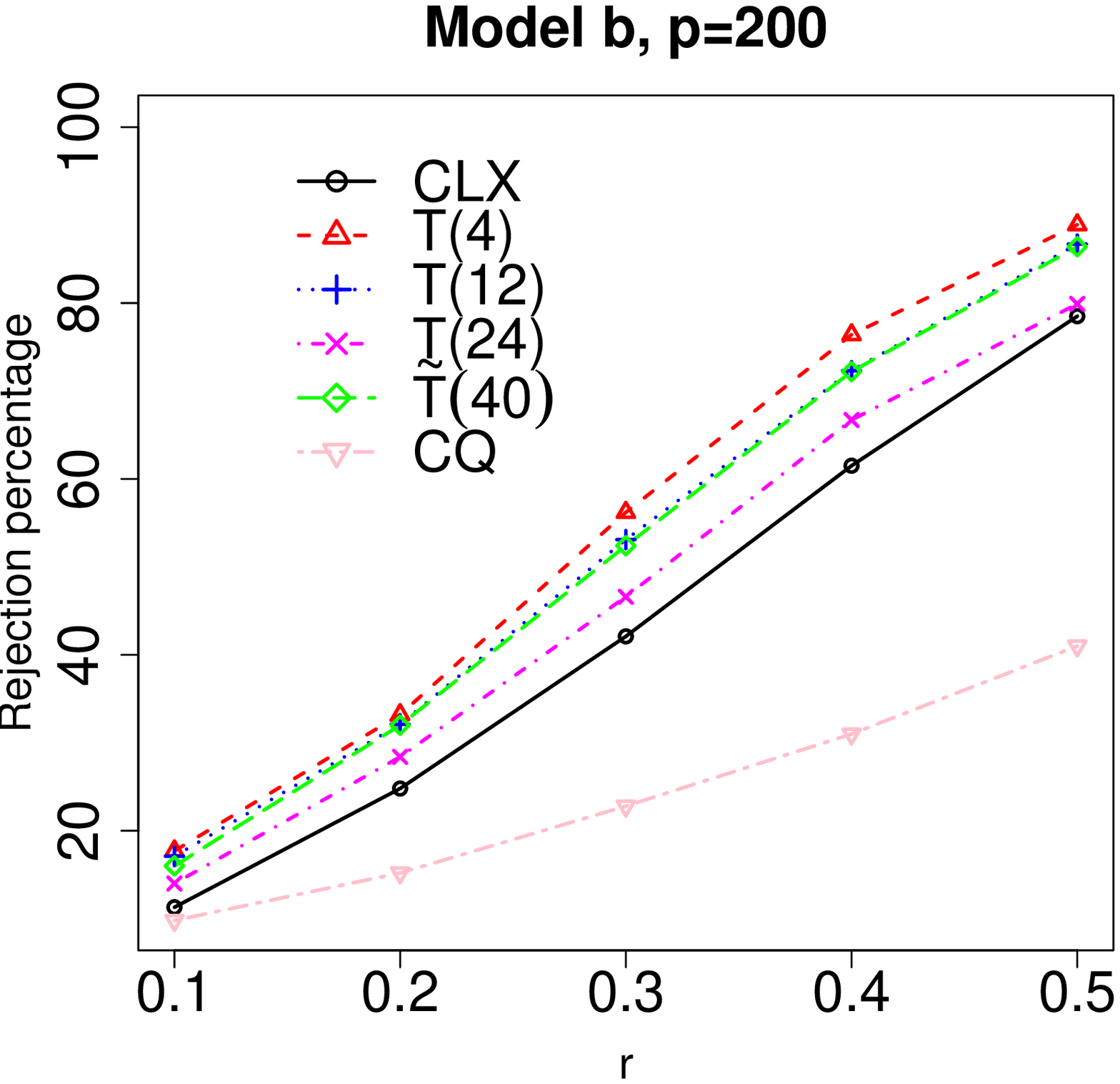}
\includegraphics[height=5.6cm,width=6cm]{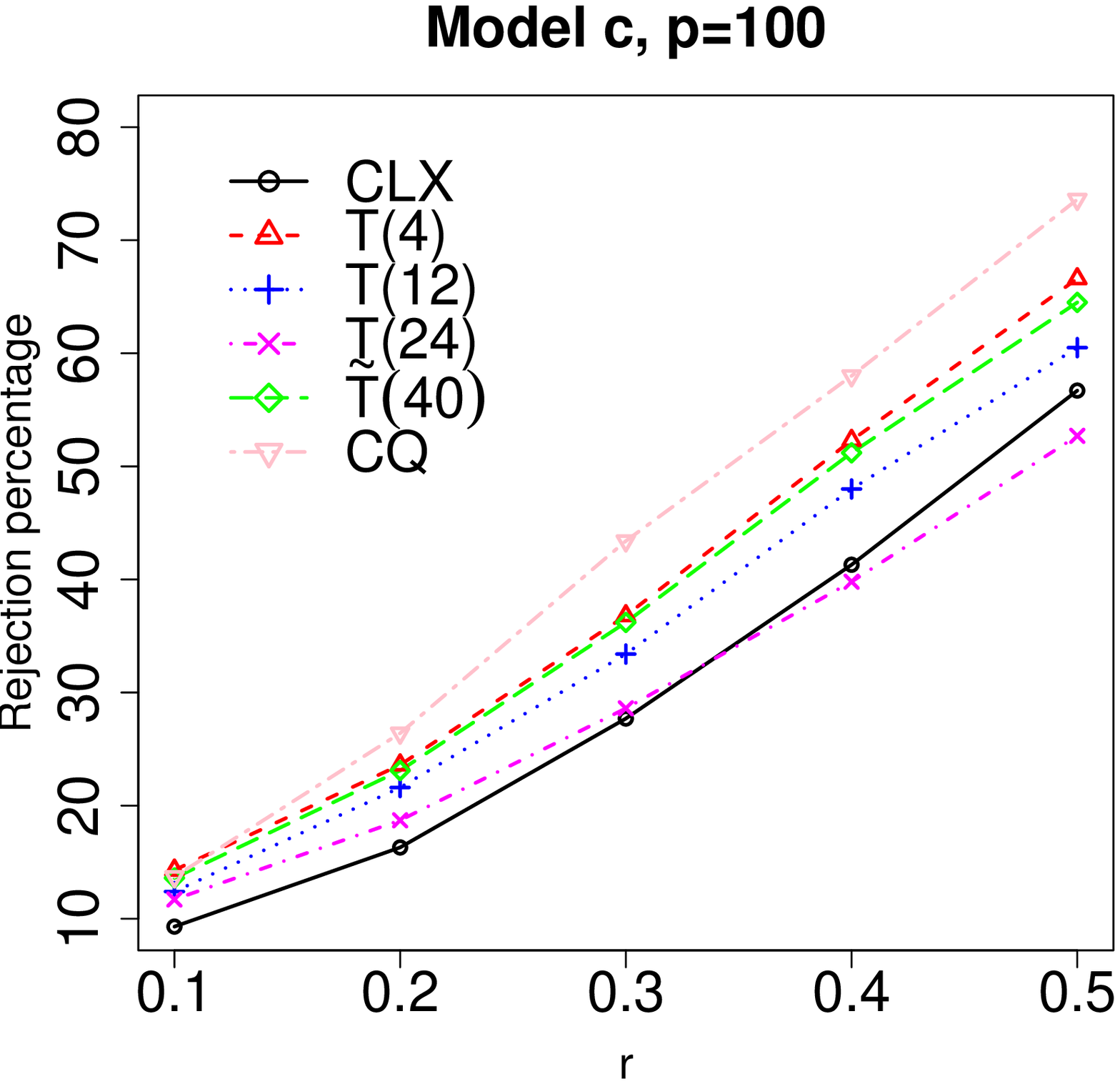}
\includegraphics[height=5.6cm,width=6cm]{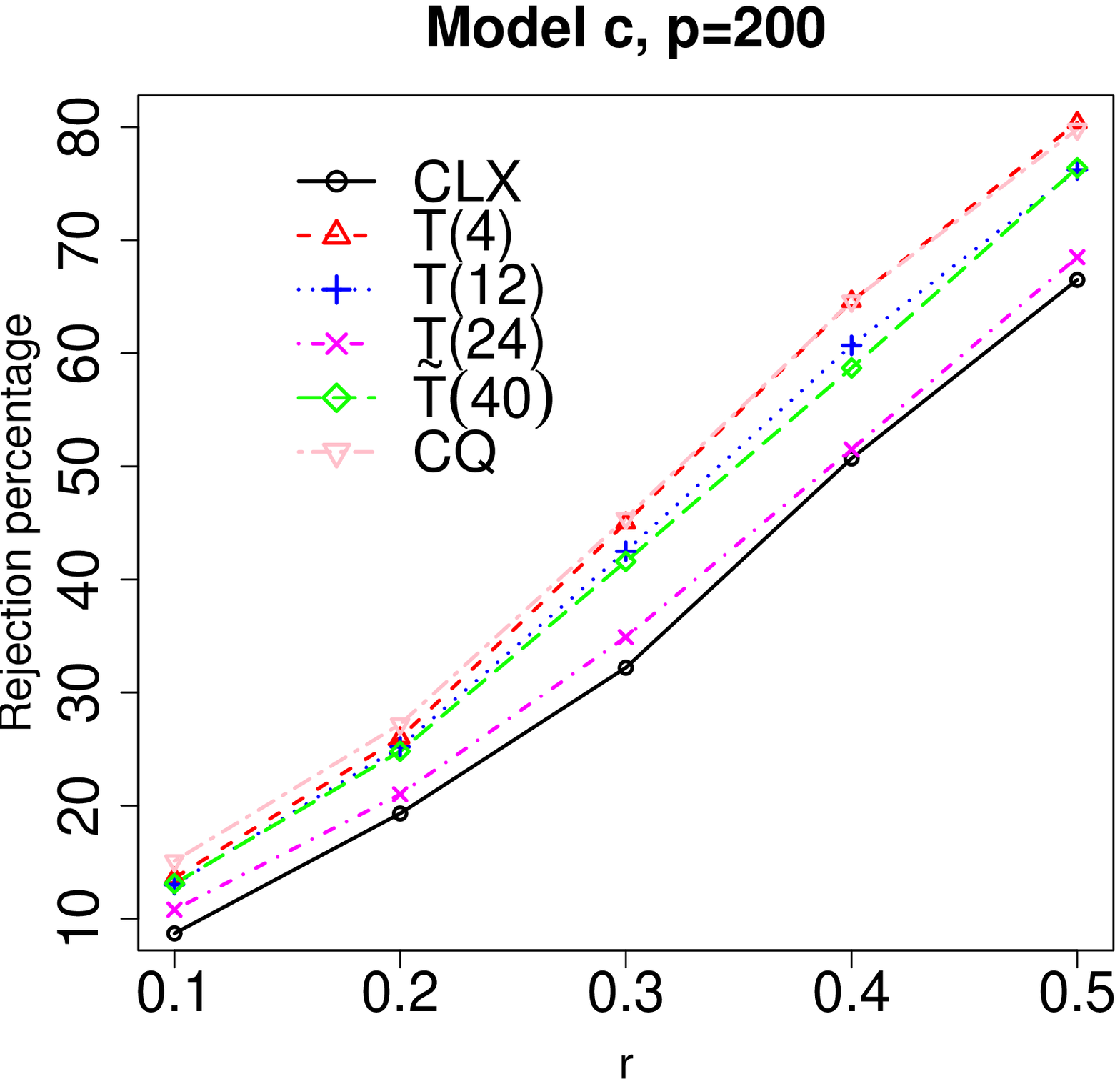}
\includegraphics[height=5.6cm,width=6cm]{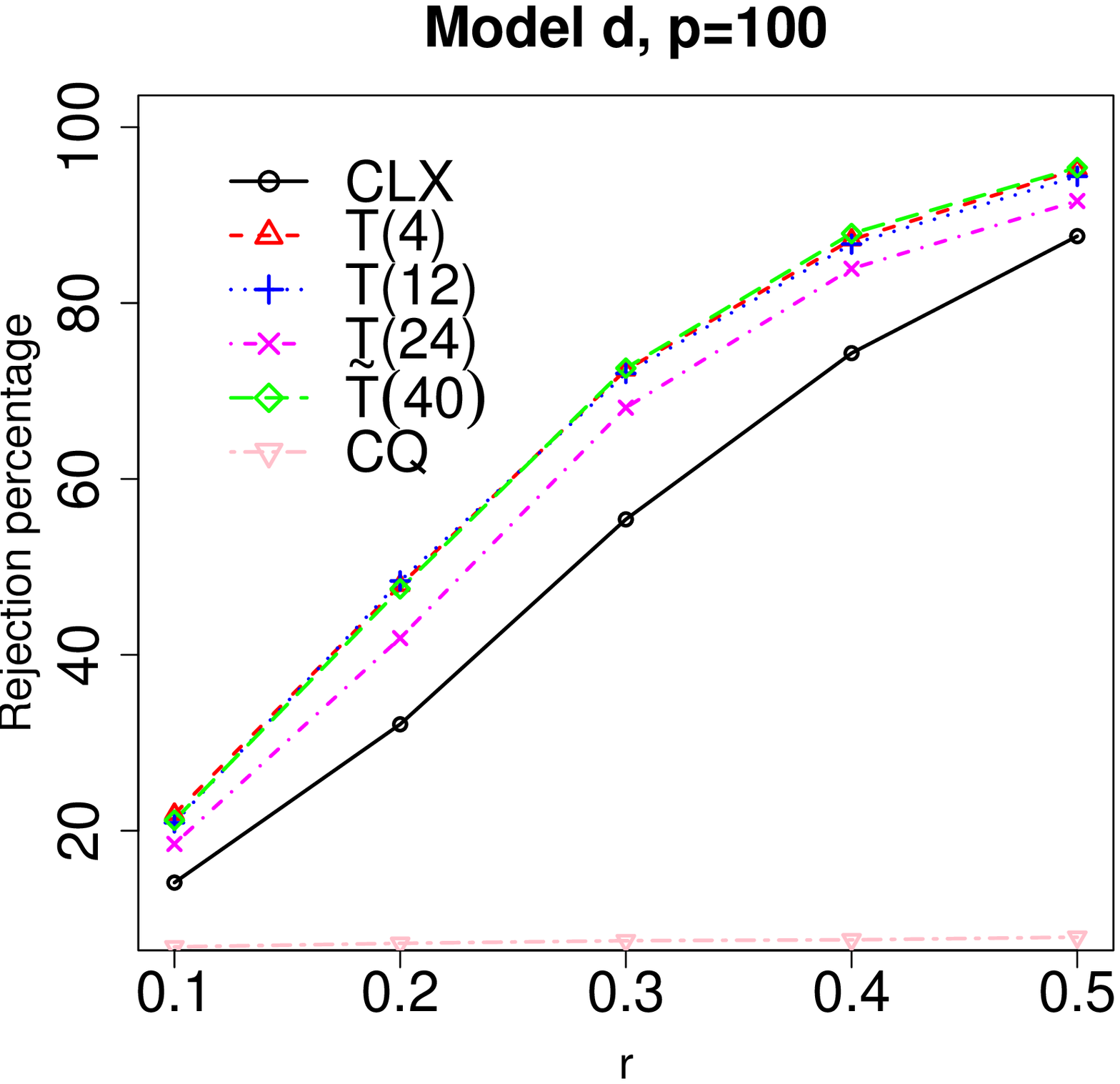}
\includegraphics[height=5.6cm,width=6cm]{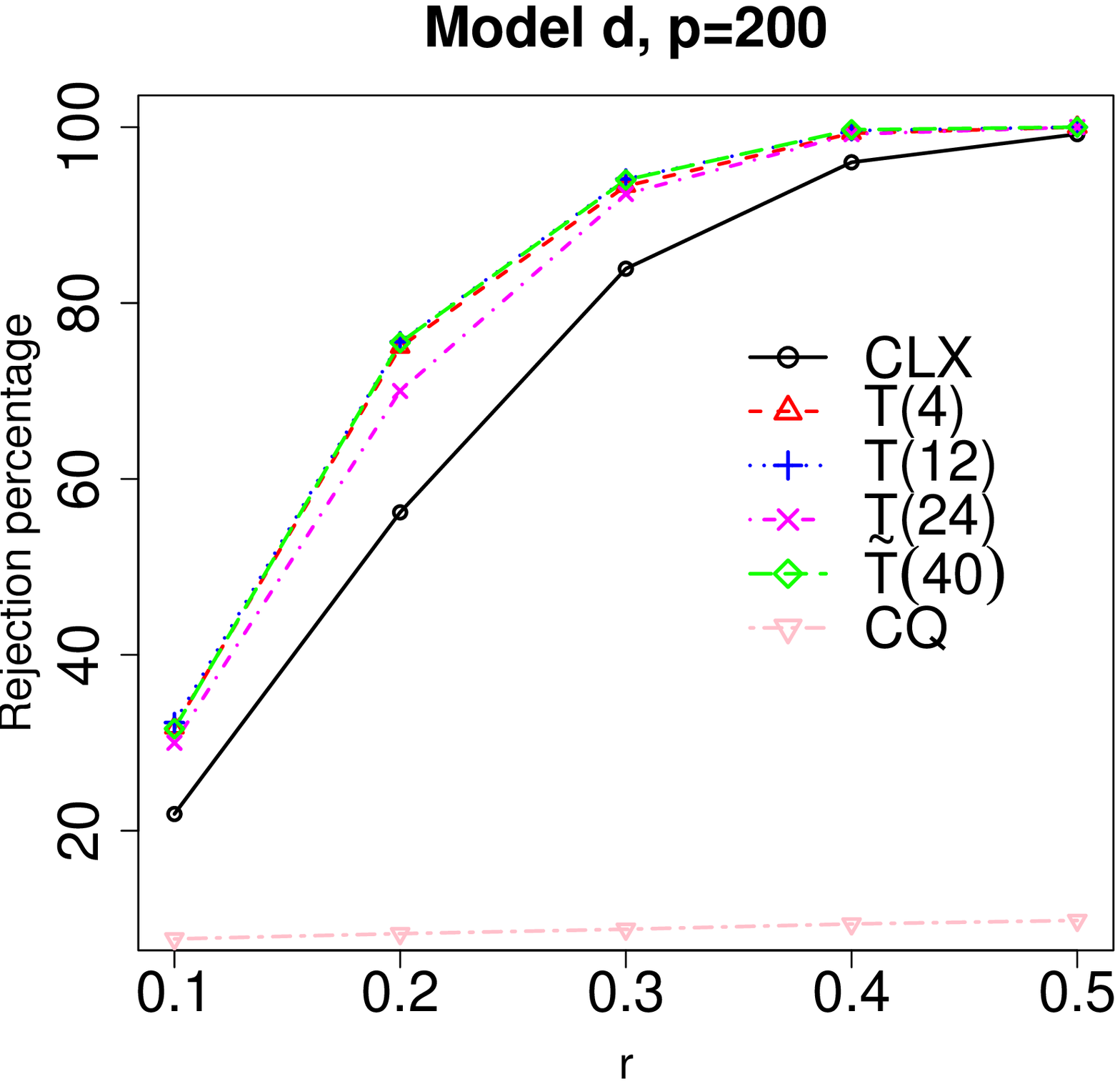}
\caption{Empirical powers for CQ, CLX and the proposed tests under
Models (a), (b), (c), and (d), and case 3, where $n_1=n_2=80$ and
$p=100,200$. The results are obtained based on 1000 Monte Carlo
replications.}\label{fig:p1}
\end{figure}

\begin{figure}[h]
\centering
\includegraphics[height=5.6cm,width=6cm]{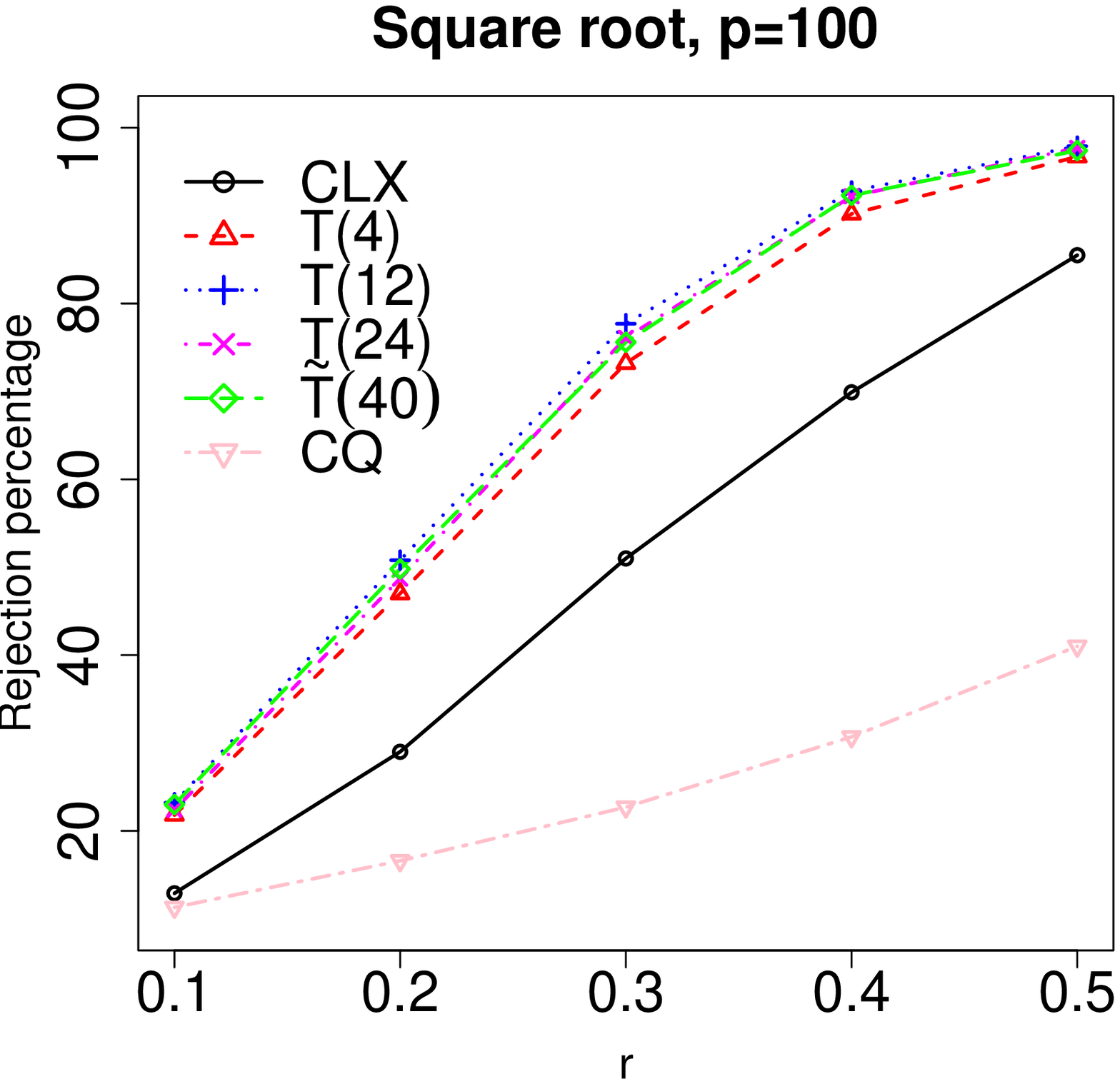}
\includegraphics[height=5.6cm,width=6cm]{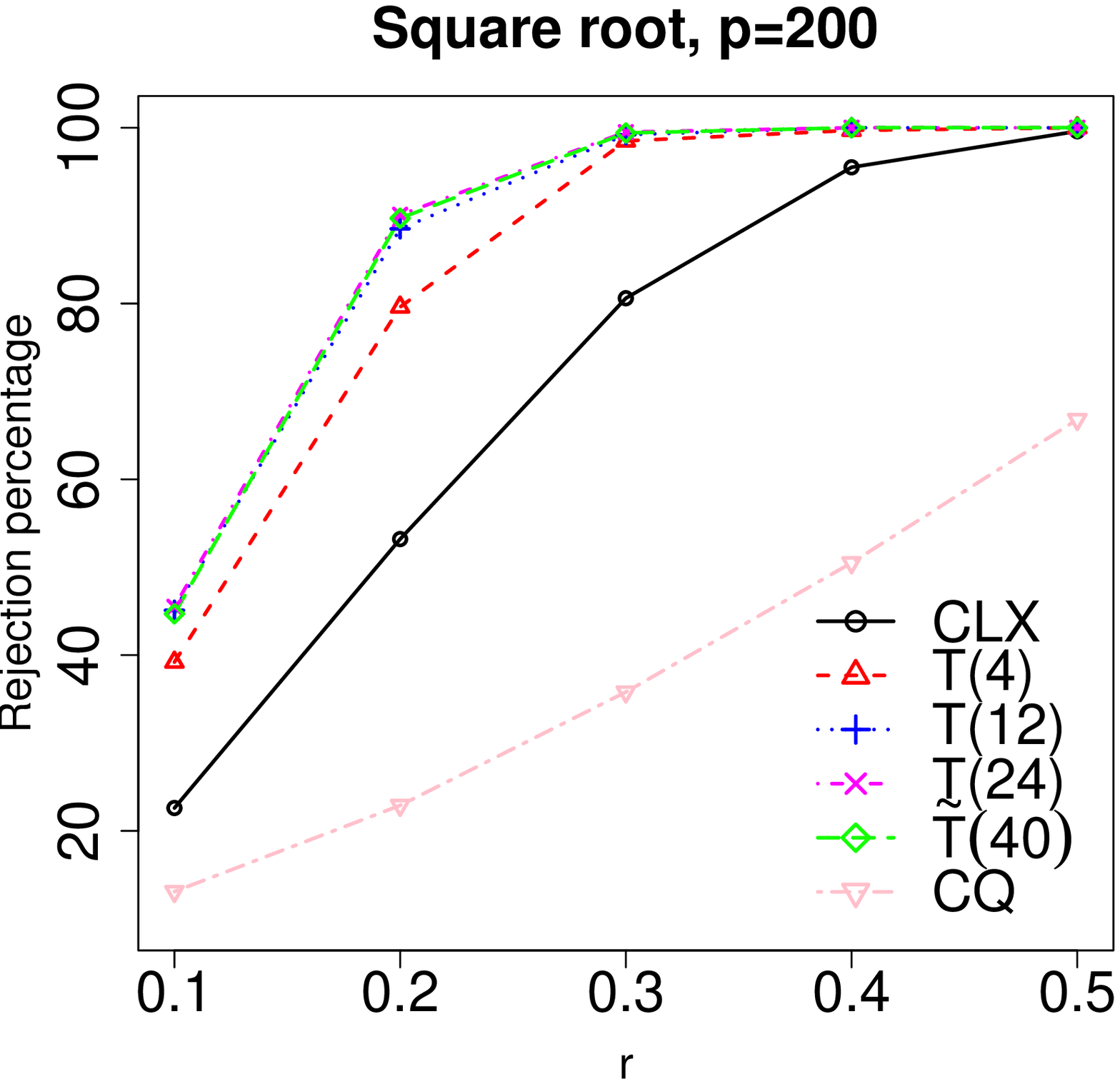}
\includegraphics[height=5.6cm,width=6cm]{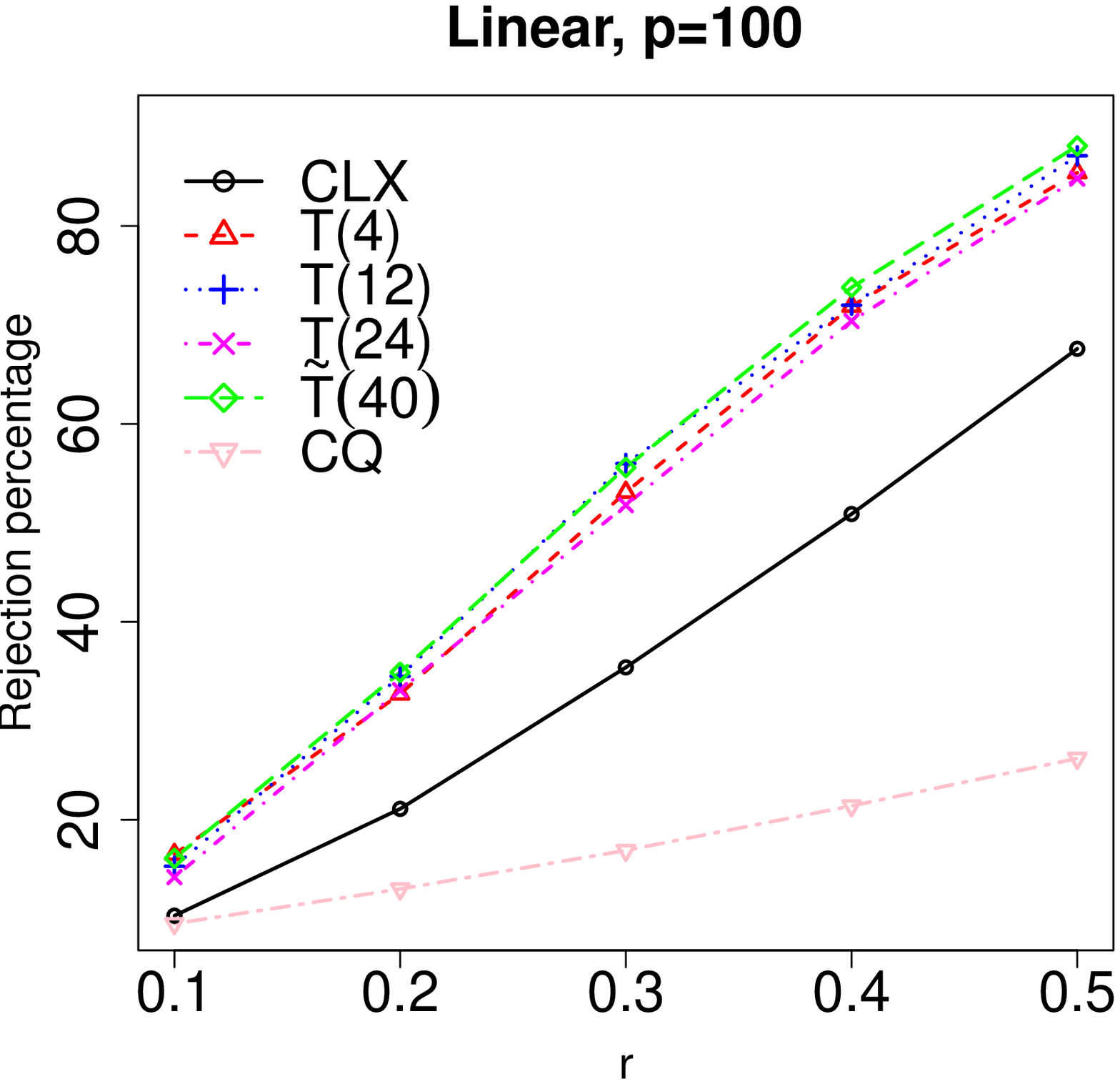}
\includegraphics[height=5.6cm,width=6cm]{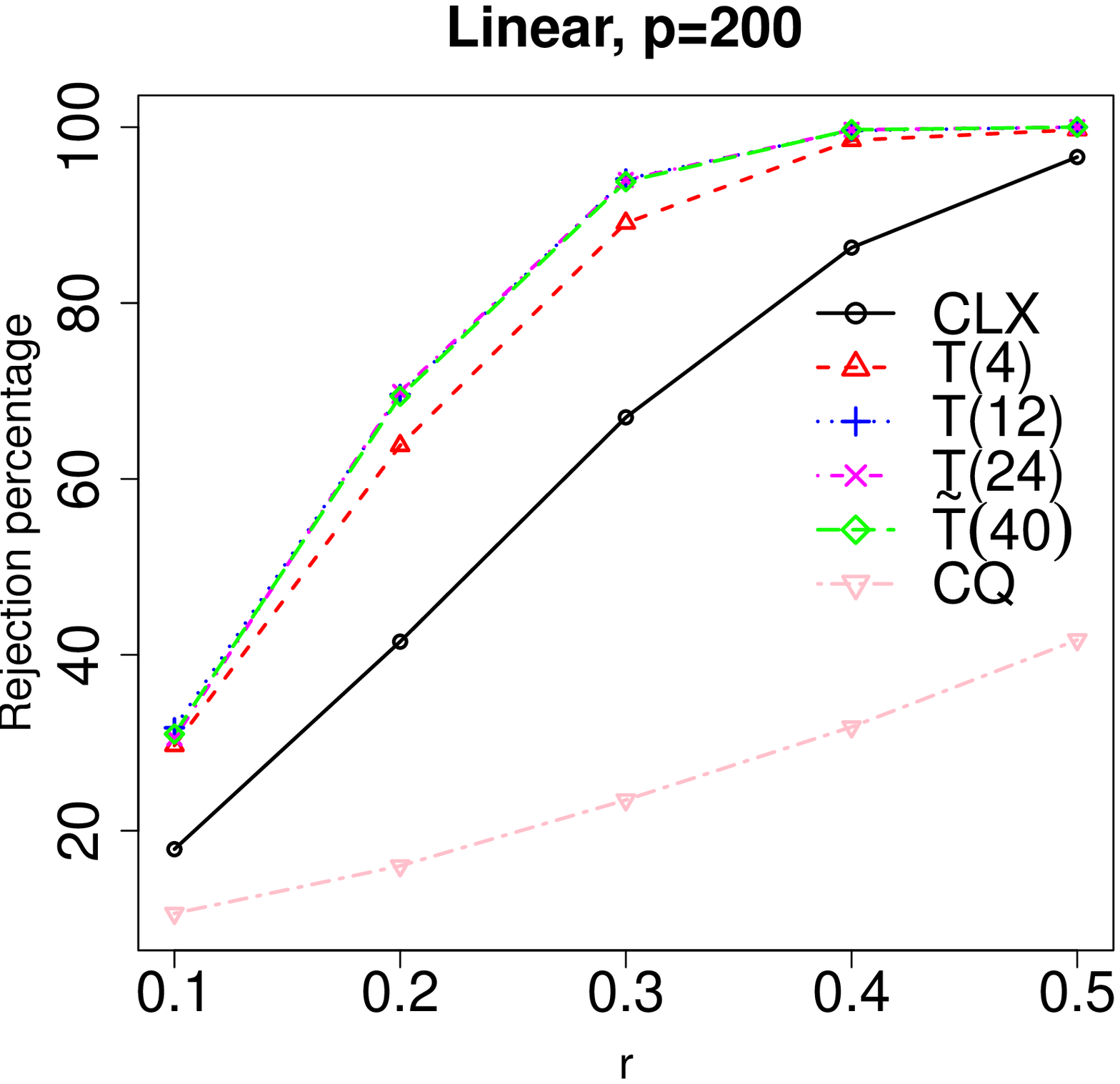}
\includegraphics[height=5.6cm,width=6cm]{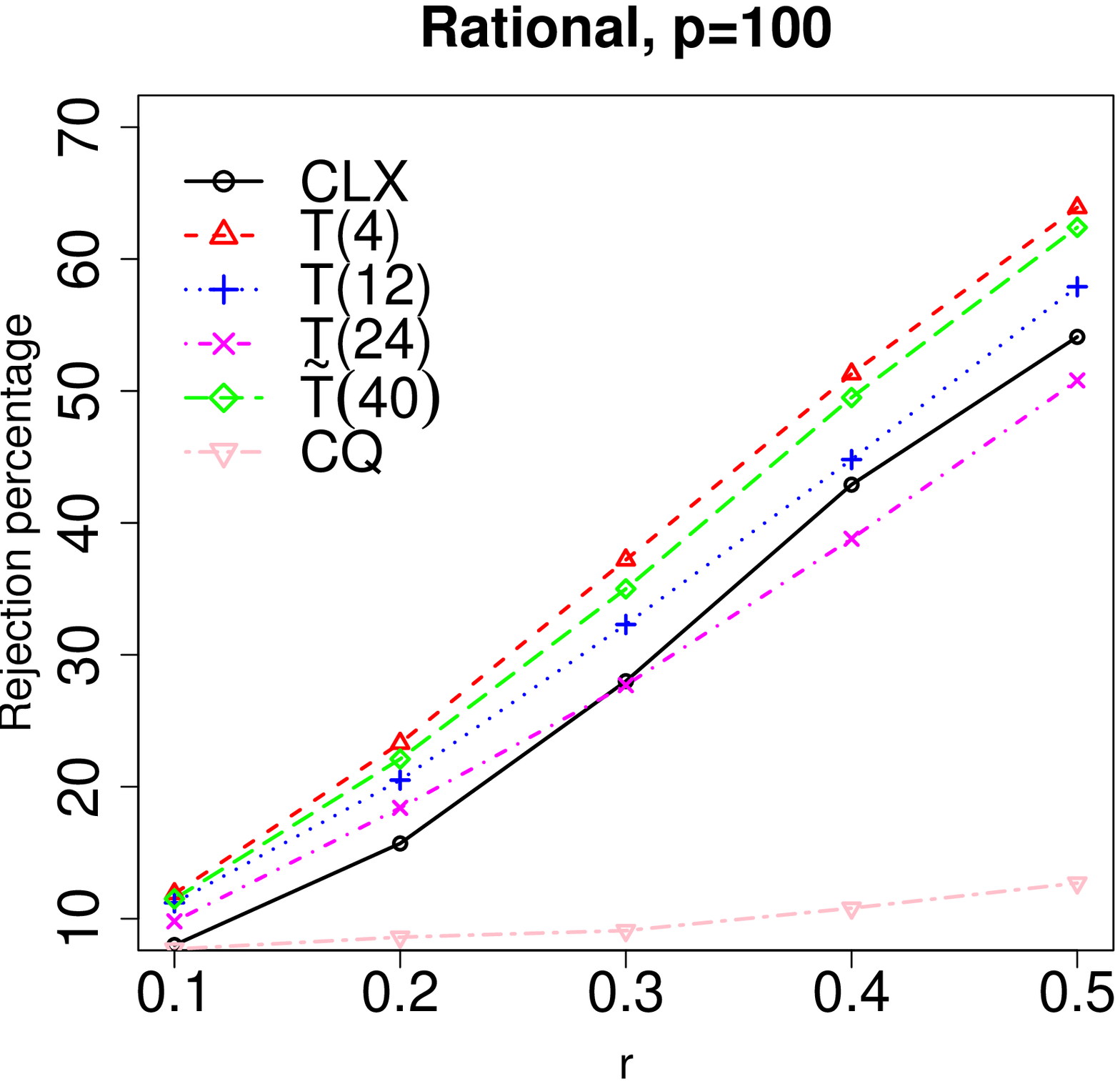}
\includegraphics[height=5.6cm,width=6cm]{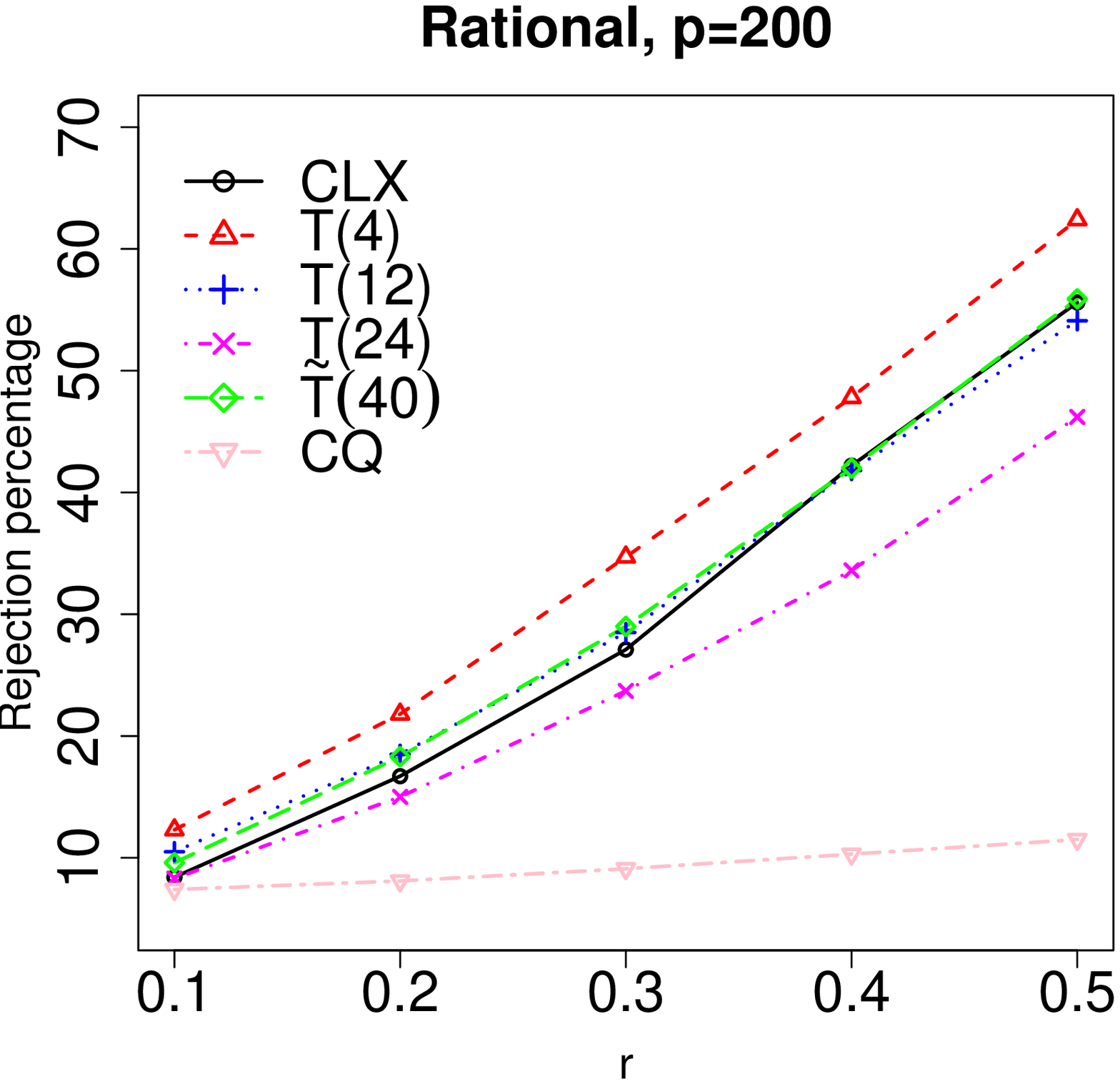}
\includegraphics[height=5.6cm,width=6cm]{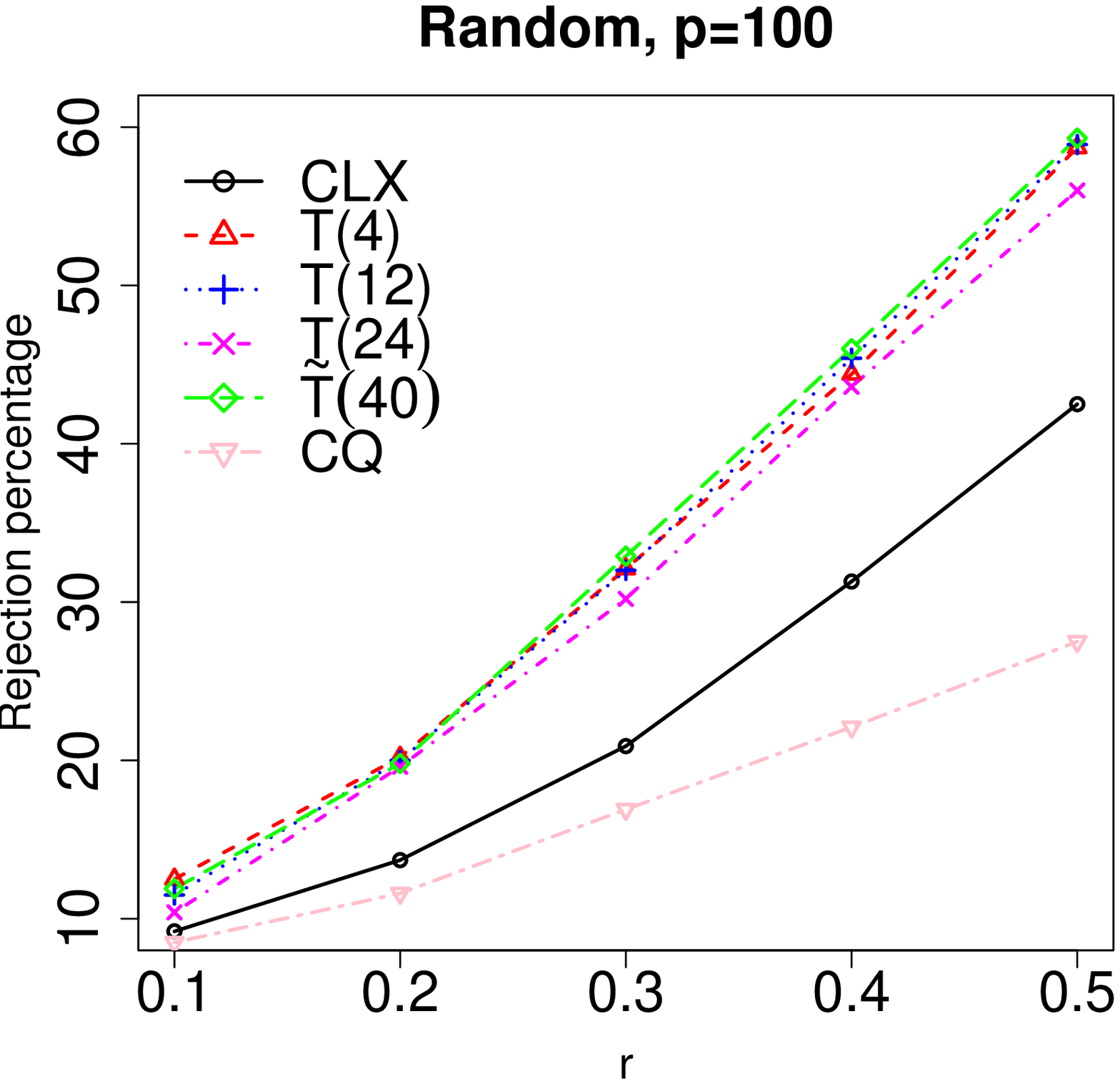}
\includegraphics[height=5.6cm,width=6cm]{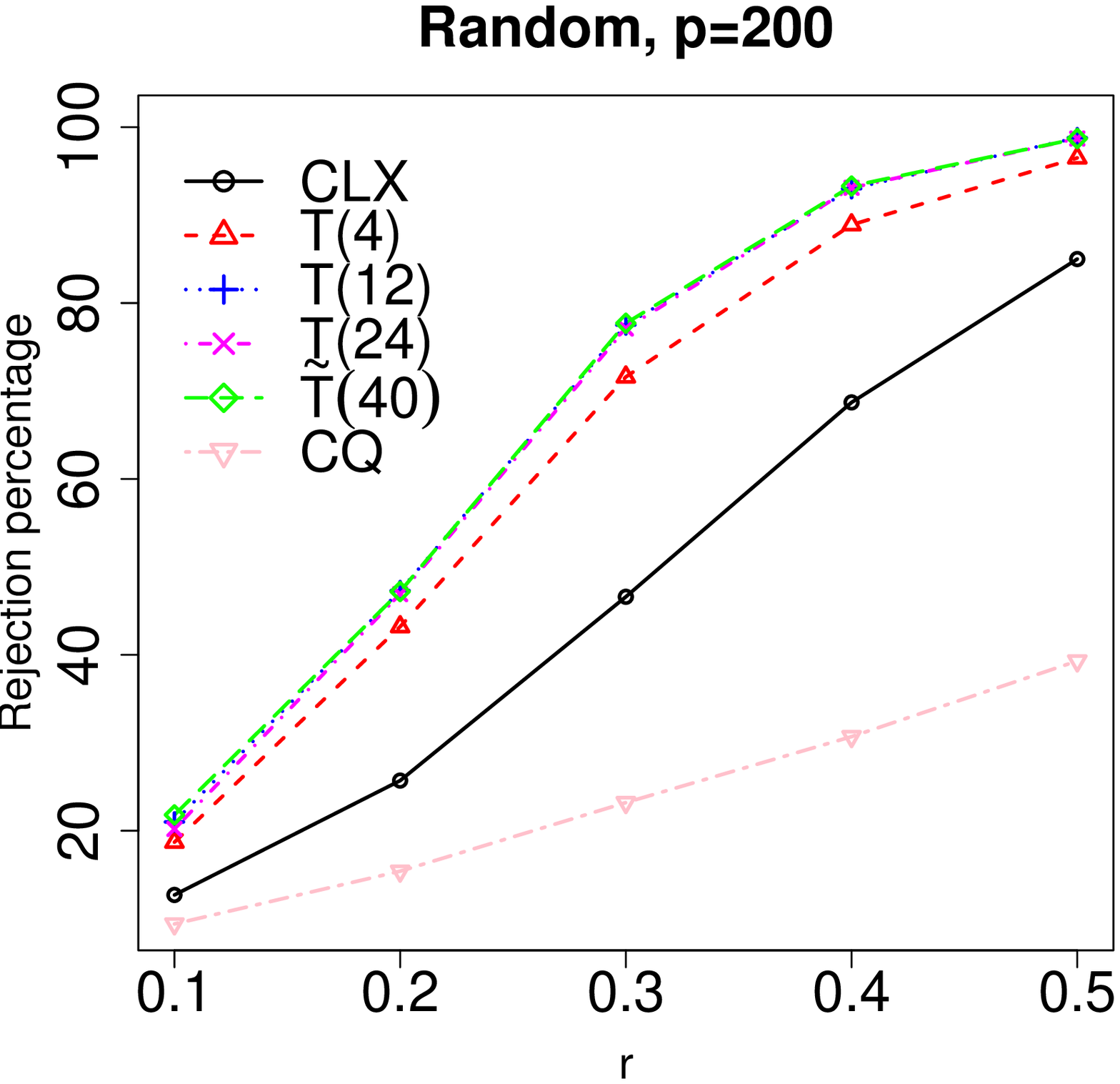}
\caption{Empirical powers for CQ, CLX and the proposed tests under
different signal allocations, where $n_1=n_2=80$ and $p=100,200$.
The results are obtained based on 1000 Monte Carlo
replications.}\label{fig:p2}
\end{figure}


\begin{thebibliography}{99}
\bibitem{acp11}
\par\noindent\hangindent2.3em\hangafter 1
\textsc{Arias-Castro, E., Cand\`{e}s, E. J. and Plan, Y.} (2011).
Global testing under sparse alternatives: ANOVA, multiple
comparisons and the higher criticism. {\it Ann. Statist.}
\textbf{39} 2533-2556.












\bibitem{and03}
\par\noindent\hangindent2.3em\hangafter 1
\textsc{Anderson, T. W.} (2003). {\it An Introduction to Multivariate Statistical Analysis}, 3rd edn. New York: Wiley-Interscience.


\bibitem{bs96}
\par\noindent\hangindent2.3em\hangafter 1
\textsc{Bai, Z. D. and Saranadasa, H.} (1996). Effect of high dimension: by an example of a two
sample problem. {\it Statist. Sinica} \textbf{6} 311-329.


\bibitem{bcw12}
\par\noindent\hangindent2.3em\hangafter 1
\textsc{Belloni, A., Chernozhukov, V. and Wang, L.} (2012).
Square-root lasso: Pivotal recovery of sparse signals via conic
programming. {\it Biometrika} \textbf{98} 791-806.


\bibitem{bl08a}
\par\noindent\hangindent2.3em\hangafter 1
\textsc{Bickel, P. and Levina, E.} (2008a). Regularized estimation of large covariance
matrices. {\it
Ann. Statist.} \textbf{36} 199-227.


\bibitem{bl08b}
\par\noindent\hangindent2.3em\hangafter 1
\textsc{Bickel, P. and Levina, E.} (2008b). Covariance regularization by thresholding. {\it Ann. Statist.} \textbf{36} 2577-2604.


\bibitem{cl11}
\par\noindent\hangindent2.3em\hangafter 1
\textsc{Cai, T. and Liu, W. D.} (2011). Adaptive thresholding for
sparse covariance matrix estimation. {\it J. Am. Statist. Ass.}
\textbf{106} 672-684.


\bibitem{cll11}
\par\noindent\hangindent2.3em\hangafter 1
\textsc{Cai, T., Liu, W. D. and Luo, X.} (2011). A constrained $l_1$
minimization approach to sparse precision matrix estimation. {\it J.
Am. Statist. Ass.} \textbf{106} 594-607.


\bibitem{clx14}
\par\noindent\hangindent2.3em\hangafter 1
\textsc{Cai, T. T., Liu, W. D. and Xia, Y.} (2014). Two-sample test
of high dimensional means under dependence. {\it J. R. Stat. Soc.
Ser. B Stat. Methodol.} \textbf{76} 349-372.


\bibitem{clz14}
\par\noindent\hangindent2.3em\hangafter 1
\textsc{Chen, S. X., Li, J. and Zhong, P.-S.} (2014). Two-sample
tests for high dimensional means with thresholding and data
transformation. arXiv:1410.2848.


\bibitem{cq2010}
\par\noindent\hangindent2.3em\hangafter 1
\textsc{Chen, S. X. and Qin, Y.-L.} (2010). A two sample test for
high dimensional data with applications to gene-set testing. {\it
Ann. Statist.} \textbf{38} 808-835.


\bibitem{cck15}
\par\noindent\hangindent2.3em\hangafter 1
\textsc{Chernozhukov, V., Chetverikov, D. and Kato, K.} (2015).
Central limit theorems and bootstrap in high dimensions.
arXiv:1412.3661.


\bibitem{dj04}
\par\noindent\hangindent2.3em\hangafter 1
\textsc{Donoho, D. and Jin, J.} (2004). Higher criticism for
detecting sparse heterogeneous mixtures. {\it Ann. Statist.}
\textbf{32} 962-994.


\bibitem{dj94}
\par\noindent\hangindent2.3em\hangafter 1
\textsc{Donoho, D. and Johnstone, I.} (1994). Ideal spatial adaptation by wavelet shrinkage.
{\it Biometrika} \textbf{81} 425-455.


\bibitem{f96}
\par\noindent\hangindent2.3em\hangafter 1
\textsc{Fan, J.} (1996). Test of significance based on wavelet thresholding and Neyman's
truncation. {\it J. Am. Statist. Ass.} \textbf{91} 674-688.


\bibitem{fsz15}
\par\noindent\hangindent2.3em\hangafter 1
\textsc{Fan, J., Shao, Q. and Zhou, W.} (2015). Are discoveries spurious? Distributions of maximum spurious correlations and their applications. arXiv:1502.04237.



\bibitem{fht08}
\par\noindent\hangindent2.3em\hangafter 1
\textsc{Friedman, J., Hastie, T. and Tibshirani, R.} (2008). Sparse
inverse covariance estimation with the graphical lasso. {\it
Biostatistics} \textbf{9} 432-441.



\bibitem{ggbl15}
\par\noindent\hangindent2.3em\hangafter 1
\textsc{Gregory, K. B., Carroll, R. J., Baladandayuthapani, V. and Lahiri, S. N.} (2015). A two-sample
test for equality of means in high dimension. {\it J. Am. Statist. Ass.} \textbf{110} 837-849.


\bibitem{hj10}
\par\noindent\hangindent2.3em\hangafter 1
Hall, P. and Jin, J. (2010). Innovated higher criticism for
detecting sparse signals in correlated noise. {\it Ann. Statist.}
\textbf{38} 1686-1732.


\bibitem{jcs01}
\par\noindent\hangindent2.3em\hangafter 1
\textsc{James, D., Clymer, B. D. and Schmalbrock, P.} (2001). Texture detection of simulated microcalcification susceptibility
effects in magnetic resonance imaging of breasts. {\it J. Magn. Reson. Imgng} \textbf{13} 876-881.




\bibitem{jzz}
\par\noindent\hangindent2.3em\hangafter 1
\textsc{Jin, J., Zhang, C.-H. and Zhang, Q.} (2014). Optimality of graphlet screening in high
dimensional variable selection. {\it J. Mach. Learn. Res.} \textbf{15} 2723-2772.


\bibitem{kjf14}
\par\noindent\hangindent2.3em\hangafter 1
\textsc{Ke, Z., Jin, J. and Fan, J.} (2014). Covariance assisted screening and estimation. {\it Ann.
Statist.} \textbf{42} 2202-2242.


\bibitem{lz15}
\par\noindent\hangindent2.3em\hangafter 1
\textsc{Li, J. and Zhong, P. S.} (2015). A rate optimal procedure for recovering sparse differences between high-dimensional means under dependence.
{\it Ann. Statist.} To appear.


\bibitem{lw12}
\par\noindent\hangindent2.3em\hangafter 1
\textsc{Liu, H. and Wang, L.} (2012). Tiger: A tuning-insensitive
approach for optimally estimating gaussian graphical models.
arXiv:1209.2437.


\bibitem{mb06}
\par\noindent\hangindent2.3em\hangafter 1
\textsc{Meinshausen, N. and B\"{u}hlmann, P.} (2006).
High-dimensional graphs and variable selection with the Lasso. {\it
Ann. Statist.} \textbf{34} 1436-1462.


\bibitem{n03}
\par\noindent\hangindent2.3em\hangafter 1
\textsc{Nazarov, F.} (2003). On the maximal perimeter of a convex set in $\mathbb{R}^n$
with respect to a Gaussian measure. In: {\it Geometric Aspects of Functional Analysis}, Lecture Notes in Mathematics Volume 1807, Springer,
pp. 169-187.


\bibitem{rock}
\par\noindent\hangindent2.3em\hangafter 1
Rockafellar, T. R. (1970). {\it Convex Analysis}. Princeton Univ.
Press.

\bibitem{s09}
\par\noindent\hangindent2.3em\hangafter 1
\textsc{Srivastava, M.} (2009). A test for the mean vector with
fewer observations than the dimension under non-normality. {\it J.
Multiv. Anal.} \textbf{100} 518-532.


\bibitem{sd08}
\par\noindent\hangindent2.3em\hangafter 1
\textsc{Srivastava, M. and Du, M.} (2008). A test for the mean vector with fewer observations than the dimension. {\it J.
Multiv. Anal.} \textbf{99} 386-402.


\bibitem{sz13}
\par\noindent\hangindent2.3em\hangafter 1
\textsc{Sun, T. and Zhang, C. H.} (2013). Sparse matrix inversion with scaled lasso. {\it J. Machine Learning Research} \textbf{14} 3385-3418.


\bibitem{v12}
\par\noindent\hangindent2.3em\hangafter 1
\textsc{Vershynin, R.} (2012). Introduction to the non-asymptotic analysis of random matrices. In Y. C. Eldar
and G. Kutyniok, editors, {\it Compressed Sensing: Theory and Applications}. Cambridge University Press, 2012.



\bibitem{y10}
\par\noindent\hangindent2.3em\hangafter 1
\textsc{Yuan, M.} (2010). High dimensional inverse covariance matrix
estimation via linear programming. {\it J. Mach. Learn. Res.}
\textbf{11} 2261-2286.


\bibitem{zzw}
\par\noindent\hangindent2.3em\hangafter 1
\textsc{Zhang, G., Zhang, S. and Wang, Y.} (2000). Application of adaptive time-frequency decomposition in ultrasonic
NDE of highly-scattering materials. {\it Ultrasonics} \textbf{38} 961-964.


\bibitem{zcx}
\par\noindent\hangindent2.3em\hangafter 1
\textsc{Zhong, P., Chen, S. X. and Xu M.} (2013). Tests alternative to higher criticism
for high dimensional means under sparsity and column-wise dependence. {\it Ann. Statist.} \textbf{41}
2820-2851.



\end{thebibliography}
\end{document}